\documentclass[11pt,letterpaper,fleqn]{article}
\usepackage[letterpaper,textwidth=6.7in,textheight=9.5in,top=0.65in,
            hcentering,footskip=0.35in]{geometry}
\usepackage{amsmath,amssymb,amsthm,mathtools,bm,mathrsfs}
\usepackage{booktabs,array,longtable,enumitem,graphicx,tikz}
\usetikzlibrary{arrows.meta,positioning,calc}
\usepackage[font=small,labelfont=bf]{caption}
\usepackage[numbers,sort&compress]{natbib}
\usepackage[hidelinks]{hyperref}
\usepackage{bookmark}
\numberwithin{equation}{section}
\allowdisplaybreaks[2]
\newtheoremstyle{authorplain}{7pt plus2pt minus2pt}{7pt plus2pt minus2pt}
  {\itshape}{0pt}{\bfseries}{.}{.5em}{}
\newtheoremstyle{authordefinition}{7pt plus2pt minus2pt}{7pt plus2pt minus2pt}
  {\normalfont}{0pt}{\bfseries}{.}{.5em}{}
\newtheoremstyle{authorremark}{7pt plus2pt minus2pt}{7pt plus2pt minus2pt}
  {\normalfont}{0pt}{\itshape}{.}{.5em}{}
\theoremstyle{authorplain}
\newtheorem{theorem}{Theorem}[section]
\newtheorem{lemma}[theorem]{Lemma}
\newtheorem{proposition}[theorem]{Proposition}
\newtheorem{corollary}[theorem]{Corollary}
\theoremstyle{authordefinition}

\theoremstyle{authorremark}

\newenvironment{bmhead}[1]
  {\par\addvspace{8pt}\noindent\textbf{#1.}\enspace\ignorespaces}
  {\par}
\newcommand{\mathsfbi}[1]{\bm{\mathsf{#1}}}
\newcommand{\Div}{\operatorname{Div}}
\newcommand{\Curl}{\operatorname{Curl}}
\newcommand{\curl}{\operatorname{curl}}
\newcommand{\tr}{\operatorname{tr}}

\newcommand{\ord}{\operatorname{ord}}

\newcommand{\diag}{\operatorname{diag}}

\newcommand{\cE}{\mathcal E}

\newcommand{\cR}{\mathcal R}

\newcommand{\cN}{\mathcal N}
\newcommand{\cF}{\mathcal F}

\newcommand{\R}{\mathbb R}
\newcommand{\C}{\mathbb C}

\newcommand{\EulerOp}{\mathbf E}

\newcommand{\dd}{\,\mathrm d}

\newcommand{\Tr}{\mathrm{T}}

\newcommand{\uu}{{\boldsymbol{u}}}
\newcommand{\xx}{{\boldsymbol{x}}}
\newcommand{\ww}{{\boldsymbol{\omega}}}
\newcommand{\xxi}{{\boldsymbol{\xi}}}
\newcommand{\LL}{{\boldsymbol{\Lambda}}}
\newcommand{\matA}{{\mathsfbi{A}}}
\newcommand{\matP}{{\mathsfbi{P}}}
\newcommand{\matS}{{\mathsfbi{S}}}

\hypersetup{pdftitle={Local conservation laws, symmetries and Hamiltonian structure of three-dimensional incompressible Euler flow},pdfauthor={Alexey Shevyakov}}
\title{Local conservation laws, symmetries and Hamiltonian structure of three-dimensional incompressible Euler flow}
\author{Alexey Shevyakov\thanks{Email: \href{mailto:shevyakov@math.usask.ca}{shevyakov@math.usask.ca}.}\\[3pt]
  {\small\itshape Department of Mathematics and Statistics, University of Saskatchewan,}\\
  {\small\itshape Saskatoon, SK S7N 5E6, Canada}}
\date{}
\begin{document}
\maketitle
\begin{abstract}
All smooth local conservation laws, symmetries and adjoint symmetries
of three-dimensional homogeneous incompressible Euler flow are
classified in the primitive velocity--pressure variables. Explicit
dependence on time and position and arbitrary finite orders of mixed
time and spatial derivatives are allowed. Up to trivial currents, the
conservation laws consist of incompressibility, linear and angular
momentum, kinetic energy, helicity and generalised time-dependent
momentum. Every local adjoint symmetry is represented by a multiplier
of these currents. Every finite-order local symmetry is equivalent
on solutions to a classical point symmetry, resolving a conjecture of
Olver. An explicit local first-order Hamiltonian structure of the
primitive system is also established in the geometric theory of
partial differential equations. It maps the energy multiplier to
time translation; its kernel on local adjoint symmetries is spanned
by the helicity and incompressibility multipliers. The proof first
excludes higher-order symmetries by comparing their responses to
transverse and harmonic characteristic directions in compatible
derivative data. The Hamiltonian map and its kernel then determine
all local adjoint symmetries from the symmetry classification.
Exact multiplier identities and the correspondence between
multipliers and current classes complete the conservation-law
classification.
\end{abstract}
\medskip

\section{Introduction}\label{sec:introduction}

The unforced motion of a homogeneous incompressible inviscid fluid in three
spatial dimensions is governed by the Euler equations
\begin{equation}\label{eq:s1-1}
 \boldsymbol F=\uu_t+(\uu\cdot\nabla)\uu+\nabla p=0,
 \qquad G=\nabla\cdot\uu=0.
\end{equation}
Here $\uu$ is the velocity, $p$ is the pressure divided by the constant
density, and $\xx$ and $t$ are Cartesian position and time.
The local conservation laws and symmetries of \eqref{eq:s1-1} are
classified below with arbitrary smooth dependence on finitely many
derivatives of $\uu,p$. Complete classifications of this kind have
also been obtained for other fundamental field equations.
\citet{AndersonTorre1996} classified all local generalised symmetries
of the four-dimensional vacuum Einstein equations, allowing
derivatives of arbitrary finite order. For Euler, the corresponding
symmetry problem was posed by \citet{Olver1982}. The present work
resolves that problem together with the completeness problem for
local conservation laws in the primitive variables.

The vorticity, kinetic energy per unit mass and helicity density are
\begin{equation}\label{eq:intro-densities}
 \ww=\nabla\times\uu,\qquad K=\tfrac12|\uu|^2,
 \qquad h=\uu\cdot\ww.
\end{equation}
A local conservation law has a density $\rho$ and flux $\boldsymbol J$
depending on time, position, the fields and finitely many derivatives
of the fields, such that
\begin{equation}\label{eq:current-divergence}
 \partial_t\rho+\nabla\cdot\boldsymbol J=0
 \qquad\text{on solutions of \eqref{eq:s1-1}}.
\end{equation}
Each local balance \eqref{eq:current-divergence} yields a conserved volume integral when its boundary
flux vanishes. For a vector density, the flux is a matrix and its divergence is taken
row by row. The tensor product has components
$(\boldsymbol a\otimes\boldsymbol b)_{ij}=a_i b_j$; $\mathsfbi I$ denotes
the identity matrix. The classical conserved densities and currents are as follows
\citep{Batchelor2000,CheviakovOberlack2014}.

\emph{Incompressibility.} The equation $G=0$ is itself a conservation
law with zero density; multiplication by an arbitrary smooth function
$c_0$ of time gives the family
\begin{equation}\label{eq:intro-mass}
 \rho=0,\qquad \boldsymbol J=c_0(t)\uu.
\end{equation}
For constant density the mass balance reduces to this family, and the
choice $c_0=1$ recovers the second Euler equation \eqref{eq:s1-1}.

\emph{Linear momentum.} Its three components have the vector density and
matrix flux
\begin{equation}\label{eq:intro-momentum}
 \boldsymbol\rho=\uu,\qquad
 \mathsfbi J=\uu\otimes\uu+p\mathsfbi I.
\end{equation}

\emph{Angular momentum.} Its three components have the vector density
and matrix flux
\begin{equation}\label{eq:intro-angular}
 \boldsymbol\rho=\xx\times\uu,\qquad
 \mathsfbi J=\mathsfbi C_{\xx}(\uu\otimes\uu+p\mathsfbi I), \qquad \mathsfbi C_{\xx}\boldsymbol a \equiv \xx\times\boldsymbol a,
\end{equation}
where $\mathsfbi C_{\xx}$ is the cross-product matrix of $\xx$, acting
on an arbitrary vector $\boldsymbol a$ as indicated.

\emph{Kinetic energy.} The density $K$ in \eqref{eq:intro-densities} corresponds to
\begin{equation}\label{eq:intro-energy}
 \rho=K,\qquad \boldsymbol J=(K+p)\uu.
\end{equation}

\emph{Helicity.} The density $h$ in \eqref{eq:intro-densities} yields
\begin{equation}\label{eq:intro-helicity}
 \rho=h,\qquad
 \boldsymbol J=h\uu+(p-K)\ww.
\end{equation}
Helicity conservation was established by \citet{Moreau1961}.
For linked vortex tubes, its integral relates their linkage to their
circulations \citep{Moffatt1969}.

\emph{Generalised linear momentum.} An arbitrary smooth vector function
$\boldsymbol b$ of time gives
\begin{equation}\label{eq:intro-general-momentum}
 \rho=\boldsymbol b(t)\cdot\uu,\qquad
 \boldsymbol J=(\boldsymbol b\cdot\uu)\uu+p\boldsymbol b
               -(\boldsymbol b'\cdot\xx)\uu.
\end{equation}
Constant $\boldsymbol b$ yields the components of
\eqref{eq:intro-momentum}, while $\boldsymbol b$ linear in time corresponds to the
time-dependent momentum current associated with Galilean boosts.

Energy and helicity constrain spectral transfer in three-dimensional
turbulence and enter the statistical equilibria of Euler systems
truncated to finitely many Fourier modes
\citep{Kraichnan1973,AlexakisBiferale2018}.
In the Eulerian Hamiltonian formulation of periodic flow, helicity is a
Casimir: its Poisson bracket with every functional vanishes \citep{Salmon1988}.

The question addressed in this paper is whether the list
\eqref{eq:intro-mass}--\eqref{eq:intro-general-momentum} is complete.
Every local conservation law admits, up to equivalence, a multiplier
representation: a linear combination of the Euler equations that is
identically a divergence \citep{MartinezAlonso1979,BlumanCheviakovAnco2010}.
The coefficients, called \emph{multipliers}, may in principle depend on $\xx$, $t$, velocity,
pressure, and their derivatives of arbitrary finite order. Two
density-flux pairs $(\rho_1, {\boldsymbol J}_1)$ and
$(\rho_2, {\boldsymbol J}_2)$ of \eqref{eq:current-divergence} are
equivalent when their difference is the sum of a current vanishing on
solutions and one with identically zero divergence. The
\emph{completeness problem} is to determine whether multipliers
containing higher derivatives yield conservation-law classes beyond
those listed above. By Theorem~\ref{thm:main}, the list
\eqref{eq:intro-mass}--\eqref{eq:intro-general-momentum} is exhaustive
at every finite differential order.

The second question concerns symmetries. The classical point
symmetries act on time, position, velocity and pressure. Their
infinitesimal generators are
\begin{equation}\label{eq:intro-point-generator}
 \mathcal V=X^0\partial_t+\boldsymbol X\cdot\partial_{\xx}
              +\boldsymbol U\cdot\partial_{\uu}+U^0\partial_p,
\end{equation}
with coefficients depending only on these variables.
The classical family \citep{Buchnev1971,Olver1982} is
\begin{equation}\label{eq:sym-family}
 \begin{aligned}
 X^0&=\psi_1t+\psi_0,\\
 \boldsymbol X&=\varsigma\xx+\boldsymbol\Omega\times\xx+\boldsymbol b(t),\\
 \boldsymbol U&=(\varsigma-\psi_1)\uu
                    +\boldsymbol\Omega\times\uu+\boldsymbol b'(t),\\
 U^0&=2(\varsigma-\psi_1)p-\boldsymbol b''(t)\cdot\xx+\varpi(t).
 \end{aligned}
\end{equation}
Here $\psi_0,\psi_1,\varsigma$ and $\boldsymbol\Omega$
are constant, while $\boldsymbol b$ and $\varpi$ are arbitrary smooth
vector and scalar functions of time. Primes denote time derivatives.
The named generators are
\begin{equation}\label{eq:intro-point-generators}
 \begin{aligned}
 \text{time translation:}\quad&\partial_t,\\
 \text{rotations:}\quad&
       (\boldsymbol\Omega\times\xx)\cdot\partial_{\xx}
           +(\boldsymbol\Omega\times\uu)\cdot\partial_{\uu},\\
 \text{time dilation:}\quad&
       t\partial_t-\uu\cdot\partial_{\uu}-2p\partial_p,\\
 \text{space dilation:}\quad&
       \xx\cdot\partial_{\xx}+\uu\cdot\partial_{\uu}+2p\partial_p,\\
 \text{moving frames:}\quad&
       \boldsymbol b\cdot\partial_{\xx}+\boldsymbol b'\cdot\partial_{\uu}
           -(\boldsymbol b''\cdot\xx)\partial_p,\\
 \text{pressure gauge:}\quad&\varpi(t)\partial_p.
 \end{aligned}
\end{equation}
The two dilations are independent because \eqref{eq:s1-1} contains no
viscous term; the Navier--Stokes equations admit only the combination
that preserves the viscosity. The moving-frame family consists
of time-dependent spatial translations with the indicated velocity
and pressure corrections.
Constant $\boldsymbol b$ gives ordinary translations; taking
$\boldsymbol b(t)=t\boldsymbol v$ with constant $\boldsymbol v$ gives
Galilean boosts. The pressure gauge adds a spatially uniform function
of time to pressure.

At fixed $(\xx,t)$, the infinitesimal field variation, called
the \emph{evolutionary characteristic} of the point symmetry \eqref{eq:intro-point-generator}, is given by
\begin{equation}\label{eq:sym-point-characteristic}
 \binom{\boldsymbol r}{r^0}=
 \begin{pmatrix}
 \boldsymbol U-X^0\uu_t-(\boldsymbol X\cdot\nabla)\uu\\
 U^0-X^0p_t-\boldsymbol X\cdot\nabla p
 \end{pmatrix}.
\end{equation}
The differential order of such an evolutionary characteristic is at
most one. A more general \emph{local symmetry} characteristic may
depend on derivatives of the fields of any finite order, and
satisfies the linearised Euler equations on solutions. Two
evolutionary characteristics $({\boldsymbol r}_1,r^0_1)$ and
$({\boldsymbol r}_2,r^0_2)$ are \emph{equivalent} if their
difference vanishes on solutions of \eqref{eq:s1-1}. By
Theorem~\ref{thm:symmetries}, every such smooth
characteristic coincides, up to equivalence, with a member of
\eqref{eq:sym-point-characteristic}.

Both results concern the equations \eqref{eq:s1-1} in the primitive
variables $\uu,p$, and the following remarks delimit them from
related statements in other formulations. In two dimensions, vortex
stretching is absent and the scalar vorticity
$\zeta=\partial_1u_2-\partial_2u_1$ is constant along fluid
trajectories, so that every smooth function $\Phi(\zeta)$ yields a
conservation law
\begin{equation}\label{eq:intro-enstrophy}
 \partial_t\Phi(\zeta)+\nabla\cdot[\uu\Phi(\zeta)]=0.
\end{equation}
The associated integrals are the generalised enstrophies, which with
energy govern the dual cascade of plane turbulence
\citep{Kraichnan1967,AlexakisBiferale2018}. Symmetry-reduced equations
may possess local conservation laws that do not extend to the full
system, as in the axisymmetric and helically symmetric reductions
\citep{KelbinCheviakovOberlack2013} and on closed surfaces
\citep{IzosimovKhesin2017}, and reduction can change the invariance
properties of inherited conservation laws \citep{DruzhkovCheviakov2026}.

Kelvin circulation is an integral around a moving fluid loop, and the
Cauchy invariants relate vorticity to derivatives of the Lagrangian
flow map \citep{Batchelor2000,BesseFrisch2017}; the relabelling
symmetry of the labels and its relation to material invariants are
described by \citet{Salmon1988}, and the derivation of momentum,
angular momentum and energy by Noether's theorem by \citet{Olver1993}.
These statements require material or variational formulations; the
present study concerns volume currents in the primitive Eulerian
variables $\uu,p$.

For Euler flow augmented by a transported scalar $\theta$, Ertel's
theorem gives the potential vorticity $q$ and its material conservation:
\begin{equation}\label{eq:intro-ertel}
 \begin{gathered}
 Y=\partial_t+\uu\cdot\nabla,\qquad Y\theta=0,
 \qquad q=\ww\cdot\nabla\theta,\qquad Yq=0,\\
 \partial_t\Phi(\theta,q)+\nabla\cdot[\uu\Phi(\theta,q)]=0.
 \end{gathered}
\end{equation}
This is the homogeneous incompressible form of Ertel's
theorem~\citep{Ertel1942}. A transported scalar, a particle label or an
auxiliary potential is determined by the flow history and is therefore
outside the class of differential functions of $\uu,p$; the augmented
systems admit further local or nonlocal invariants
\citep{CheviakovOberlack2014,AncoWebb2020}.

The vorticity dynamics equation is given by the curl of the momentum equation:
\begin{equation}\label{eq:intro-vorticity}
 \partial_t\ww+\nabla\cdot(\ww\otimes\uu-\uu\otimes\ww)
 =\partial_t\ww+\nabla\times(\ww\times\uu)=0.
\end{equation}
The vorticity current
\eqref{eq:intro-vorticity} is trivial as a local conservation law of
the Euler equations in primitive variables: with
$\Theta_{ij}=\epsilon_{ijk}u_k$ one has $\ww=\nabla\cdot\Theta$
row by row, and the flux $\ww\otimes\uu-\uu\otimes\ww$ differs from
$-\partial_t\Theta$ by a matrix whose row divergence vanishes
identically plus the matrix $\epsilon_{ijk}F_k$, which vanishes on
solutions; the volume integral of the vorticity is accordingly a
boundary quantity.

For three-dimensional Euler flows \eqref{eq:s1-1}, \citet{Serre1984}
classified conserved integrals whose densities depend on the velocity
and its first spatial derivatives, with no explicit dependence on
position and time, and obtained linear momentum, kinetic energy and
helicity. \citet{CheviakovOberlack2014} classified primitive-variable
multipliers through second order and left the higher-order Euler
problem open. Complete local classifications were obtained by
\citet{GusyatnikovaYumaguzhin1989} for the symmetries and conservation
laws of the three-dimensional incompressible Navier--Stokes equations,
whose conservation laws are generated by the generalised momentum,
angular momentum and incompressibility currents, and by
\citet{KhorkovaVerbovetsky1995} for the conservation laws of the
$k$--$\varepsilon$ turbulence model, whose classical symmetry algebra
is isomorphic to that of the Navier--Stokes equations. Both proofs
lower the derivative order of the generating functions by induction,
comparing the coefficients of the highest derivatives, which enter
through the second-order viscous terms. The Euler equations
\eqref{eq:s1-1} are of first order, and their principal symbol is
singular on two families of covectors; the order reduction below is
built on this symbol.

\citet[\S2, p.~240]{Olver1982} conjectured that every local
generalised symmetry of Euler is equivalent to the evolutionary
form of a point symmetry. \citet{BihloPopovych2020}
reported this question as still open in 2020.
For the three-dimensional incompressible Navier--Stokes equations,
\citet{GusyatnikovaYumaguzhin1989} proved the corresponding statement:
every local symmetry is generated by the classical point symmetries.
Theorem~\ref{thm:symmetries} gives the Euler classification in the
primitive variables.

The classification allows explicit dependence on time and position,
on velocity and pressure, and smooth, possibly non-polynomial,
dependence on derivatives of arbitrary finite order, including mixed
time and spatial derivatives. The first result,
Theorem~\ref{thm:main}, states that every such local conservation
current is equivalent to a linear combination of
\eqref{eq:intro-mass}--\eqref{eq:intro-general-momentum}, with the
time-dependent freedom specified there. It also classifies all local
\emph{adjoint symmetries}: pairs satisfying the formal adjoint of
the linearised Euler equations on solutions, where the formal adjoint
is obtained by integration by parts in time and space. Every local
adjoint symmetry agrees on solutions with a multiplier of the
classified currents.

The second result, Theorem~\ref{thm:symmetries}, states that every
smooth finite-order local symmetry agrees on solutions with the
evolutionary characteristic \eqref{eq:sym-point-characteristic} of
a point symmetry \eqref{eq:sym-family}. This proves the conjecture
of \citet{Olver1982} for the three-dimensional Euler equations in
the primitive variables.

The third result is a local Hamiltonian structure relating the two
classifications. Write a four-component input as
$Q=(\LL,\Lambda^0)$, where the vector $\LL$ and scalar $\Lambda^0$
are paired with the momentum and incompressibility equations,
respectively. Total derivatives $D_t$ and $D_i$ apply the chain rule
to functions of the fields and their derivatives. The spatial
gradient is formed from $D_1,D_2,D_3$. The material derivative and
the operator are
\begin{equation}\label{eq:noether-map}
 Y=D_t+\uu\cdot\nabla,\qquad
 \mathscr NQ
 =\binom{-\ww\times\LL-\nabla\Lambda^0}{Y\Lambda^0}.
\end{equation}
The operator maps every local adjoint symmetry to a symmetry
characteristic. It also induces a Poisson bracket on conservation-law
classes, making it a Hamiltonian structure in the general-PDE sense
of \citet{KerstenKrasilshchikVerbovetskyVitolo2009}
(Theorem~\ref{thm:hamiltonian}). The energy current
\eqref{eq:intro-energy} gives the example
\begin{equation}\label{eq:intro-hamiltonian-energy}
 \mathscr N(\uu,p+K)=\binom{\uu_t}{p_t}
 \qquad\text{on solutions of \eqref{eq:s1-1}}.
\end{equation}
Thus kinetic energy generates time translation. On local adjoint
symmetries of arbitrary finite order, the kernel of $\mathscr N$ is
the constant span of the helicity and incompressibility multipliers;
their conservation-law classes have zero Poisson bracket with every
class and are called \emph{Casimirs}. On solutions, its image consists of
the point characteristics \eqref{eq:sym-point-characteristic} with
$\varsigma=\psi_1=0$ in \eqref{eq:sym-family}. These
properties determine the adjoint classification from the symmetry
classification.

The classical Eulerian velocity bracket contains the same vorticity
cross product as \eqref{eq:noether-map}
\citep{Arnold1966,Salmon1988}. Taking the curl of the velocity output
when the vector input is a curl recovers the vorticity Hamiltonian
operator of \citet[Theorem~3.2]{Olver1982}. A related transformation
from adjoint symmetries to symmetries of a three-dimensional
vector-potential formulation was obtained by
\citet[\S6, Theorem~6, equation~(14)]{Morozov2026EulerHelmholtz}.
After a spatial vector potential for $\LL$ is introduced, Morozov's
time coordinate is replaced by its negative and the output is
multiplied by $-1$. The spatial curl of the resulting output
agrees with the velocity component of \eqref{eq:noether-map} on
adjoint symmetries.
The Hamiltonian property of the full operator, including its pressure
output, is verified here directly for the primitive system
\eqref{eq:s1-1}. For the two-dimensional vorticity equation written
for the stream function, \citet{KrasilshchikMorozov2025} obtained a
local Hamiltonian structure in the same general-PDE sense, whose
operator is the material derivative.

Local conservation laws, their equivalence and the direct construction
method are described by
\citet{Olver1993,AncoBluman2002b,BlumanCheviakovAnco2010}; their
connections with other analytical properties of nonlinear partial
differential equations are developed by \citet{CheviakovZhao2024}.
The present classification complements the geometric uniqueness
results for helicity of \citet{EncisoPeraltaSalasTorres2016}.
Conservation at weak regularity is a separate question
\citep{ConstantinETiti1994,Isett2018,DeRosa2020,WangWeiWuYe2024}.

Section~\ref{sec:1} states the three principal results and gives
explicit multiplier identities for the classical currents.
Section~\ref{sec:2} describes compatible derivative data and the two
characteristic families. Their response equations are derived in
Section~\ref{sec:4}, and the transported transverse matrices are
classified in Section~\ref{sec:5}. Comparing the two responses gives
the symmetry classification in Section~\ref{sec:sym}. The Hamiltonian
structure and its kernel are then established in Section~\ref{sec:8},
where they are used to classify all local adjoint symmetries.
The multiplier--current correspondence in Section~\ref{sec:9}
completes the conservation-law classification. The induced bracket
on these currents is given in Appendix~\ref{app:poisson}.

\section{Local conservation laws, symmetries and adjoint symmetries}\label{sec:1}

\subsection{Differential functions and identities on solutions}\label{sec:differential-functions}

A \emph{differential function} depends smoothly on the independent
variables, the fields, and finitely many field derivatives. Its
\emph{joint order}, denoted by $\ord f$ for a function $f$, is the largest
total number of time and spatial differentiations on any field argument.
An identity between differential functions holds \emph{on solutions}
when it holds after the Euler equations \eqref{eq:s1-1} and their
differentiated consequences are imposed. All statements are local.
The classifications below hold on sufficiently small coordinate
neighbourhoods in the space of compatible derivative data. On each
such neighbourhood the constants are fixed and the arbitrary
functions of time are defined on the corresponding time interval.
The total derivatives $D_t,D_i$ apply the chain rule to differential
functions; on fields and currents they agree with the derivatives in
Section~\ref{sec:introduction}. They act componentwise on vector and
matrix differential functions. For a scalar differential function $f$
and a vector differential function $\boldsymbol a$, the spatial operators are
\[
 \nabla f=(D_1f,D_2f,D_3f)^{\Tr},\qquad
 \Div\boldsymbol a=\sum_{i=1}^3D_i a_i,\qquad
 \Curl\boldsymbol a=
 \begin{pmatrix}
 D_2a_3-D_3a_2\\
 D_3a_1-D_1a_3\\
 D_1a_2-D_2a_1
 \end{pmatrix}.
\]
Spatial indices range from one to three, and summation is assumed in repeated indices.
Vectors are columns, and the superscript $\Tr$ denotes transpose.
The material derivative, the velocity gradient and the cross-product
matrix are
\begin{equation}\label{eq:s1-2}
 \begin{gathered}
 Y=D_t+u_iD_i,\qquad A_{ij}=D_ju_i,\qquad
 \mathsfbi C_{\xxi}\boldsymbol v=\xxi\times\boldsymbol v,\\
 \matA-\matA^{\Tr}=\mathsfbi C_{\ww},\qquad \tr\matA=G.
 \end{gathered}
\end{equation}

\subsection{Conservation laws, multipliers and adjoint symmetries}

Two conservation currents are \emph{equivalent} if their difference is
the sum of a current vanishing on solutions and a current whose divergence
vanishes identically for unrestricted fields.
A current is \emph{local} when its components are differential functions,
and \emph{trivial} when it is equivalent to zero; an equivalence class
of currents is a \emph{conservation-law class}.
A \emph{multiplier} is a pair of differential functions satisfying an
unrestricted identity of the form
\begin{equation}\label{eq:s1-4}
 E=\binom{\boldsymbol F}{G},\qquad
 Q=\binom{\LL}{\Lambda^0},\qquad
 \LL\cdot\boldsymbol F+\Lambda^0G=D_t\rho+\Div\boldsymbol J.
\end{equation}
Here $\LL$ is a three-component column and $\Lambda^0$ is scalar.
Two multipliers are identified when they agree on solutions.
Applying the variational derivative with respect to the fields to
\eqref{eq:s1-4} annihilates the divergence on the right; the left
side becomes the formal adjoint of the linearised Euler operator
applied to $Q$, plus terms proportional to $E$ and its derivatives
\citep{AncoBluman2002b,BlumanCheviakovAnco2010}. Restricting to
solutions gives the \emph{adjoint equations}
\begin{equation}\label{eq:s3-adjoint-equations}
 Y\LL-\matA^{\Tr}\LL+\nabla\Lambda^0=0,
 \qquad \Div\LL=0\qquad\text{on solutions};
\end{equation}
the derivation is given in Section~\ref{sec:9},
Eq.~\eqref{eq:s9-multiplier-Euler-identity}. The adjointness is the
identity, valid for arbitrary fields and arbitrary pairs
$(\boldsymbol v,\pi)$ and $(\LL,\Lambda^0)$,
\[
 \begin{aligned}
 &\LL\cdot\bigl(Y\boldsymbol v+\matA\boldsymbol v+\nabla\pi\bigr)
 +\Lambda^0\Div\boldsymbol v
 +\boldsymbol v\cdot\bigl(Y\LL-\matA^{\Tr}\LL+\nabla\Lambda^0\bigr)
 +\pi\Div\LL\\
 &\qquad=D_t(\LL\cdot\boldsymbol v)
 +\Div\bigl[(\LL\cdot\boldsymbol v)\uu+\pi\LL+\Lambda^0\boldsymbol v\bigr]
 -G\,(\LL\cdot\boldsymbol v),
 \end{aligned}
\]
in which the first bracket is the linearised Euler operator of
\eqref{eq:sym-determining} below.
A \emph{local adjoint symmetry} is a pair of differential functions
satisfying \eqref{eq:s3-adjoint-equations}. Pairs agreeing on
solutions represent the same adjoint symmetry. Every multiplier
restricts to an adjoint symmetry; the converse generally requires
additional integrability conditions.

The classical conservation laws of the Introduction have the
following multipliers. Let $c,\nu$ and $\boldsymbol\Omega$ be
constants, and let $\boldsymbol b(t),c_0(t)$ be arbitrary smooth
functions. Define
\begin{equation}\label{eq:s1-5}
 \begin{aligned}
 \LL&=c\ww+\nu\uu+\boldsymbol\Omega\times\xx+\boldsymbol b(t),\\
 \Lambda^0&=\nu(p+K)
 +(\boldsymbol\Omega\times\xx+\boldsymbol b(t))\cdot\uu
 -\boldsymbol b'(t)\cdot\xx+c_0(t).
 \end{aligned}
\end{equation}
The corresponding density and flux are
\begin{equation}\label{eq:s9-explicit-currents}
 \begin{gathered}
 \begin{aligned}
 \rho={}&\tfrac12c\,h+\nu K+\boldsymbol\beta\cdot\uu,\\
 \boldsymbol J={}&c\bigl[(p+K)\ww+\tfrac12\uu\times\uu_t\bigr]
       +\nu(p+K)\uu +(\boldsymbol\beta\cdot\uu)\uu+p\boldsymbol\beta
       -(\boldsymbol b'\cdot\xx)\uu+c_0\uu,
 \end{aligned}
 \end{gathered}
\end{equation}
where $\boldsymbol\beta=\boldsymbol\Omega\times\xx+\boldsymbol b(t)$. For the multipliers \eqref{eq:s1-5} and the current
\eqref{eq:s9-explicit-currents}, identity \eqref{eq:s1-4}
is obtained by combining the following product-rule identities,
which hold for arbitrary velocity and pressure fields:
\begin{equation}\label{eq:s9-current-identities}
 \begin{aligned}
 \uu_t\cdot\ww
    &=D_t(\tfrac12h)+\tfrac12\Div(\uu\times\uu_t),\\
 (\matA\uu+\nabla p)\cdot\ww
    &=\Div[(p+K)\ww],\\
 D_tK+\Div[(p+K)\uu]
    &=\uu\cdot\boldsymbol F+(p+K)G,\\
 D_t(\boldsymbol\beta\cdot\uu)
    +\Div[(\boldsymbol\beta\cdot\uu)\uu+p\boldsymbol\beta]
    &=\boldsymbol\beta\cdot\boldsymbol F
      +(\boldsymbol\beta\cdot\uu)G+\boldsymbol b'\cdot\uu,\\
 \Div[-(\boldsymbol b'\cdot\xx)\uu]
    &=-\boldsymbol b'\cdot\uu-(\boldsymbol b'\cdot\xx)G,\\
 \Div(c_0\uu)&=c_0G.
 \end{aligned}
\end{equation}
Identity \eqref{eq:s1-4} for \eqref{eq:s1-5} and
\eqref{eq:s9-explicit-currents} is $c$ times the sum of the first two
identities, which add to $c\,\ww\cdot\boldsymbol F$, plus $\nu$ times
the third, plus the sum of the fourth and fifth, in which the terms
$\pm\boldsymbol b'\cdot\uu$ cancel, plus the sixth.
The helicity flux in \eqref{eq:s9-explicit-currents} is related to
\eqref{eq:intro-helicity} by the unrestricted identity
\begin{equation}\label{eq:s9-helicity-flux-equivalence}
 \begin{aligned}
 &2\bigl[(p+K)\ww+\tfrac12\uu\times\uu_t\bigr]
       -\bigl[h\uu+(p-K)\ww\bigr]=\Curl[(p+K)\uu]+\uu\times\boldsymbol F.
 \end{aligned}
\end{equation}
Here $\boldsymbol F$ is the left-hand side of the momentum equation
in \eqref{eq:s1-1}. On solutions, the last term in \eqref{eq:s9-helicity-flux-equivalence} vanishes, and the curl has identically zero
divergence, so the terms proportional to $c$ in \eqref{eq:s9-explicit-currents}
form a current equivalent to $c/2$ times the helicity current
\eqref{eq:intro-helicity}.

The first main result states that the family \eqref{eq:s1-5}
contains every local adjoint symmetry, and the family
\eqref{eq:s9-explicit-currents} every local conservation current,
at every finite order.

\begin{theorem}[Local conservation laws and adjoint symmetries]\label{thm:main}
Every smooth local conservation current of finite joint order is
equivalent to a current \eqref{eq:s9-explicit-currents} for some
constants $c,\nu,\boldsymbol\Omega$ and smooth functions
$\boldsymbol b(t),c_0(t)$. Every smooth local adjoint symmetry of
finite joint order agrees on solutions of \eqref{eq:s1-1} with a
multiplier \eqref{eq:s1-5}.
\end{theorem}

\subsection{Local symmetries}

A \emph{local symmetry} is a pair $q=(\boldsymbol r,r^0)$ of differential
functions satisfying the linearised Euler equations on solutions,
\begin{equation}\label{eq:sym-determining}
 Y\boldsymbol r+\matA\boldsymbol r+\nabla r^0=0,\qquad
 \Div\boldsymbol r=0\qquad\text{on solutions}.
\end{equation}
It is the evolutionary form of the infinitesimal transformation
$\uu\mapsto\uu+\epsilon\boldsymbol r$, $p\mapsto p+\epsilon r^0$, which
maps solutions to solutions to first order in $\epsilon$; pairs
agreeing on solutions are equivalent. For a point symmetry generator
\eqref{eq:intro-point-generator}, $q$ is given by \eqref{eq:sym-point-characteristic}.

\begin{theorem}[Local symmetries]\label{thm:symmetries}
Every smooth local symmetry of finite joint order agrees on solutions with
the characteristic \eqref{eq:sym-point-characteristic} of a point
symmetry \eqref{eq:sym-family}. In particular, every local symmetry
has a representative of joint order at most one.
\end{theorem}

\subsection{A local Hamiltonian structure}

The Hamiltonian operator $\mathscr N$ is defined in
\eqref{eq:noether-map}. Its input is the pair $Q=(\LL,\Lambda^0)$ of
\eqref{eq:s1-4}, and its output is a velocity--pressure symmetry
characteristic when $Q$ satisfies \eqref{eq:s3-adjoint-equations}.

The notion of a Hamiltonian structure for systems of partial
differential equations with differential constraints, due to
\citet{KerstenKrasilshchikVerbovetskyVitolo2009}, is used here. Such a
structure maps multipliers to symmetry characteristics and induces a
Poisson bracket on conservation-law classes. The operator criterion for
this property is stated and verified in
Section~\ref{sec:hamiltonian-proof}; it uses the \emph{formal adjoint},
which transfers total derivatives between the factors of a scalar
product by integration by parts in time and space. Two such operators
are identified when their difference factors through $L^*$ on
solutions, and a structure is \emph{nontrivial} when it is not
identified with zero. This setting, with adjoints in time and space,
is the appropriate one for \eqref{eq:s1-1}, which includes the
incompressibility constraint and has no evolution equation for $p$.
A \emph{Casimir} is a conservation-law class whose multiplier is in
the kernel of the Hamiltonian structure; by the bracket of
Appendix~\ref{app:poisson}, its bracket with every class is zero.
Define the two multipliers
\begin{equation}\label{eq:hamiltonian-casimirs}
 Q_G=\binom{\boldsymbol0}{1},\qquad
 Q_H=\binom{\ww}{0}.
\end{equation}
They correspond to incompressibility and one half of the helicity
current, respectively, by \eqref{eq:s9-current-identities}.

\begin{theorem}[Hamiltonian structure]\label{thm:hamiltonian}
The operator \eqref{eq:noether-map} is a nontrivial local first-order
Hamiltonian structure of the Euler equations \eqref{eq:s1-1}.
On finite-order local adjoint symmetries its kernel consists of the
constant linear combinations of $Q_G$ and $Q_H$ in
\eqref{eq:hamiltonian-casimirs}. Its image on finite-order local
adjoint symmetries consists, on solutions, of the point symmetry
characteristics \eqref{eq:sym-point-characteristic} with
$\varsigma=\psi_1=0$ in \eqref{eq:sym-family}.
\end{theorem}

For example, direct substitution of the energy multiplier gives the
unrestricted identity
\begin{equation}\label{eq:hamiltonian-energy}
 \mathscr N(\uu,p+K)
 =\binom{\uu_t-\boldsymbol F}{p_t+\uu\cdot\boldsymbol F}.
\end{equation}
On solutions this is $(\uu_t,p_t)$, the characteristic
\eqref{eq:sym-point-characteristic} of time translation with $X^0=-1$.
The two multipliers \eqref{eq:hamiltonian-casimirs}
have zero image identically. The incompressibility Casimir is
represented by the current \eqref{eq:intro-mass}, whose density is zero.

\subsection{Plan of the proofs}

The symmetry classification is proved first. For a symmetry of order
$m$, the derivatives of order $m$ are changed while the lower-order
derivatives and field values at the point are held fixed. The changes
must satisfy the differentiated Euler equations. They are spanned by
the transverse and harmonic characteristic families described in
Section~\ref{sec:2}. If the symmetry is unchanged under both families,
its order is at most $m-1$.

The two responses are determined in Sections~\ref{sec:4}--\ref{sec:sym}.
The harmonic response is a polynomial in the covector with coefficients
independent of the fields. The transverse response satisfies a
transport equation whose solutions are classified in
Proposition~\ref{thm:dependence}. Comparing the responses at the complex
covectors common to both families excludes every order above one.
The remaining first-order equations give \eqref{eq:sym-family}.

For an adjoint symmetry $Q$ of arbitrary finite order, the operator
\eqref{eq:noether-map} gives a symmetry $\mathscr NQ$. Proposition~
\ref{thm:dependence} also implies that a finite-order differential
function satisfying $Yf=0$ identically on compatible Euler data is
locally constant. This consequence determines the full kernel of
$\mathscr N$ at every finite order of $Q$. Every image $\mathscr NQ$
preserves the two multipliers \eqref{eq:hamiltonian-casimirs},
which excludes both dilations from its point-symmetry representative.
Each remaining point symmetry has an explicit classical preimage;
subtracting it leaves an element of the known kernel. This proves the
adjoint classification in Section~\ref{sec:8}. The current
correspondence in Section~\ref{sec:9} then gives that the classical
currents exhaust all local conservation-law classes. The dependencies
are summarised in Figure~\ref{fig:descent}.

\begin{figure}
\centering
% Native TikZ figure; input inside the manuscript's figure environment.
% Requires \usepackage{tikz} and \usetikzlibrary{arrows.meta,calc}.
% References use the manuscript labels and update with its numbering.
\begin{tikzpicture}[
  font=\small,
  every node/.style={align=center},
  proofbox/.style={draw, line width=.45pt, inner sep=4pt,
                  text width=3.5cm, anchor=north},
  combined/.style={proofbox, text width=7.55cm},
  conclusion/.style={proofbox, text width=11.6cm, fill=black!3},
  dependency/.style={-{Latex[length=1.8mm]}, line width=.5pt}
]

\node[proofbox, minimum height=3.1cm] (symmetry) at (-4.05,0)
  {All local symmetries have\\
   point representatives\\
   Theorem~\ref{thm:symmetries}\\[3pt]
   Harmonic and transverse tests,\\
   Sections~\ref{sec:4}--\ref{sec:sym}};

\node[proofbox, minimum height=3.1cm] (hamiltonian) at (0,0)
  {Hamiltonian structure \eqref{eq:noether-map}\\
   Theorem~\ref{thm:hamiltonian}\\[3pt]
   Adjoint symmetries map\\
   to local symmetries};

\node[proofbox, minimum height=3.1cm] (transport) at (4.05,0)
  {Transported transverse\\
   matrices\\
   Proposition~\ref{thm:dependence}};

\node[combined, minimum height=1.4cm] (image) at (-2.025,-3.7)
  {Every image $\mathscr NQ$ has a point representative.\\
   Preservation of the helicity and incompressibility\\
   multipliers excludes both dilations.};

\node[proofbox, minimum height=1.4cm] (scalar) at (4.05,-3.7)
  {Local material\\
   scalar invariants\\
   are constant\\
   Corollary~\ref{cor:material-scalars}};

\node[combined, minimum height=1.4cm] (preimage) at (-2.025,-6.0)
  {Classical multipliers give preimages\\
   of the remaining point symmetries.};

\node[proofbox, minimum height=1.4cm] (kernel) at (4.05,-6.0)
  {Kernel on local\\
   adjoint symmetries\\
   Proposition~\ref{prop:hamiltonian-kernel}};

\node[conclusion, minimum height=1.2cm] (adjoint) at (0,-8.0)
  {All local adjoint symmetries\\
   from classical preimages and the kernel\\
   Proposition~\ref{prop:adjoint-classification}};

\node[conclusion, minimum height=1.2cm] (currents) at (0,-9.75)
  {All local conservation-law classes\\
   from the multiplier--current correspondence,\\
   Proposition~\ref{prop:current-correspondence}; Theorem~\ref{thm:main}};

\draw[dependency] (symmetry.south)
  -- ($(image.north)+(-2.025cm,0)$);
\draw[dependency] (hamiltonian.south)
  -- ($(image.north)+(2.025cm,0)$);
\draw[dependency] (transport.south) -- (scalar.north);
\draw[dependency] (transport.north) -- ++(0,0.45) -| (symmetry.north);
\draw[dependency] (image.south) -- (preimage.north);
\draw[dependency] (scalar.south) -- (kernel.north);
\draw[dependency] (preimage.south)
  -- ($(adjoint.north)+(-2.025cm,0)$);
\draw[dependency] (kernel.south)
  -- ($(adjoint.north)+(4.05cm,0)$);
\draw[dependency] (adjoint.south) -- (currents.north);

\end{tikzpicture}
\caption{Dependencies among the stages of the proof. An arrow runs
from a result to the statement in whose proof it is used.
Proposition~\ref{thm:dependence} enters the symmetry reduction
(Proposition~\ref{thm:sym-reduction}) directly and the kernel
(Proposition~\ref{prop:hamiltonian-kernel}) through
Corollary~\ref{cor:material-scalars}. In \eqref{eq:ham-action-representatives},
normality (Lemma~\ref{lem:normality}) ensures that the action on a
multiplier is independent of the chosen symmetry representative.}
\label{fig:descent}
\end{figure}

\section{Derivative data of solutions and characteristic directions}\label{sec:2}

The symmetry classification is proved by lowering the differential
order of a pair of differential functions step by step; the adjoint
classification of Section~\ref{sec:8} uses the coordinates of
Section~\ref{sec:free}, the factorisation of Section~\ref{sec:symbols}
and Proposition~\ref{thm:dependence}. The tool is the
elementary criterion that a smooth function is independent of an
argument on a convex neighbourhood when its partial derivative with
respect to that argument vanishes there. The criterion is applied to
the derivatives of order $m$, with time, position, the fields and all
derivatives of lower order held fixed. Three ingredients are needed.
First, coordinates on the derivative data of solutions in which the
order-$m$ entries can be changed freely (Section~\ref{sec:free}).
Second, the linearisation of a differential function and its symbol,
which measure how the function responds to such a change
(Section~\ref{sec:symbols}). Third, a description of the admissible
changes of the order-$m$ entries: they are spanned by two families of
characteristic increments, so that a function whose response vanishes
along both families has order at most $m-1$
(Sections~\ref{sec:increments} and~\ref{sec:descent}).

\subsection{Free derivative coordinates}\label{sec:free}

The jet space of order $m$ has coordinates consisting of time,
position, the values of velocity and pressure, and their derivatives
through order $m$. A point of this space specifies derivative data
at a single space-time point, with mixed partials identified. Let
$\boldsymbol w$ denote the column of the three velocity components
followed by pressure. For a
multi-index $I=(I_0,I_1,I_2,I_3)$ in $(t,x_1,x_2,x_3)$, write
$|I|=\sum_\mu I_\mu$ and $\boldsymbol w_I=D_I\boldsymbol w$, and let
$\varepsilon_\mu$ be the unit index of the $\mu$th variable,
$\mu=0,\dots,3$. Throughout, $\partial f/\partial\boldsymbol w_I$
denotes the row of partial derivatives of $f$ with respect to the four
entries of $\boldsymbol w_I$; for vector-valued $f$ it is a matrix. The
total derivative of a differential function is
\[
 D_\mu f=\partial_\mu f
 +\sum_I\dfrac{\partial f}{\partial\boldsymbol w_I}\boldsymbol w_{I+\varepsilon_\mu},
\]
where $\partial_\mu$ differentiates only the explicit dependence
on the corresponding independent variable, holding all field values
and derivative coordinates $\boldsymbol w_I$ fixed.

Let $E$ denote the column of left-hand sides of the Euler equations,
defined in \eqref{eq:s1-4}. The derivative data of solutions, called
\emph{compatible} derivative data, form at order $m$ the set
\begin{equation}\label{eq:s1-3}
 j^m\boldsymbol w=(\boldsymbol w_I)_{|I|\le m},\qquad
 \cE_m=\{j^m\boldsymbol w:\ D_IE=0\ \text{for }|I|\le m-1\},
\end{equation}
and $\cE_\infty$ denotes compatible derivative data of all orders. The
field values themselves are unrestricted; the undifferentiated
equations \eqref{eq:s1-1} first constrain the data at order one. With
compatible derivative data through order $m-1$ fixed, the set of
admissible order-$m$ entries is called the \emph{fibre} over those
lower-order data. The fibre is the set in which the order-$m$ entries
may be changed when the criterion above is applied, so its structure
is determined first.

The Euler equations \eqref{eq:s1-1} contain no time derivative of the
pressure, so they cannot be solved for all four time derivatives
$\boldsymbol w_t$. They can be solved for all four derivatives in a
direction that combines $t$ and $x_3$. For the selected derivative
data, choose a constant $a$ with $a+u_3\ne0$; then $b_*=a+u_3$ remains
nonzero on a neighbourhood in jet space. Introduce the coordinates
\begin{equation}\label{eq:s2-tilted-coordinates}
 s=x_3+at,\qquad y=(y_0,y_1,y_2)=(t,x_1,x_2),\qquad
 D_t=D_{y_0}+aD_s,\quad D_{x_3}=D_s .
\end{equation}
Derivatives along $s$ at fixed $y$ are \emph{transverse}; derivatives
in $y$ at fixed $s$ are \emph{tangential}. A derivative in the coordinates
\eqref{eq:s2-tilted-coordinates} is written $D_s^kD_y^J$, with $k$ the
number of transverse differentiations and $J\in\mathbb N^3$ a
tangential multi-index; its order is $k+|J|$. With
\begin{equation}\label{eq:s2-tangential-transport}
 \mathcal Y=D_{y_0}+u_1D_{y_1}+u_2D_{y_2},\qquad Y=\mathcal Y+b_*D_s,
\end{equation}
the Euler equations \eqref{eq:s1-1} are solved for the four transverse
first derivatives:
\begin{equation}\label{eq:s2-explicit-solved-form}
 \begin{aligned}
 D_su_i&=-\dfrac{\mathcal Yu_i+D_{y_i}p}{b_*},\qquad i=1,2,\\
 D_su_3&=-D_{y_1}u_1-D_{y_2}u_2,\\
 D_sp&=-\mathcal Yu_3+b_*\left(D_{y_1}u_1+D_{y_2}u_2\right).
 \end{aligned}
\end{equation}
The first two lines are the first two momentum equations, the third is
the incompressibility equation, and the last is the third momentum
equation with $D_su_3$ eliminated by the third line. Denote the
right-hand sides of \eqref{eq:s2-explicit-solved-form} by $\cF$. The
column $E$ of \eqref{eq:s1-4} is affine in
$\boldsymbol w_s=D_s\boldsymbol w$ with Jacobian $\cN$:
\begin{equation}\label{eq:s3-normal-form}
 \begin{gathered}
 E=\cN\boldsymbol n,\qquad
 \boldsymbol n=\boldsymbol w_s-\cF(\boldsymbol w,D_y\boldsymbol w),\qquad
 \cN=\begin{pmatrix}b_*\mathsfbi I&\boldsymbol e_3\\
 \boldsymbol e_3^{\Tr}&0\end{pmatrix},
 \qquad\det\cN=-b_*^2,\\
 \cF=-\cN^{-1}\left(E|_{\boldsymbol w_s=0}\right),
 \qquad \boldsymbol e_3=(0,0,1)^{\Tr}.
 \end{gathered}
\end{equation}
Thus $\cF$ is linear in the first tangential derivatives, with smooth
coefficients depending on $\boldsymbol w$, and the \emph{residual}
$\boldsymbol n$ vanishes exactly where the solved equations hold.
The tilted coordinate is used because the Jacobian $\cN$ is
invertible: the same Jacobian with $x_3$ in place of $s$ has
determinant $-u_3^2$ and is singular wherever $u_3=0$, and with $t$ in
place of $s$ it is singular everywhere, because no equation in
\eqref{eq:s1-1} contains $p_t$. The constant $a$ depends on the
selected derivative data, and all statements are local.

The solved form \eqref{eq:s3-normal-form} separates the derivative
coordinates into two groups: the tangential derivatives
$D_y^J\boldsymbol w$, which the equations leave free, and the
residuals and their derivatives, which the equations set to zero.

\begin{lemma}[Free coordinates]\label{lem:free}
At every order $m$, with time and position fixed, independent
coordinates on the unrestricted jet space are
\begin{equation}\label{eq:s2-residual-coordinates}
 D_y^J\boldsymbol w\quad(|J|\le m),\qquad
 \boldsymbol n_{k,J}=D_s^kD_y^J\boldsymbol n\quad(k+|J|\le m-1).
\end{equation}
Both families can be prescribed independently; the residual entries
replace the derivatives containing $s$. For compatible derivative
data all residual entries are zero, so only the tangential family
remains free. Every set of compatible derivative data of finite order
extends to compatible data of every higher order. For $m\ge1$ the fibre over
compatible derivative data of order $m-1$ is an affine space with the
$4\binom{m+2}{2}$ tangential entries $D_y^J\boldsymbol w$, $|J|=m$, as
coordinates.
\end{lemma}
\begin{proof}
Differentiating $E=\cN\boldsymbol n$ in \eqref{eq:s3-normal-form}
gives the leading term $\cN D_s^kD_y^J\boldsymbol n$ plus terms
containing lower residual derivatives. Since $\cN$ is invertible,
successive differentiation shows that the equations $D_IE=0$ are
equivalent to the vanishing of the corresponding residual derivatives.
Each transverse derivative is reconstructed from
\begin{equation}\label{eq:s2-principal-recursion}
 D_s^{k+1}D_y^J\boldsymbol w
   =D_s^kD_y^J\cF+D_s^kD_y^J\boldsymbol n .
\end{equation}
The $\cF$ term has order at most $k+|J|+1$ and at most $k$ transverse
differentiations. Independence of \eqref{eq:s2-residual-coordinates}
follows by induction in total order and then in the number of
transverse differentiations: each equation is solved for a distinct
derivative, so the map from the derivatives
$(\boldsymbol w_I)_{|I|\le m}$ to the coordinates
\eqref{eq:s2-residual-coordinates} is triangular with respect to this
ordering and is a bijection; the counts agree, since
$4\binom{m+4}{4}=4\binom{m+3}{3}+4\binom{m+3}{4}$. The tangential data
and residual values therefore determine all remaining derivatives. Setting the residuals to zero gives
compatible derivative data for any prescribed tangential data.
Since $\cF$ is linear in the first derivatives, the reconstruction is
affine in the highest entries; there are $\binom{m+2}{2}$ tangential
indices of length $m$ per field.
\end{proof}

Two differential functions that agree on solutions should be regarded
as the same object, and the order of a function should be measured
after the Euler equations have been used to eliminate every derivative
they determine. Both notions are made precise through the free
coordinates. Expressing $f$ in the coordinates
\eqref{eq:s2-residual-coordinates} and setting all residual
derivatives to zero defines its \emph{tangential representative}
\begin{equation}\label{eq:s2-tangential-representative}
 \cR(f)=f\big|_{\boldsymbol n_{k,J}=0\ \text{for all }(k,J)} .
\end{equation}
The recursion \eqref{eq:s2-principal-recursion} does not increase order, and
independence of the tangential derivative data gives $f=g$ on
solutions exactly when $\cR(f)=\cR(g)$. An \emph{ambient extension}
of $f$ is a smooth function on unrestricted jet space whose restriction
to compatible derivative data agrees with $f$. The \emph{order of $f$
on solutions} is the least order of an ambient extension of $f$.
Since the recursion \eqref{eq:s2-principal-recursion} does not
increase order and $\cR(f)$, Eq.~\eqref{eq:s2-tangential-representative},
is itself an extension, this order equals $\ord\cR(f)$; in
particular it does not depend on the constant $a$ in
\eqref{eq:s2-tilted-coordinates}. The order on solutions measures the
derivative dependence that the Euler equations do not remove, and it
is the quantity lowered by the descent argument below.

\subsection{Linearisation, symbols and characteristic amplitudes}\label{sec:symbols}

The response of a differential function to a change of the derivative
data is measured by its linearisation. The \emph{linearisation} of a
differential function $f$ is the operator
\begin{equation}\label{eq:s2-linearization}
 \ell_f(\boldsymbol z)=\sum_I\dfrac{\partial f}{\partial\boldsymbol w_I}D_I\boldsymbol z,
 \qquad \boldsymbol z=(\boldsymbol v,\pi)^{\Tr},
\end{equation}
acting on a four-component test column, with coefficients evaluated at
the selected derivative data. It is the derivative at $\epsilon=0$
of $f$ under the change
$\boldsymbol w\mapsto\boldsymbol w+\epsilon\boldsymbol z$, with
time and position fixed. Its \emph{operator order} $\ord_{\rm op}$ is
the largest number of differentiations applied to the test column.
Three linearisations are used throughout: the linearised Euler
operator $L$, the operator $M$ of the adjoint equations
\eqref{eq:s3-adjoint-equations}, and the linearisation $\widehat L$ of
the residual,
\begin{equation}\label{eq:finite-LM}
 \begin{gathered}
 L=\ell_E,\qquad
 L(\boldsymbol v,\pi)=\binom{Y\boldsymbol v+\matA\boldsymbol v+\nabla\pi}{\Div\boldsymbol v},
 \qquad
 M(\boldsymbol r,r^0)=\binom{Y\boldsymbol r-\matA^{\Tr}\boldsymbol r+\nabla r^0}{\Div\boldsymbol r},\\
 \widehat L=\ell_{\boldsymbol n}=\mathsfbi ID_s-\ell_{\cF},
 \qquad L=\cN\widehat L\ \ \text{on solutions}.
 \end{gathered}
\end{equation}
In \eqref{eq:finite-LM} and below, $\nabla$ and $\Div$ consist of
spatial total derivatives acting on every argument of the test column.
On solutions, $M=-L^*$, where $L^*$ is the formal adjoint of $L$
obtained by integration by parts and written out in
\eqref{eq:s9-multiplier-Euler-identity}; the minus sign is included in $M$ so
that its derivative terms agree with those of $L$, the two operators
differing only in the multiplication terms $+\matA$ and $-\matA^{\Tr}$,
which have operator order zero. The last identity in
\eqref{eq:finite-LM} holds because
$\ell_E=\cN\ell_{\boldsymbol n}+(\ell_{\cN})\boldsymbol n$, with
$\cN$ of \eqref{eq:s3-normal-form}, and $\boldsymbol n=0$ on solutions.

The highest-order part of a linear differential operator is recorded
by its symbol. A \emph{covector} $\boldsymbol\zeta$ has components
$(\tau,\xxi)$ in $(t,\xx)$ and $(\lambda,\boldsymbol\eta)$ in $(s,y)$:
\begin{equation}\label{eq:s2-covector-split}
 \boldsymbol\zeta=\tau\dd t+\xi_i\dd x_i=\lambda\dd s+\eta_j\dd y_j,
 \qquad \lambda=\xi_3,\qquad
 \boldsymbol\eta=(\eta_0,\eta_1,\eta_2)=(\tau-a\xi_3,\xi_1,\xi_2).
\end{equation}
To isolate the derivatives of order $m$ in an operator, each $D_I$
with $|I|=m$ is replaced by the monomial $\boldsymbol\zeta^I$ and
the lower-order terms are omitted. The resulting polynomial matrix
is the \emph{degree-$m$ symbol}:
\begin{equation}\label{eq:s2-general-symbol-definition}
 \sigma_m\Bigl(\sum_{|I|\le m}\mathsfbi\Theta_ID_I\Bigr)(\boldsymbol\zeta)
 =\sum_{|I|=m}\mathsfbi\Theta_I\boldsymbol\zeta^I,\qquad
 \boldsymbol\zeta^I=\tau^{I_0}\xi_1^{I_1}\xi_2^{I_2}\xi_3^{I_3},
\end{equation}
with each coefficient matrix $\mathsfbi\Theta_I$ evaluated at the selected derivative data.
In the coordinates \eqref{eq:s2-tilted-coordinates}, the corresponding
substitution is $D_s^kD_y^J\mapsto\lambda^k\boldsymbol\eta^J$;
the symbol has the same value for the same covector in either
coordinate system. The symbol at the highest operator order is the
\emph{principal symbol}. Since $L$ and $\widehat L$ are first-order
operators, their principal symbols are
\begin{equation}\label{eq:s2-euler-symbol}
 \sigma_E(\boldsymbol\zeta)=\sigma_1(L)(\boldsymbol\zeta)
 =\begin{pmatrix}\alpha\mathsfbi I&\xxi\\ \xxi^{\Tr}&0\end{pmatrix},
 \quad \alpha=\tau+\uu\cdot\xxi,\quad q_{\xxi}=\xxi^{\Tr}\xxi,
 \quad \det\sigma_E=-\alpha^2q_{\xxi},
\end{equation}
\begin{equation}\label{eq:s2-solved-symbol}
 \sigma_1(\widehat L)(\boldsymbol\zeta)
 =\lambda\mathsfbi I-\mathsfbi B(\boldsymbol\eta)
 =\cN^{-1}\sigma_E(\boldsymbol\zeta),\qquad
 \mathsfbi B(\boldsymbol\eta)
 =\sum_{j=0}^2\dfrac{\partial\cF}{\partial(D_{y_j}\boldsymbol w)}\eta_j .
\end{equation}
The operators $L$ and $M$ in \eqref{eq:finite-LM} have the same
principal symbol $\sigma_E$. The matrix $\sigma_E$ is symmetric, and
the symbol of the formal adjoint $L^*$ is $-\sigma_E$. The symbol
identity in \eqref{eq:s2-solved-symbol} holds for all derivative data,
because the term $(\ell_{\cN})\boldsymbol n$ has operator order zero.

The determinant in \eqref{eq:s2-euler-symbol} is the product of two
factors with distinct physical origins: $\alpha$ is the symbol of the
material derivative $Y$, and $q_{\xxi}$ is the symbol of the Laplacian
in the pressure equation obtained by taking the divergence of the
momentum equation. A nonzero complex covector is \emph{characteristic}
when $\sigma_E(\boldsymbol\zeta)$ is singular, and a nonzero column in
its kernel is a \emph{characteristic amplitude}:
\begin{equation}\label{eq:s2-kernel-equations}
 \sigma_E(\boldsymbol\zeta)\boldsymbol z=0
 \ \Longleftrightarrow\
 \alpha\boldsymbol v+\xxi\pi=0,\ \ \xxi^{\Tr}\boldsymbol v=0
 \ \Longleftrightarrow\
 \mathsfbi B(\boldsymbol\eta)\boldsymbol z=\lambda\boldsymbol z .
\end{equation}
For a fixed tangential covector the characteristic values of $\lambda$
are thus the eigenvalues of $\mathsfbi B(\boldsymbol\eta)$, with the
amplitudes as eigenvectors. By the determinant in
\eqref{eq:s2-euler-symbol}, a nonzero amplitude exists exactly when
$\alpha=0$ or $q_{\xxi}=0$, which gives two families:
\begin{equation}\label{eq:s2-characteristic-families}
\begin{array}{lll}
 \text{transverse:}&\alpha=0,\quad\xxi\in\R^3,\quad\xxi\ne0,
 &\boldsymbol z=(\boldsymbol v,0)^{\Tr},\quad\xxi^{\Tr}\boldsymbol v=0;\\[4pt]
 \text{harmonic:}&q_{\xxi}=0,\quad\xxi\in\C^3,\quad\alpha\ne0,
 &\boldsymbol z=\boldsymbol h_{\boldsymbol\zeta}
 =\left(-\dfrac{\xxi}{\alpha},\,1\right)^{\Tr}.
\end{array}
\end{equation}
At a transverse covector the kernel is the real plane
$\xxi^{\Tr}\boldsymbol v=0$, $\pi=0$, of dimension two; at a harmonic
covector it is the complex line through
$\boldsymbol h_{\boldsymbol\zeta}$. A transverse pair describes a
short-wave velocity disturbance with wave vector $\xxi$, transverse
to $\xxi$ and without pressure, carried with the flow at the
velocity $\uu$: these are the disturbances of the short-wave
stability theory \citep{CraikCriminale1986,LifschitzHameiri1991}, and
they change the vorticity. A harmonic pair describes a pressure
disturbance with a harmonic spatial pattern and the velocity response
$-\xxi\pi/\alpha$ that keeps the flow divergence-free; it changes the
pressure gradient and the strain at fixed vorticity. The two families
together exhaust the admissible highest-order changes of the
derivative data, which is proved in Lemma~\ref{lem:descent} below.
The harmonic family contains the pure-time covectors $\xxi=0$,
$\tau\ne0$, whose kernel is the pressure line
\begin{equation}\label{eq:s2-pure-time-kernel}
 \ker\sigma_E(\tau,\boldsymbol 0)=\C\,(\boldsymbol 0,1)^{\Tr},
 \qquad \tau\ne0 .
\end{equation}
The derivative data are always real. Complex covectors and amplitudes
enter only through complex-linear extension of real-linear maps and
polynomial substitution; $\Tr$ is the bilinear transpose. Complex
covectors are needed because the pressure equation is elliptic: the
Laplace equation $\psi_{xx}+\psi_{yy}=0$ has no real characteristic
covector, while its complex characteristic covectors $(1,\pm\mathrm i)$
give tensors whose real and imaginary parts span all trace-free
Hessian increments.

A function that vanishes on solutions is not zero as a function of
the unrestricted derivative data, and its linearisation does not
vanish. The following lemma expresses such a linearisation through the
operator $\widehat L$, with a bound on the operator order of the
factor. It is used in Lemma~\ref{lem:kernel} to convert the
statement that a pair satisfies its determining equation on solutions
into an operator identity for arbitrary test columns, and its
factorisation \eqref{eq:s2-residual-factorization} is used in
Sections~\ref{sec:8} and~\ref{sec:9}.

\begin{lemma}[Residual factorisation]\label{lem:finite-residual}
If $f$ is a function of unrestricted derivative data of order at most $r\ge1$
and vanishes on solutions, then,
on $\cE_\infty$,
\begin{equation}\label{eq:s2-finite-residual}
 \ell_f=\widehat T\circ\widehat L,\qquad
 \widehat T(\boldsymbol z)=\sum_{k+|J|\le r-1}\mathsfbi\Theta_{k,J}D_s^kD_y^J\boldsymbol z,
 \qquad \ord_{\rm op}\widehat T\le r-1 .
\end{equation}
Here $\boldsymbol z$ is an arbitrary four-component test column and the
coefficient matrices $\mathsfbi\Theta_{k,J}$ depend smoothly on finitely many field
derivatives; for scalar $f$ they are rows. The background coefficients
in \eqref{eq:s2-finite-residual} are evaluated on solutions, while the test
column remains unrestricted.
\end{lemma}
\begin{proof}
By Lemma~\ref{lem:free}, $f$ vanishes on order-$r$ derivative data
whose residuals are zero. On a coordinate neighbourhood convex in the
residual entries, hold the tangential data fixed and integrate the
partial derivatives of $f$ along segments from zero residuals to the
given residual values. The resulting factorisation has smooth coefficients
$f_{k,J}$ with
\begin{equation}\label{eq:s2-residual-factorization}
 f=\sum_{k+|J|<r}f_{k,J}\,\boldsymbol n_{k,J},\qquad
 \ell_f\big|_{\cE_\infty}
   =\sum_{k+|J|<r}f_{k,J}\big|_{\cE_\infty}\,D_s^kD_y^J\circ\widehat L .
\end{equation}
In the linearisation of this factorisation, the derivatives of the
coefficients $f_{k,J}$ multiply zero residuals, and
linearisation commutes with total differentiation, so each residual
contributes $D_s^kD_y^J\circ\widehat L$. The order bound in
\eqref{eq:s2-finite-residual} follows and is preserved by the linear
coordinate change \eqref{eq:s2-tilted-coordinates}.
\end{proof}

\subsection{Characteristic increments and the response}\label{sec:increments}

The admissible changes of the order-$m$ entries are now described.
Fix $m\ge1$ and compatible derivative data. An \emph{order-$m$
increment} changes the entries $\boldsymbol w_I$, $|I|=m$, and keeps
the independent variables and all lower entries fixed. The order-$m$
part of $D_IE$, $|I|=m-1$, is
\[
 \sum_{\mu=0}^{3}\dfrac{\partial E}{\partial\boldsymbol w_{\varepsilon_\mu}}
 \boldsymbol w_{I+\varepsilon_\mu},
\]
with coefficient matrices depending on $\boldsymbol w$ only, so the
fibre of Lemma~\ref{lem:free} is affine with direction space
\begin{equation}\label{eq:s2-prolonged-symbol}
 g_m=\Bigl\{(\boldsymbol T_I)_{|I|=m}:\
 \sum_{\mu=0}^{3}\dfrac{\partial E}{\partial\boldsymbol w_{\varepsilon_\mu}}
 \boldsymbol T_{I+\varepsilon_\mu}=0\ \text{ for all }|I|=m-1\Bigr\},
\end{equation}
the \emph{symbol space} of order $m$; its complexification
$g_m^{\C}$ consists of the same arrays with complex entries satisfying
the same linear equations. By Lemma~\ref{lem:free} an increment in $g_m$ is
determined by its tangential entries $\delta(D_y^J\boldsymbol w)$,
$|J|=m$, and these are free. For a characteristic pair
$(\boldsymbol\zeta,\boldsymbol z)$ the \emph{characteristic increment}
is
\begin{equation}\label{eq:s2-characteristic-increments}
 \boldsymbol T(\boldsymbol\zeta,\boldsymbol z)
 =(\boldsymbol\zeta^I\boldsymbol z)_{|I|=m},\qquad
 \sum_{\mu=0}^{3}\dfrac{\partial E}{\partial\boldsymbol w_{\varepsilon_\mu}}
 \boldsymbol\zeta^{I+\varepsilon_\mu}\boldsymbol z
 =\boldsymbol\zeta^I\sigma_E(\boldsymbol\zeta)\boldsymbol z=0 .
\end{equation}
It is the array of order-$m$ derivatives of the plane wave
$\boldsymbol z\exp(\tau t+\xxi\cdot\xx)$ at the point, and the last
identity in \eqref{eq:s2-characteristic-increments} states that it
preserves the order-$m$ compatibility conditions: it belongs to
$g_m^{\C}$, and its real and imaginary parts belong to $g_m$. In the
coordinates \eqref{eq:s2-tilted-coordinates} its entry for
$D_s^kD_y^J$ is $\lambda^k\boldsymbol\eta^J\boldsymbol z$; in
particular the tangential entries are
$\boldsymbol\eta^J\boldsymbol z$. At order one the increment changes
$D_\mu\boldsymbol w$ by $\zeta_\mu\boldsymbol z$. For a transverse pair
and for the harmonic pair with amplitude
$\alpha\boldsymbol h_{\boldsymbol\zeta}=(-\xxi,\alpha)^{\Tr}$,
substitution in $\delta(D_\mu\boldsymbol w)=\zeta_\mu\boldsymbol z$
gives, with $\matA$ and $\ww$ as in \eqref{eq:s1-2},
\begin{equation}\label{eq:s2-first-order-increments}
 \begin{array}{lllll}
 \text{transverse:}&
 \delta\matA=\boldsymbol v\xxi^{\Tr},&
 \delta\ww=\xxi\times\boldsymbol v,&
 \delta(\nabla p)=0,&
 \delta p_t=0;\\[4pt]
 \text{harmonic:}&
 \delta\matA=-\xxi\xxi^{\Tr},&
 \delta\ww=0,&
 \delta(\nabla p)=\alpha\xxi,&
 \delta p_t=\alpha\tau .
 \end{array}
\end{equation}
In both cases $\delta\uu_t=\tau\boldsymbol v$. Thus the transverse
increment changes the velocity gradient and the vorticity at fixed
pressure gradient, and the harmonic increment changes the symmetric
part of the velocity gradient and the pressure gradient at fixed
vorticity, as stated after \eqref{eq:s2-characteristic-families}.

Let $f$ have order at most $m$ on solutions and let $\widehat f$ be
any ambient extension of order at most $m$; write
$\sigma_{\widehat f,m}=\sigma_m(\ell_{\widehat f})$. The
\emph{degree-$m$ response} of $f$ to a characteristic pair is the
derivative of $\widehat f$ along the increment
\eqref{eq:s2-characteristic-increments}:
\begin{equation}\label{eq:s2-intrinsic-response}
 \sum_{|I|=m}\dfrac{\partial\widehat f}{\partial\boldsymbol w_I}
 \boldsymbol\zeta^I\boldsymbol z
 =\sigma_{\widehat f,m}(\boldsymbol\zeta)\boldsymbol z
 =\sum_{|J|=m}\dfrac{\partial\cR(f)}{\partial(D_y^J\boldsymbol w)}
 \boldsymbol\eta^J\boldsymbol z .
\end{equation}
The response is independent of the extension: two extensions of order
at most $m$ differ by a function vanishing on $\cE_m$, whose
differential vanishes on the fibre direction space $g_m$ of
\eqref{eq:s2-prolonged-symbol}, hence on $g_m^{\C}$. The last
expression in \eqref{eq:s2-intrinsic-response} is the case
$\widehat f=\cR(f)$, which contains no transverse derivatives. The
response is a polynomial in the covector and the amplitude with
coefficients depending smoothly on the derivative data; for complex pairs it is
the complex-linear extension of the real differential, and for
vector-valued $f$ it is a column. At order zero the response is
\begin{equation}\label{eq:s2-zero-order-response}
 \dfrac{\partial f}{\partial\boldsymbol w}\boldsymbol z ,
\end{equation}
the differential of $f$ in the four field values at fixed independent
variables.

\subsection{Spanning, descent and null-cone division}\label{sec:descent}

The following lemma is the order-reduction criterion in the form used
in Section~\ref{sec:sym}: a function whose responses
to both characteristic families vanish has lower order on solutions.

\begin{lemma}[Characteristic spanning and descent]\label{lem:descent}\label{lem:spanning}
Let $m\ge1$.
\begin{enumerate}[label=(\roman*),nosep]
\item The real and imaginary parts of the increments
\eqref{eq:s2-characteristic-increments} of the two families
\eqref{eq:s2-characteristic-families} span $g_m$.
\item Let $f$ be a real- or complex-valued differential function of
order at most $m$ on solutions. If
$\sigma_{\widehat f,m}(\boldsymbol\zeta)\boldsymbol z=0$ for all
characteristic pairs of both families for every set of compatible
derivative data in a neighbourhood,
then $f$ has order at most $m-1$ on solutions on a smaller neighbourhood.
If $f$ has order zero and \eqref{eq:s2-zero-order-response} vanishes
for all such pairs, then $f$ depends only on $(t,\xx)$. Vector-valued
$f$ are covered componentwise.
\end{enumerate}
\end{lemma}
\begin{proof}
Fix compatible derivative data and a real tangential covector $\boldsymbol\eta$ with
$\varrho^2=\eta_1^2+\eta_2^2>0$. By \eqref{eq:s2-solved-symbol} and
\eqref{eq:s2-euler-symbol},
\[
 \det\left(\lambda\mathsfbi I-\mathsfbi B(\boldsymbol\eta)\right)
 =\dfrac{\det\sigma_E}{\det\cN}
 =\dfrac{\alpha^2q_{\xxi}}{b_*^2},\qquad
 \alpha=\eta_0+u_1\eta_1+u_2\eta_2+b_*\lambda,\qquad
 q_{\xxi}=\varrho^2+\lambda^2,
\]
so the eigenvalues of $\mathsfbi B(\boldsymbol\eta)$ are
\begin{equation}\label{eq:s2-three-roots}
 \lambda_0=-\dfrac{\eta_0+u_1\eta_1+u_2\eta_2}{b_*},\qquad
 \lambda_\pm=\pm\mathrm i\varrho .
\end{equation}
At $\lambda_0$ the covector is real with $\alpha=0$ and
$\xxi=(\eta_1,\eta_2,\lambda_0)\ne0$: a transverse covector, whose
kernel has dimension two. At $\lambda_\pm$ the covector has
$q_{\xxi}=0$ and $\alpha=b_*(\lambda_\pm-\lambda_0)\ne0$: a harmonic
covector, whose kernel has dimension one. The eigenvalue $\lambda_0$
is real, because the derivative data are real and $b_*\ne0$, while
$\lambda_\pm=\pm\mathrm i\varrho$ are purely imaginary and distinct,
because $\varrho>0$. The characteristic polynomial has degree four,
$\lambda_0$ is a double root with a two-dimensional kernel, and
$\lambda_\pm$ are simple roots; eigenvectors of distinct eigenvalues
are linearly independent, so the amplitudes of the two families over
$\boldsymbol\eta$, two transverse and two harmonic, span $\C^4$.

(i) The tangential entries of these increments span all arrays
$(\boldsymbol\eta^J\boldsymbol c)_{|J|=m}$, $\boldsymbol c\in\C^4$. A
linear functional with coefficient rows $C_J$ that annihilates them
for every $\boldsymbol\eta$ with $\varrho>0$ satisfies
$\sum_{|J|=m}C_J\boldsymbol\eta^J=0$ on an open set, so every $C_J=0$.
Hence the increments span all tangential entries over $\C$. By
Lemma~\ref{lem:free} an element of $g_m^{\C}$ is determined by its
tangential entries, so the increments span $g_m^{\C}$, and their real
and imaginary parts span $g_m$.

(ii) By \eqref{eq:s2-intrinsic-response} the response to each pair
over $\boldsymbol\eta$ is $S(\boldsymbol\eta)\boldsymbol z$ with the
polynomial row
\[
 S(\boldsymbol\eta)=\sum_{|J|=m}
 \dfrac{\partial\cR(f)}{\partial(D_y^J\boldsymbol w)}\,\boldsymbol\eta^J .
\]
The hypothesis gives $S(\boldsymbol\eta)\boldsymbol z=0$ for amplitudes
$\boldsymbol z$ spanning $\C^4$, so $S(\boldsymbol\eta)=0$ for every
$\boldsymbol\eta$ with $\varrho>0$, and $S$ is the zero polynomial:
all $\partial\cR(f)/\partial(D_y^J\boldsymbol w)$, $|J|=m$, vanish at
every set of compatible derivative data in the neighbourhood.
The order-$m$ tangential entries are independent coordinates by
Lemma~\ref{lem:free}. Hold the lower-order data fixed and integrate
along segments in a smaller neighbourhood convex in those highest
entries. Since all the corresponding partial derivatives vanish,
$\cR(f)$ is independent of the order-$m$ entries and is an extension
of order $m-1$. At order zero the same
argument applies to $S=\partial f/\partial\boldsymbol w$. At
$\varrho=0$ with $\eta_0\ne0$ the two harmonic roots in
\eqref{eq:s2-three-roots} coincide at $\lambda=0$ with the single
eigenvector \eqref{eq:s2-pure-time-kernel}, which is why the open set
$\varrho>0$ is used.
\end{proof}

The harmonic family is a complex cone, and identities proved on it are
polynomial identities modulo $q_{\xxi}$. The following lemma converts
vanishing on that family into divisibility by $q_{\xxi}$ of
\eqref{eq:s2-euler-symbol}; it is used in Section~\ref{sec:4} and in
the harmonic-response calculations of Section~\ref{sec:sym}.

\begin{lemma}[Null-cone division]\label{lem:null-cone}
Let $P(\tau,\xxi)$ be a polynomial whose coefficients depend smoothly
on the derivative data, and let $P=0$ wherever $q_{\xxi}=0$ and
$\alpha\ne0$. Then $P=q_{\xxi}H$ with $H$ polynomial and coefficients
smooth in those data. The form $q_{\xxi}$ is prime in $\C[\tau,\xxi]$ and
coprime to $\alpha$ and to each $\xi_i$.
\end{lemma}
\begin{proof}
A product of two linear forms has rank at most two and $q_{\xxi}$ has
rank three, so $q_{\xxi}$ is irreducible, hence prime: if it divides
a product of polynomials, it divides at least one factor. The
polynomials $\alpha$ and $\xi_i$ are linear and share no factor with
$q_{\xxi}$. Division by $q_{\xxi}$,
a monic quadratic in $\xi_3$, gives
\begin{equation}\label{eq:s2-null-cone-division}
 P=q_{\xxi}H+e_1(\tau,\xi_1,\xi_2)\,\xi_3+e_0(\tau,\xi_1,\xi_2).
\end{equation}
At fixed real $(\xi_1,\xi_2)$ with $\xi_1^2+\xi_2^2\ne0$, the two
roots $\xi_3=\pm\mathrm i\left(\xi_1^2+\xi_2^2\right)^{1/2}$ are
distinct, and $\alpha$ vanishes at each root for exactly one value of
$\tau$. For all other $\tau$ the remainder vanishes at both roots, so
$e_0=e_1=0$ at $(\tau,\xi_1,\xi_2)$. A polynomial in $\tau$ that
vanishes outside a finite set vanishes identically, and a polynomial
in $(\xi_1,\xi_2)$ that vanishes on an open set vanishes identically;
hence $e_0=e_1=0$. Division is linear in the coefficients of $P$ and
preserves their smooth dependence on the derivative data.
\end{proof}

\section{The two characteristic tests}\label{sec:4}

Let $q=(\boldsymbol r,r^0)$ be a local symmetry, a pair of
differential functions with $Lq=0$ on solutions. In this section the
determining equation
is used to restrict the highest-order response of $q$. The first
lemma shows that the degree-$n$ symbol of $q$ maps each characteristic
kernel into itself. The harmonic test (Section~\ref{sec:harmonic-test})
then applies the determining equation to a test column built on a
harmonic covector and yields a transport equation for the harmonic
response along the direction $\xxi$. The transverse test
(Section~\ref{sec:transverse-test}) applies it to a short-wave test
column built on a transverse covector and yields a transport law for
the transverse response along the flow. The results of both tests are
used in the symmetry classification of Section~\ref{sec:sym}; the
argument is summarised in Figure~\ref{fig:descent}.

\begin{lemma}[Determining identity and kernel preservation]\label{lem:kernel}
Let $q=(\boldsymbol r,r^0)$ be a pair of differential functions with
$Lq=0$ on solutions and order at most $n\ge1$, represented by an
ambient pair of the same order, with $L$ as in \eqref{eq:finite-LM}.
Then there is a $4\times4$ differential operator $T$
with $\ord_{\rm op}T\le n$ such that, on solutions and for every test
column $(\boldsymbol v,\pi)$,
\begin{equation}\label{eq:kernel-identity}
 L\ell_q(\boldsymbol v,\pi)
 =TL(\boldsymbol v,\pi)+\binom{\mathcal C_L(\boldsymbol v)\boldsymbol r}{0},
 \qquad
 \mathcal C_L(\boldsymbol v)\boldsymbol r=-(D_{\xx}\boldsymbol v)\boldsymbol r
   -(\boldsymbol v\cdot D_{\xx})\boldsymbol r .
\end{equation}
The coefficient term has operator order one in $\boldsymbol v$, with
degree-one symbol $-\boldsymbol v(\xxi\cdot\boldsymbol r)$; it does not
enter the degree-$(n+1)$ comparison below, and it enters the
degree-$n$ comparisons of Section~\ref{sec:sym} only for $n=1$.
Consequently,
with $\sigma_{q,n}=\sigma_n(\ell_q)$ and $T_n=\sigma_n(T)$,
\begin{equation}\label{eq:kernel-preservation}
 \sigma_E(\boldsymbol\zeta)\sigma_{q,n}(\boldsymbol\zeta)
 =T_n(\boldsymbol\zeta)\sigma_E(\boldsymbol\zeta),
 \qquad
 \sigma_E(\boldsymbol\zeta)\boldsymbol z=0\ \Longrightarrow\
 \sigma_E(\boldsymbol\zeta)\,\sigma_{q,n}(\boldsymbol\zeta)\boldsymbol z=0 :
\end{equation}
the degree-$n$ symbol of $q$ maps each characteristic kernel into
itself. By \eqref{eq:s2-intrinsic-response},
$\sigma_{q,n}(\boldsymbol\zeta)\boldsymbol z$ is the response of $q$,
independent of the ambient representative.
\end{lemma}
\begin{proof}
The coefficients of $L$ are $\uu$ and $\matA=D_{\xx}\uu$, which change
by $\boldsymbol v$ and $D_{\xx}\boldsymbol v$ under
$\uu\mapsto\uu+\epsilon\boldsymbol v$. Differentiating $Lq$ with
respect to $\epsilon$ at zero by the product rule, with $q$ itself
changing by $\epsilon\ell_q(\boldsymbol v,\pi)$, gives the
unrestricted identity
\[
 \ell_{Lq}=L\ell_q+\binom{(\boldsymbol v\cdot D_{\xx})\boldsymbol r
   +(D_{\xx}\boldsymbol v)\boldsymbol r}{0}
 =L\ell_q-\binom{\mathcal C_L(\boldsymbol v)\boldsymbol r}{0}.
\]
The function $Lq$ has order at most $n+1$ and vanishes on solutions, so
Lemma~\ref{lem:finite-residual} and $L=\cN\widehat L$ in
\eqref{eq:finite-LM} give
$\ell_{Lq}=TL$ on solutions with $T=\widehat T\cN^{-1}$,
$\ord_{\rm op}T\le n$. Substitution in the linearised identity gives
\eqref{eq:kernel-identity}. Both sides
are operators of order at most $n+1$ in the test column; their degree-$(n+1)$
symbols are $\sigma_E\sigma_{q,n}$ and $T_n\sigma_E$, since
$\sigma_1(L)=\sigma_E$ by \eqref{eq:s2-euler-symbol} and the
coefficient term has order one. Applying
the first identity in \eqref{eq:kernel-preservation} to $\boldsymbol z$
gives the second.
\end{proof}

\subsection{Harmonic test}\label{sec:harmonic-test}

The harmonic test evaluates the determining identity
\eqref{eq:kernel-identity} on a test column that oscillates with a
fixed harmonic covector. Choose a constant complex covector
$\boldsymbol\zeta=(\tau,\xxi)$ of the harmonic family
\eqref{eq:s2-characteristic-families} with $\xxi\ne0$, and work on a
neighbourhood of compatible derivative data where $\alpha\ne0$. The
covector is held fixed under every total derivative, while $\alpha$
depends on the velocity:
\begin{equation}\label{eq:harmonic-fixed-parameters}
 \begin{gathered}
 D_\mu\tau=0,\qquad D_\mu\xxi=0,\qquad
 \alpha=\tau+\uu\cdot\xxi,\\
 D_\mu\alpha=\xi_jD_\mu u_j,\qquad
 \nabla\alpha=\matA^{\Tr}\xxi.
 \end{gathered}
\end{equation}
For a nonzero constant parameter $\kappa$, the exponential factor
\begin{equation}\label{eq:core-formal-exponential}
 \mathfrak e=\exp\!\bigl[\kappa(\tau t+\xxi\cdot\xx)\bigr],\qquad
 D_\mu\mathfrak e=\kappa\zeta_\mu\mathfrak e,
\end{equation}
turns each total derivative of the factor into multiplication by
$\kappa\zeta_\mu$, which is the substitution defining the symbol in
\eqref{eq:s2-general-symbol-definition}. After division by $\mathfrak e$,
every expression below is a finite sum of integer powers of $\kappa$,
and the determining identity becomes a polynomial identity in
$\kappa$ whose coefficients are compared degree by degree. For every
four-component column $\boldsymbol c$ of differential functions,
applying the product rule to $\mathfrak e\boldsymbol c$ gives
\begin{equation}\label{eq:harmonic-product-rule}
 L(\mathfrak e\boldsymbol c)
 =\mathfrak e\bigl[\kappa\sigma_E(\boldsymbol\zeta)\boldsymbol c+L\boldsymbol c\bigr].
\end{equation}

The harmonic kernel is spanned by $\boldsymbol h_{\boldsymbol\zeta}$ in
\eqref{eq:s2-characteristic-families}; multiplication by $\alpha$ gives
an amplitude whose entries are polynomial in the covector:
\begin{equation}\label{eq:finite-harmonic-test}
 \boldsymbol h_{\boldsymbol\zeta}=\binom{-\xxi/\alpha}{1},\qquad
 \boldsymbol b_0=\alpha\boldsymbol h_{\boldsymbol\zeta}
               =\binom{-\xxi}{\alpha},\qquad
 \sigma_E(\boldsymbol\zeta)\boldsymbol b_0=\binom{\boldsymbol 0}{-\xxi^{\Tr}\xxi}=0.
\end{equation}
The column $\mathfrak e\boldsymbol b_0$ is annihilated by the symbol
of $L$ but not by $L$ itself, because the pressure entry $\alpha$ of
$\boldsymbol b_0$ is a differential function: \eqref{eq:finite-LM}
and \eqref{eq:harmonic-fixed-parameters} give
\begin{equation}\label{eq:harmonic-leading-residual}
 L\boldsymbol b_0
 =\binom{-\matA\xxi+\nabla\alpha}{0}
 =\binom{(\matA^{\Tr}-\matA)\xxi}{0}.
\end{equation}
This term would enter the determining identity at the same degree
in $\kappa$ as the response under examination. It is removed by a
correction of one degree lower whose symbol cancels
\eqref{eq:harmonic-leading-residual}. Put
\begin{equation}\label{eq:harmonic-corrected-test}
 \boldsymbol b_1=\binom{(\matA-\matA^{\Tr})\xxi/\alpha}{0},\qquad
 \boldsymbol a=\mathfrak e
          \left(\boldsymbol b_0+\kappa^{-1}\boldsymbol b_1\right),\qquad
 \gamma=\xxi^{\Tr}\matA\xxi .
\end{equation}
For polynomials in the covector with coefficients depending on the derivative data,
$K\equiv K'\pmod{\xxi^{\Tr}\xxi}$ means that $K-K'$ is divisible by
$\xxi^{\Tr}\xxi$. By Lemma~\ref{lem:null-cone}, this divisibility holds exactly when
$K-K'$ vanishes on the null cone wherever $\alpha\ne0$.

The following lemma lists the four facts about the harmonic test
that the symmetry classification uses: the degree bookkeeping of the test
column (a), the contraction formula that converts the determining
identity into a transport equation along $\xxi$ (b), the form of a
polynomial row that annihilates every transverse plane (c), and the
descent statement that a polynomial transported along $\xxi$ without
change has field-independent coefficients (d).

\begin{lemma}[Harmonic test]\label{lem:harmonic-test}
Let $\boldsymbol\zeta$, $\mathfrak e$, $\boldsymbol b_0$,
$\boldsymbol b_1$, $\boldsymbol a$ and $\gamma$ be as in
\eqref{eq:harmonic-fixed-parameters}--\eqref{eq:harmonic-corrected-test},
and let $N\ge1$.
\begin{enumerate}[label=(\alph*)]
\item\label{ht-a}
The corrected test satisfies
\begin{equation}\label{eq:finite-harmonic-error}
 \sigma_E(\boldsymbol\zeta)\boldsymbol b_1=-L\boldsymbol b_0,\qquad
 L\boldsymbol a=\mathfrak e\,\kappa^{-1}L\boldsymbol b_1 .
\end{equation}
Consequently $\deg_\kappa(\mathfrak e^{-1}TL\boldsymbol a)\le N-1$ for
every differential operator $T$ with $\ord_{\rm op}T\le N$, and every
pair $q$ represented by an ambient pair of order at most $N$ has the
expansion
\begin{equation}\label{eq:harmonic-leading-output}
 \ell_q(\boldsymbol a)=\mathfrak e
 \bigl(\kappa^N\boldsymbol c_0+\kappa^{N-1}\boldsymbol c_1+\cdots\bigr),
 \qquad
 \boldsymbol c_0=\sigma_{q,N}(\boldsymbol\zeta)\boldsymbol b_0,
\end{equation}
the dots denoting powers of $\kappa$ at most $N-2$. If moreover $Lq=0$
on solutions, then, with $(\boldsymbol W_N,Z_N)$ the last row of
$\sigma_{q,N}(\boldsymbol\zeta)$,
\begin{equation}\label{eq:harmonic-scalar-output}
 \boldsymbol c_0=\boldsymbol h_{\boldsymbol\zeta}H_N,\qquad
 H_N=\alpha Z_N-\boldsymbol W_N\xxi ,
\end{equation}
where $H_N$ is a homogeneous polynomial of degree $N+1$ in the covector
whose coefficients are differential functions.
\item\label{ht-b}
For every scalar differential function $H$ and every covector with
$\xxi^{\Tr}\xxi=0$ and $\alpha\ne0$,
\begin{equation}\label{eq:harmonic-projection-product}
 \boldsymbol h_{\boldsymbol\zeta}^{\Tr}L(\boldsymbol h_{\boldsymbol\zeta}H)
 =-\dfrac{2}{\alpha^2}\left(\alpha\,\xxi\cdot D_{\xx}H-\gamma H\right).
\end{equation}
\item\label{ht-c}
If a row $\boldsymbol W(\xxi)$, polynomial in $\xxi$ with coefficients
smooth on the neighbourhood in jet space, annihilates $\xxi^\perp$ for every real
$\xxi\ne0$, then
\begin{equation}\label{eq:core-harmonic-root}
 \boldsymbol W(\xxi)=w(\xxi)\,\xxi^{\Tr}
\end{equation}
for a polynomial $w$ with smooth coefficients. In particular, if $Lq=0$
on solutions, then
$\boldsymbol W_N(-\uu\cdot\xxi,\xxi)=w_N\xxi^{\Tr}$ and
$H_N(-\uu\cdot\xxi,\xxi)=-w_N\,\xxi^{\Tr}\xxi$, with $w_N$ a polynomial
of degree $N-1$ in $\xxi$.
\item\label{ht-d}
Let $K(\tau,\xxi)$ be a polynomial in the covector whose coefficients
are real- or complex-valued differential functions of finite order on
solutions, defined on a neighbourhood in $\cE_\infty$. If
$\xxi\cdot D_{\xx}K\equiv0\pmod{\xxi^{\Tr}\xxi}$ on solutions, then
in a smaller neighbourhood
\[
 K\equiv\Phi(t,\xx,\tau,\xxi)\pmod{\xxi^{\Tr}\xxi}
 \qquad\text{on solutions},
\]
where $\Phi$ is a polynomial in the covector whose coefficients depend
only on $(t,\xx)$. If $K$ is homogeneous of degree $d$ in the covector,
so is $\Phi$.
\end{enumerate}
\end{lemma}

\begin{proof}
\ref{ht-a} By \eqref{eq:harmonic-product-rule},
\begin{equation}\label{eq:harmonic-error-expansion}
 \begin{gathered}
 L\boldsymbol a=\mathfrak e\left[
   \kappa\sigma_E(\boldsymbol\zeta)\boldsymbol b_0
   +L\boldsymbol b_0+\sigma_E(\boldsymbol\zeta)\boldsymbol b_1
   +\kappa^{-1}L\boldsymbol b_1\right],\\
 \sigma_E(\boldsymbol\zeta)\boldsymbol b_1
 =\binom{(\matA-\matA^{\Tr})\xxi}
        {\xxi^{\Tr}(\matA-\matA^{\Tr})\xxi/\alpha}.
 \end{gathered}
\end{equation}
The scalar entry of $\sigma_E(\boldsymbol\zeta)\boldsymbol b_1$ vanishes
because $\matA-\matA^{\Tr}$ is skew-symmetric, so
$\sigma_E(\boldsymbol\zeta)\boldsymbol b_1=-L\boldsymbol b_0$ by
\eqref{eq:harmonic-leading-residual}, and
$\sigma_E(\boldsymbol\zeta)\boldsymbol b_0=0$ by
\eqref{eq:finite-harmonic-test}. Substitution in
\eqref{eq:harmonic-error-expansion} gives
\eqref{eq:finite-harmonic-error}. Each derivative in $T$ raises the
$\kappa$-degree by at most one, which gives the bound on
$\mathfrak e^{-1}TL\boldsymbol a$. For the ambient pair $q$ of order at
most $N$, the terms of $\ell_q(\mathfrak e\boldsymbol b_0)$ in which
all $N$ derivatives fall on $\mathfrak e$ give
$\mathfrak e\kappa^N\sigma_{q,N}(\boldsymbol\zeta)\boldsymbol b_0$ by
\eqref{eq:s2-general-symbol-definition}, every derivative falling on the
entry $\alpha$ of $\boldsymbol b_0$ lowers the degree, and
$\ell_q(\mathfrak e\kappa^{-1}\boldsymbol b_1)$ has degree at most
$N-1$, proving \eqref{eq:harmonic-leading-output}. If $Lq=0$
on solutions, \eqref{eq:kernel-preservation} gives
$\sigma_E(\boldsymbol\zeta)\boldsymbol c_0
=T_N(\boldsymbol\zeta)\sigma_E(\boldsymbol\zeta)\boldsymbol b_0=0$.
Write $\boldsymbol c_0=(\boldsymbol C_N,H_N)^{\Tr}$. The kernel
equations \eqref{eq:s2-kernel-equations} at $\boldsymbol\zeta$ give
$\alpha\boldsymbol C_N+\xxi H_N=0$, so
$\boldsymbol C_N=-\xxi H_N/\alpha$ and
$\boldsymbol c_0=\boldsymbol h_{\boldsymbol\zeta}H_N$, while $H_N$ is
the last entry of $\sigma_{q,N}(\boldsymbol\zeta)\boldsymbol b_0$,
namely $\alpha Z_N-\boldsymbol W_N\xxi$. The symbol has degree $N$ and
$\boldsymbol b_0$ has degree one, so $H_N$ is homogeneous of degree
$N+1$.

\ref{ht-b} Denote the velocity entries of $\boldsymbol h_{\boldsymbol\zeta}$
by $\boldsymbol v_h=-\xxi/\alpha$. By
\eqref{eq:harmonic-fixed-parameters},
\begin{equation}\label{eq:harmonic-amplitude-derivatives}
 D_j(v_h)_i=\dfrac{\xi_iD_j\alpha}{\alpha^2},\qquad
 \Div\boldsymbol v_h
 =\dfrac{\xxi^{\Tr}\nabla\alpha}{\alpha^2}
 =\dfrac{\xxi^{\Tr}\matA^{\Tr}\xxi}{\alpha^2}
 =\dfrac{\gamma}{\alpha^2},
\end{equation}
and the bilinear contractions are
\begin{equation}\label{eq:core-harmonic-coefficients}
 \boldsymbol v_h^{\Tr}\boldsymbol v_h=\dfrac{\xxi^{\Tr}\xxi}{\alpha^2}=0,
 \qquad
 \boldsymbol v_h^{\Tr}Y\boldsymbol v_h
       =-\dfrac{\xxi^{\Tr}\xxi\,Y\alpha}{\alpha^3}=0,
 \qquad
 \boldsymbol v_h^{\Tr}\matA\boldsymbol v_h
       =\dfrac{\gamma}{\alpha^2}.
\end{equation}
Expanding $L(\boldsymbol h_{\boldsymbol\zeta}H)$ by \eqref{eq:finite-LM}
and contracting with $\boldsymbol h_{\boldsymbol\zeta}^{\Tr}$ gives
\[
 \begin{aligned}
 \boldsymbol h_{\boldsymbol\zeta}^{\Tr}L(\boldsymbol h_{\boldsymbol\zeta}H)
 &=(\boldsymbol v_h^{\Tr}\boldsymbol v_h)YH
  +2\boldsymbol v_h\cdot D_{\xx}H\\
 &\quad+H\left[\boldsymbol v_h^{\Tr}Y\boldsymbol v_h
   +\boldsymbol v_h^{\Tr}\matA\boldsymbol v_h
   +\Div\boldsymbol v_h\right]\\
 &=-\dfrac{2}{\alpha}\,\xxi\cdot D_{\xx}H+\dfrac{2\gamma}{\alpha^2}H,
 \end{aligned}
\]
which is \eqref{eq:harmonic-projection-product}.

\ref{ht-c} For real $\xxi\ne0$ the row is proportional to $\xxi^{\Tr}$,
so $W_i\xi_j-W_j\xi_i=0$ for all $i,j$. These are polynomial identities
in $\xxi$ at fixed derivative data, valid on an open set of real $\xxi$, hence
identically and for complex $\xxi$. For $i\ne j$ the coprime linear
forms $\xi_i$ and $\xi_j$ give $\xi_i\mid W_i$, and the identities make
the three quotients $W_i/\xi_i$ equal to one polynomial $w$. Division
is linear in the coefficients, so $w$ has smooth coefficients. For the
particular case, fix a real $\xxi\ne0$ and
$\boldsymbol\zeta=(-\uu\cdot\xxi,\xxi)$. The kernel of
$\sigma_E(\boldsymbol\zeta)$ is the plane $(\boldsymbol V,0)^{\Tr}$,
$\xxi^{\Tr}\boldsymbol V=0$, by \eqref{eq:s2-characteristic-families},
so by \eqref{eq:kernel-preservation} the pressure entry
$\boldsymbol W_N(\boldsymbol\zeta)\boldsymbol v$ of
$\sigma_{q,N}(\boldsymbol\zeta)(\boldsymbol v,0)^{\Tr}$ vanishes for
every $\boldsymbol v\perp\xxi$. Then \eqref{eq:core-harmonic-root}
applies to $\boldsymbol W_N(-\uu\cdot\xxi,\xxi)$, a row of degree $N$ in
$\xxi$, and substitution of $\alpha=0$ in
\eqref{eq:harmonic-scalar-output} gives
$H_N(-\uu\cdot\xxi,\xxi)=-w_N\xxi^{\Tr}\xxi$.

\ref{ht-d} Division by $q_{\xxi}$ gives $K=q_{\xxi}H+K_{\rm rem}$. Since total
derivatives hold the covector fixed,
\[
 \xxi\cdot D_{\xx}K
 =q_{\xxi}\,\xxi\cdot D_{\xx}H+\xxi\cdot D_{\xx}K_{\rm rem}.
\]
The hypothesis therefore also holds for the remainder $K_{\rm rem}$. Replace
$K$ by this remainder, which preserves homogeneity when present, and
write $K=\sum_Jk_J\boldsymbol\zeta^J$
over the remainder monomials, whose degree in $\xi_3$ is at most one,
and let $m\ge0$ be the largest order on solutions of the coefficients $k_J$.
By Lemma~\ref{lem:null-cone} the hypothesis states that
$\xxi\cdot D_{\xx}K=\xxi^{\Tr}\xxi\,R$ on solutions for a covector polynomial
$R$ whose coefficients have order at most $m+1$ on solutions, since division
is linear in the coefficients.

Fix compatible derivative data, a characteristic pair $(\boldsymbol\theta,\boldsymbol z)$
from either family, and its order-$(m+1)$ increment, with
$\boldsymbol\theta=(\theta_t,\boldsymbol\theta_{\xx})$,
$\boldsymbol\theta_{\xx}\ne0$, and take the degree-$(m+1)$ response
\eqref{eq:s2-intrinsic-response} of both sides. Two covectors are in
play: $\boldsymbol\zeta=(\tau,\xxi)$ is the polynomial variable of
$K$ and is held fixed, while $\boldsymbol\theta$ is the covector of
the increment that probes the coefficients $k_J$; the condition
$\boldsymbol\theta_{\xx}\ne0$ makes $\xxi\cdot\boldsymbol\theta_{\xx}$
a nonzero linear form in $\xxi$. A function vanishing on
solutions has zero response, so the two responses agree. For an ambient
extension $\widehat k_J$ of order $m$, $\ell_{D_j\widehat k_J}=D_j\ell_{\widehat k_J}$,
so the response of $\xxi\cdot D_{\xx}k_J$ is
$(\xxi\cdot\boldsymbol\theta_{\xx})\sigma_{\widehat k_J,m}(\boldsymbol\theta)\boldsymbol z$,
where for $m=0$ the symbol is the differential
\eqref{eq:s2-zero-order-response}. Hence
\begin{equation}\label{eq:core-directional-descent}
 (\xxi\cdot\boldsymbol\theta_{\xx})
 \sum_J\sigma_{\widehat k_J,m}(\boldsymbol\theta)\boldsymbol z\,\boldsymbol\zeta^J
 =\xxi^{\Tr}\xxi\sum_{J'}\rho_{J'}\boldsymbol\zeta^{J'} ,
\end{equation}
where $\rho_{J'}$ are the responses of the coefficients of $R$; both sides
are polynomials in the covector at the fixed derivative data and pair. By
Lemma~\ref{lem:null-cone} the prime $\xxi^{\Tr}\xxi$ is coprime to the
nonzero linear form $\xxi\cdot\boldsymbol\theta_{\xx}$, so it divides
the sum on the left. That sum contains only remainder monomials, and a
remainder divisible by $\xxi^{\Tr}\xxi$ is zero. Therefore
$\sigma_{\widehat k_J,m}(\boldsymbol\theta)\boldsymbol z=0$ for every
$J$ and every characteristic pair with $\boldsymbol\theta_{\xx}\ne0$.
The remaining harmonic pairs have $\boldsymbol\theta=(\theta_t,\boldsymbol0)$,
$\theta_t\ne0$, with the amplitude $(\boldsymbol 0,1)^{\Tr}$ of
\eqref{eq:s2-pure-time-kernel}; they are limits along the null cone of
harmonic pairs with $\boldsymbol\theta_{\xx}\ne0$, and the response is
polynomial in $\boldsymbol\theta$ and $\boldsymbol z$, so it vanishes
there as well. For $m\ge1$, by Lemma~\ref{lem:descent}(ii), which admits
complex-valued functions, every $k_J$ has order at most $m-1$ on a
smaller neighbourhood, and the hypothesis on $K$ is unchanged. Repetition ends at $m=0$, where by the order-zero
assertion of Lemma~\ref{lem:descent} every $k_J$ is a function of
$(t,\xx)$ alone.
\end{proof}

\subsection{Transverse test}\label{sec:transverse-test}

The transverse test evaluates the determining identity on a
short-wave test column whose phase is carried with the flow. Its
outcome is recorded as a matrix acting on the plane $\xxi^\perp$ of
transverse amplitudes. For a real spatial covector $\xxi\ne0$ put
\begin{equation}\label{eq:s4-5}
 \matP=\mathsfbi I-\dfrac{\xxi\xxi^{\Tr}}{q_{\xxi}},
\end{equation}
the orthogonal projection onto $\xxi^\perp$, with $q_{\xxi}=\xxi^{\Tr}\xxi$
as in \eqref{eq:s2-euler-symbol}. Let $q=(\boldsymbol r,r^0)$ be a pair of
differential functions of order at most $n\ge1$ on solutions, represented
by an ambient pair of the same order. Its velocity-input symbol and its
\emph{transverse response} are
\begin{equation}\label{eq:direct-velocity-symbol}
 (\sigma^u_{\boldsymbol r,n}(\boldsymbol\zeta))_{ij}
 =\sum_{|I|=n}\dfrac{\partial r_i}{\partial(D_Iu_j)}\boldsymbol\zeta^I,
\end{equation}
\begin{equation}\label{eq:s4-7}
 \matS=\matP\sigma^u_{\boldsymbol r,n}(-\uu\cdot\xxi,\xxi)\matP,
 \qquad \matS=\matP\matS\matP.
\end{equation}
Here $I$ ranges over space--time multi-indices, and $\xxi$ is an
independent real argument. For $\boldsymbol v\perp\xxi$, the
characteristic increment \eqref{eq:s2-characteristic-increments} at the
covector $(-\uu\cdot\xxi,\xxi)$ with amplitude $(\boldsymbol v,0)^{\Tr}$
changes the order-$n$ entries by
$\delta(D_t^aD_{\xx}^J\uu)=\boldsymbol v(-\uu\cdot\xxi)^a\xxi^J$ and
$\delta(D_t^aD_{\xx}^Jp)=0$, $a+|J|=n$, and keeps the lower entries
fixed. With $\boldsymbol d=(d_1,d_2,d_3)$ the spatial displacement from the point, its
spatial Taylor patterns are
\begin{equation}\label{eq:transverse-spatial-patterns}
 \delta\uu(\boldsymbol d)
 =\boldsymbol v\dfrac{(\xxi\cdot\boldsymbol d)^n}{n!},\qquad
 \delta\ww(\boldsymbol d)
 =(\xxi\times\boldsymbol v)
       \dfrac{(\xxi\cdot\boldsymbol d)^{n-1}}{(n-1)!}.
\end{equation}
For $n=1$ these are the transverse entries of
\eqref{eq:s2-first-order-increments}. By
\eqref{eq:s2-intrinsic-response}, the velocity entries of the
response of $q$ to this pair are
$\sigma^u_{\boldsymbol r,n}(-\uu\cdot\xxi,\xxi)\boldsymbol v$,
independent of the representative. When $Lq=0$ on solutions, the
response lies by Lemma~\ref{lem:kernel} in the kernel of
$\sigma_E(-\uu\cdot\xxi,\xxi)$, the plane $(\boldsymbol V,0)^{\Tr}$ with
$\xxi^{\Tr}\boldsymbol V=0$; the pressure response is then zero and the
velocity response is the column $\matS\boldsymbol v$ of
\eqref{eq:s4-7}.

For the harmonic family a constant covector is characteristic at
every set of derivative data, because the condition $q_{\xxi}=0$ does
not involve the fields. The transverse condition
$\alpha=\tau+\uu\cdot\xxi=0$ involves the velocity, so a constant
covector is transverse only where $\uu\cdot\xxi$ takes one value. The
covector is therefore taken as the gradient of a phase $\phi$ carried
with the flow, $Y\phi=0$: then
$\alpha=D_t\phi+\uu\cdot D_{\xx}\phi=Y\phi=0$ wherever the phase
equation holds. The amplitude $\boldsymbol v_0$ is chosen so that the
term $L(\boldsymbol v_0,0)^{\Tr}$ lies in the range of
$\sigma_E(\boldsymbol\zeta)$; in the degree-$n$ comparison of
Section~\ref{sec:sym} this term is then removed by contraction with
transverse columns, as $\sigma_E\boldsymbol c_1$ is removed in the
harmonic test. Choose a phase $\phi$ and a transverse velocity
amplitude $\boldsymbol v_0$ with
\begin{equation}\label{eq:core-transverse-leading}
 \begin{gathered}
 Y\phi=0,\qquad \xxi=D_{\xx}\phi\ne0,\qquad
 Y\xxi=-\matA^{\Tr}\xxi,\qquad \xxi^{\Tr}\boldsymbol v_0=0,\\
 Y\boldsymbol v_0=-\matA\boldsymbol v_0
       +2\xxi\dfrac{\xxi^{\Tr}\matA\boldsymbol v_0}{q_{\xxi}}.
 \end{gathered}
\end{equation}
The covector equation follows from the phase equation by spatial
differentiation, since $D_j(Y\phi)=Y\xi_j+(\matA^{\Tr}\xxi)_j$. The
amplitude equation states that $Y\boldsymbol v_0+\matA\boldsymbol v_0$
is parallel to $\xxi$, with the multiple
$2\xxi^{\Tr}\matA\boldsymbol v_0/q_{\xxi}$ determined by
$Y(\xxi^{\Tr}\boldsymbol v_0)=(Y\xxi)^{\Tr}\boldsymbol v_0
+\xxi^{\Tr}Y\boldsymbol v_0=0$; it is equivalent to
$\matP(Y\boldsymbol v_0+\matA\boldsymbol v_0)=0$ together with
$Y(\xxi^{\Tr}\boldsymbol v_0)=0$, so transversality is preserved along
the flow; these are the transport equations of short-wave
disturbances \citep{LifschitzHameiri1991,CraikCriminale1986}. A column
$\boldsymbol b\perp\xxi$ on the same phase with
$\matP(Y\boldsymbol b-\matA^{\Tr}\boldsymbol b)=0$ is called an
\emph{adjoint amplitude}: it satisfies the transport equation obtained
from the adjoint operator $M$ of \eqref{eq:finite-LM}, whose matrix
term is $-\matA^{\Tr}$, in the same way as $\boldsymbol v_0$ satisfies
the one obtained from $L$.

Only finitely many derivatives of the phase and amplitude at a single
point enter the symbol comparison, so the phase and amplitude
equations need to hold only through a finite Taylor order. The
following lemma provides such data with prescribed leading values.

\begin{lemma}[Transverse data]\label{lem:transverse-data}
Fix compatible derivative data at a point $(t_0,\xx_0)$ and an integer $k\ge1$.
For every real $\xxi_*\ne0$ and every $\boldsymbol v_*\perp\xxi_*$ there
are a phase $\phi$ and an amplitude $\boldsymbol v_0$, polynomial in
$(t,\xx)$, which satisfy \eqref{eq:core-transverse-leading} through
order $k$ at $(t_0,\xx_0)$ and take the values $D_{\xx}\phi=\xxi_*$,
$\boldsymbol v_0=\boldsymbol v_*$ there. Only finitely many entries of
the compatible derivative data enter their Taylor coefficients.
\end{lemma}
\begin{proof}
In the tilted coordinates \eqref{eq:s2-tilted-coordinates} write
$Y=\mathcal Y+b_*D_s$ as in \eqref{eq:s2-tangential-transport}, with
$b_*\ne0$. The phase and amplitude equations in
\eqref{eq:core-transverse-leading} determine the successive
$s$-derivatives of $\phi$ and $\boldsymbol v_0$ at the point from their
tangential derivatives and finitely many entries of the compatible
derivative data. Choose the
tangential data of $\phi$ linear, with $D_{y_0}\phi$ such that
$D_s\phi=-\mathcal Y\phi/b_*$ gives $D_{\xx}\phi=\xxi_*$ at the point,
and the tangential data of $\boldsymbol v_0$ orthogonal to
$D_{\xx}\phi$ on the hypersurface through the point, with value
$\boldsymbol v_*$. Impose the phase equation $Y\phi=0$ through
order $k+1$, which determines the required derivatives of $\phi$
through order $k+2$. One spatial differentiation then gives the
covector equation through order $k$. Impose the amplitude equation
through order $k$ as well. The identity
$Y(\xxi^{\Tr}\boldsymbol v_0)=0$ propagates transversality from the
tangential data. The Taylor polynomials with these coefficients have
the required properties.
\end{proof}

For a pair of order at most $n$, choose the Taylor order $k\ge n+1$
in Lemma~\ref{lem:transverse-data}. The transverse test is the column
\begin{equation}\label{eq:finite-transverse-test}
 \boldsymbol a=\mathfrak e\binom{\boldsymbol v_0}{0},\qquad
 \mathfrak e=\exp(\kappa\phi),\qquad
 \boldsymbol\zeta=(-\uu\cdot\xxi,\xxi).
\end{equation}
By the product rule, the coefficient of $\kappa$ in
$\mathfrak e^{-1}L\boldsymbol a$ is
\[
 \sigma_E(D\phi)\binom{\boldsymbol v_0}{0}
 =\binom{(Y\phi)\boldsymbol v_0}{\xxi^{\Tr}\boldsymbol v_0},
 \qquad D\phi=(D_t\phi,\xxi).
\]
Both entries vanish through the chosen Taylor order. At the selected
point, together with all derivatives through order $n$, this gives
\begin{equation}\label{eq:finite-transverse-error}
 L\boldsymbol a=\mathfrak e
 \binom{Y\boldsymbol v_0+\matA\boldsymbol v_0}{\Div\boldsymbol v_0}.
\end{equation}
Thus $\deg_\kappa(\mathfrak e^{-1}TL\boldsymbol a)\le n$ at that point for every
operator $T$ with $\ord_{\rm op}T\le n$, and for every pair $q$ of
order at most $n$ the coefficient of $\kappa^n$ in
$\mathfrak e^{-1}\ell_q(\boldsymbol a)$ is
$\sigma_{q,n}(\boldsymbol\zeta)(\boldsymbol v_0,0)^{\Tr}$, whose
velocity entries are
$\sigma^u_{\boldsymbol r,n}(\boldsymbol\zeta)\boldsymbol v_0$.

The transverse response $\matS$ is a matrix function of the derivative
data and the covector, and the determining identity constrains its
material derivative along the transported phase. Let $\matS$ be any
matrix depending smoothly on $(t,\xx)$, finitely many entries of the
compatible derivative data and the real covector $\xxi\ne0$, with
$\matS\xxi=0$. Along the data \eqref{eq:core-transverse-leading}, the
chain rule and $Y\xxi=-\matA^{\Tr}\xxi$ give
\begin{equation}\label{eq:core-transverse-product}
 Y(\matS\boldsymbol v_0)
 =\left[Y\matS-(\matA^{\Tr}\xxi)\cdot\partial_{\xxi}\matS
                         -\matS\matA\right]\boldsymbol v_0,
\end{equation}
where $Y\matS$ holds $\xxi$ fixed and $\partial_{\xxi}$ differentiates
all covector dependence; the part of $Y\boldsymbol v_0$ parallel to
$\xxi$ is removed by $\matS\xxi=0$. A transverse matrix
$\matS=\matP\matS\matP$ satisfies the \emph{adjoint
transport law} when
\begin{equation}\label{eq:s4-8}
 \matP\left[Y\matS-(\matA^{\Tr}\xxi)\cdot\partial_{\xxi}\matS
                  -\matA^{\Tr}\matS-\matS\matA\right]\matP=0,
\end{equation}
with the same reading of $Y\matS$ and $\partial_{\xxi}$. By
\eqref{eq:core-transverse-product} and Lemma~\ref{lem:transverse-data},
\eqref{eq:s4-8} states that $\matS\boldsymbol v_0$ is an adjoint
amplitude for every velocity amplitude $\boldsymbol v_0$ of
\eqref{eq:core-transverse-leading}. The adjoint transport law is the
hypothesis of Proposition~\ref{thm:dependence} in Section~\ref{sec:5}; in
Section~\ref{sec:sym} it is derived for the matrix
$\mathsfbi C_{\xxi}\matS$ formed from the transverse response
\eqref{eq:s4-7} of a local symmetry.

The passage from the symmetry equation to the adjoint transport law
uses the following fact: the vorticity pattern $\xxi\times\boldsymbol v_0$
of \eqref{eq:transverse-spatial-patterns} is an adjoint amplitude
whenever $\boldsymbol v_0$ is a velocity amplitude,
\begin{equation}\label{eq:core-vorticity-amplitude}
 \begin{gathered}
 \matP(Y\boldsymbol v_0+\matA\boldsymbol v_0)=0,
 \quad \xxi^{\Tr}\boldsymbol v_0=0,
 \quad Y\xxi=-\matA^{\Tr}\xxi,\quad \tr\matA=0\\
 \Longrightarrow\quad
 \matP\left[Y(\xxi\times\boldsymbol v_0)
             -\matA^{\Tr}(\xxi\times\boldsymbol v_0)\right]=0.
 \end{gathered}
\end{equation}
To verify it, write $Y\boldsymbol v_0=-\matA\boldsymbol v_0+\rho\xxi$.
With the cross-product matrix of \eqref{eq:s1-2},
\[
 \begin{aligned}
 &Y(\xxi\times\boldsymbol v_0)-\matA^{\Tr}(\xxi\times\boldsymbol v_0)\\
 &\quad=-\mathsfbi C_{\matA^{\Tr}\xxi}\boldsymbol v_0
  -\left(\mathsfbi C_{\xxi}\matA+\matA^{\Tr}\mathsfbi C_{\xxi}\right)\boldsymbol v_0\\
 &\quad=\mathsfbi C_{(\matA-\matA^{\Tr})\xxi}\boldsymbol v_0
 =(\ww\times\xxi)\times\boldsymbol v_0
 =\xxi(\ww\cdot\boldsymbol v_0),
 \end{aligned}
\]
by the identity $\matA^{\Tr}\mathsfbi C_{\xxi}+\mathsfbi C_{\xxi}\matA
=-\mathsfbi C_{\matA\xxi}$ for $\tr\matA=0$, which is the derivative at
the identity of
$\mathsfbi M^{\Tr}\mathsfbi C_{\mathsfbi M\xxi}\mathsfbi M=(\det\mathsfbi M)\mathsfbi C_{\xxi}$
along $\mathsfbi M=\mathsfbi I+\epsilon\matA$ and is used again as
\eqref{eq:cross-product-covariance} in Section~\ref{sec:5},
together with $\matA-\matA^{\Tr}=\mathsfbi C_{\ww}$ and
$\xxi\cdot\boldsymbol v_0=0$. The projection of
$\xxi(\ww\cdot\boldsymbol v_0)$ is zero.

\section{Transported transverse matrices}\label{sec:5}

The transverse test of Section~\ref{sec:transverse-test} produces,
for a local symmetry, a transverse matrix that satisfies the adjoint
transport law \eqref{eq:s4-8}. In this section every
transverse matrix $\matS$ with that property is determined on a neighbourhood of
compatible derivative data in jet space. The argument changes the
velocity and pressure Taylor coefficients at the point independently,
in the directions the Euler equations permit, and reads off from
\eqref{eq:s4-8} which arguments $\matS$ can depend on and what form
it must have. The result, Proposition~\ref{thm:dependence}, is that
$\matS$ is a function of the single scalar
\begin{equation}\label{eq:core-vorticity-directional}
 \chi=\ww\cdot\xxi,
\end{equation}
the vorticity component along $\xxi$, times the cross-product matrix
$\mathsfbi C_{\xxi}$. In Section~\ref{sec:sym} the proposition is
applied to the matrix $\mathsfbi C_{\xxi}\matS$ formed from the
transverse response \eqref{eq:s4-7} of a local symmetry, and in
Section~\ref{sec:8} to the matrix $f\mathsfbi C_{\xxi}$ formed from a
material scalar invariant $f$. Here
$\boldsymbol d=(d_1,d_2,d_3)$ denotes spatial displacement from the
selected point.

\subsection{Independent Taylor data}

The Euler equations determine the time derivatives of the fields at a
point from their spatial Taylor coefficients, except for a harmonic
part of the pressure at each time-derivative level. These harmonic
parts and the divergence-free velocity polynomial are the independent
data that are varied below. At fixed $(t,\xx)$, let $\uu_q=D_t^q\uu$
and $p_q=D_t^qp$ denote the spatial Taylor polynomials through the
finite orders needed. Their coefficients are determined by the values
of the corresponding spatial derivatives at the selected point.
Differentiation in $\boldsymbol d$ acts on these polynomials with their
coefficients held fixed. The ordinary gradient, Jacobian and Laplacian are
\[
 (\nabla_{\boldsymbol d}f)_i=(D_{\boldsymbol d}f)_i
   =\frac{\partial f}{\partial d_i},\qquad
 (D_{\boldsymbol d}\boldsymbol a)_{ij}
   =\frac{\partial a_i}{\partial d_j},\qquad
 \Delta_{\boldsymbol d}=\sum_{i=1}^3\frac{\partial^2}{\partial d_i^2}.
\]
The corresponding ordinary divergence and curl are denoted by
$\operatorname{div}_{\boldsymbol d}$ and $\curl_{\boldsymbol d}$.
The divergence of the momentum equation and its successive time
derivatives give
\begin{equation}\label{eq:s5-euler-coordinate-recursion}
 \begin{aligned}
 p_q&=p^{\rm h}_q-\Delta_{\boldsymbol d}^{-1}\sum_{j=0}^q\binom qj
       \tr\left[(D_{\boldsymbol d}\uu_j)(D_{\boldsymbol d}\uu_{q-j})\right],\\
 \uu_{q+1}&=-\sum_{j=0}^q\binom qj
       (\uu_j\cdot D_{\boldsymbol d})\uu_{q-j}-D_{\boldsymbol d}p_q,
 \qquad D_tp^{\rm h}_q=p^{\rm h}_{q+1}.
 \end{aligned}
\end{equation}
Here $\Delta_{\boldsymbol d}^{-1}$ is a fixed linear right inverse of the polynomial
Laplacian, raising homogeneous degree by two. Its existence follows
from a finite-dimensional calculation: on homogeneous polynomials the
Laplacian is adjoint, under
$\langle\boldsymbol d^I,\boldsymbol d^J\rangle=I!\,\delta_{IJ}$, to
the injective multiplication by $|\boldsymbol d|^2$, and is therefore
surjective. The independent
coefficients are those of the divergence-free velocity $\uu_0$ and of the harmonic polynomials $p^{\rm h}_q$ of
\emph{level} $q$: the recursion reconstructs all time derivatives
from them, with each $\uu_{q+1}$ divergence-free by the first line.
Conversely, for compatible derivative data, the divergence of the $q$-th time derivative of
the momentum equation determines $\Delta_{\boldsymbol d} p_q$, so that $p^{\rm h}_q$
is harmonic. Each $\uu_{q+1}$ uses pressure levels at most $q$.
In a spatial derivative of $p_q^{\rm h}$, commuting a derivative in
$\boldsymbol d$ past the fixed inverse in \eqref{eq:s5-euler-coordinate-recursion}
adds a polynomial whose Laplacian vanishes because
$\Delta_{\boldsymbol d}$ commutes with spatial differentiation.
This correction depends only on $\uu_0,\ldots,\uu_q$, so no higher
pressure level occurs. Time differentiation commutes with the fixed
inverse, giving the last identity in \eqref{eq:s5-euler-coordinate-recursion}. These
coefficients form a second set of independent coordinates on the
compatible derivative data, equivalent to the tangential coordinates
of Lemma~\ref{lem:free}; they are used here because changes of a
single level $q$, or of the velocity polynomial alone, can be made
independently.

Let $V_r$ be the space of divergence-free homogeneous vector
polynomials of degree $r\ge0$. Curl-free divergence-free polynomials
are exactly harmonic gradients, and curl maps $V_r$ onto $V_{r-1}$
for $r\ge1$: for $\boldsymbol w\in V_{r-1}$,
\begin{equation}\label{eq:s5-curl-exact-sequence}
 \boldsymbol v=-\dfrac{\boldsymbol d\times\boldsymbol w}{r+1}
                      +\nabla_{\boldsymbol d}\varphi,\qquad
 \Delta_{\boldsymbol d}\varphi=\dfrac{\operatorname{div}_{\boldsymbol d}(\boldsymbol d\times\boldsymbol w)}{r+1}
 \quad\Longrightarrow\quad
 \curl_{\boldsymbol d}\boldsymbol v=\boldsymbol w,\quad \operatorname{div}_{\boldsymbol d}\boldsymbol v=0.
\end{equation}
These two facts are used to replace velocity coefficients by vorticity
coefficients in Step~1 below and to construct velocity changes with
prescribed vorticity changes in Step~4.

\subsection{Dependence and matrix form}

\begin{proposition}[Classification of a transported response]\label{thm:dependence}\label{lem:covariance}
Let $\matS=\matP\matS\matP$ depend smoothly on $(t,\xx)$, finitely
many compatible Euler derivatives, and every real $\xxi\ne0$ over
the same neighbourhood in jet space. If it satisfies \eqref{eq:s4-8}, then, on a
smaller neighbourhood in jet space and for every real $\xxi\ne0$,
\begin{equation}\label{eq:core-dependence-conclusion}
 \matS=F(\chi)\,\mathsfbi C_{\xxi},\qquad \chi=\ww\cdot\xxi,
\end{equation}
with $F$ a smooth function of one real variable. In particular, if
$\matS$ is homogeneous of integer degree $k\ge1$ under all nonzero
real rescalings of $\xxi$, then
\begin{equation}\label{eq:covariance-scalar-form}
 \matS=c\chi^{k-1}\mathsfbi C_{\xxi},
\end{equation}
where $c$ is locally constant.
\end{proposition}

The covector ranges over every real $\xxi\ne0$ at each choice of
compatible derivative data in the same neighbourhood. This domain
hypothesis permits integration over the planes $\ww\cdot\xxi=\chi$,
with the origin excluded, in Step~6.

The proof has six steps. Steps~1 and~4 remove the dependence of
$\matS$ on the pressure data, the velocity value, the position and
the vorticity derivatives of positive order, leaving a function of
$(t,\ww,\xxi)$. Steps~2 and~3 establish, for every divergence-free
polynomial change of the velocity, a first-order differential
equation for $\matS$ that encodes the transport law; this is the
instrument of Steps~4 and~5. Step~5 determines the matrix form
$\beta\mathsfbi C_{\xxi}$, and Step~6 reduces the scalar $\beta$ to a
function of $\chi$.

\begin{proof}
\emph{1. Pressure, velocity and position.}
Express the arguments of $\matS$ using the independent Taylor
coefficients in \eqref{eq:s5-euler-coordinate-recursion}. When pressure
arguments occur, let $q_*$ be the highest harmonic pressure level on
which $\matS$ depends.
The free
coefficients of $p^{\rm h}_{q_*+1}$ occur in \eqref{eq:s4-8} only
through $D_t\matS$, multiplied by the derivatives of $\matS$ in the
$p^{\rm h}_{q_*}$ coefficients. These matrix derivatives are transverse:
differentiating $\matS=\matP\matS\matP$ preserves the two projections
because $\matP$ depends only on the fixed covector. Consequently the
projections in \eqref{eq:s4-8} retain each coefficient, and independent
changes of $p^{\rm h}_{q_*+1}$ require it to vanish. Repeating the
argument for the lower levels removes every pressure argument, so
$\matS=\matS(t,\xx,\uu_0,\xxi)$. The polynomial $p^{\rm h}_0$ now enters
only through $\uu_1$ in $D_t\matS$, with coefficient
$-\mathrm d_{\uu_0}\matS\left(\nabla_{\boldsymbol d}p^{\rm h}_0\right)$ by
\eqref{eq:s5-euler-coordinate-recursion}; hence
$\mathrm d_{\uu_0}\matS$ vanishes on every harmonic-gradient change
of the velocity coefficients. A harmonic polynomial $p^{\rm h}_0$ of degree one has a constant
gradient, and the corresponding change of the velocity coefficients
alters only the value $\uu(0)$; hence $\matS$ is independent of the
velocity value $\uu(0)$. Two divergence-free velocity polynomials with
the same curl differ by a harmonic gradient, and by
\eqref{eq:s5-curl-exact-sequence} the curl maps the velocity
polynomials onto the divergence-free vorticity polynomials. Since
$\mathrm d_{\uu_0}\matS$ vanishes along harmonic gradients, $\matS$ is
constant on each class of velocity polynomials with a common curl, and
composition with a linear right inverse of the curl expresses it as a
smooth function of finitely many spatial vorticity derivatives.

Hereafter in this proof, $I$ is a spatial multi-index and evaluation
at $0$ means $\boldsymbol d=0$. Write $\partial_{\boldsymbol d}^I$
for the corresponding ordinary polynomial derivative.
The vorticity equation $Y\ww=\matA\ww$, the curl of the momentum
equation on solutions, gives
$D_t\ww=(\ww\cdot D_{\xx})\uu-(\uu\cdot D_{\xx})\ww$; applying
$\partial^I_{\boldsymbol d}$ at $\boldsymbol d=0$ and adding
$\uu(0)\cdot D_{\xx}$ gives the material rates of the vorticity
coefficients,
\begin{equation}\label{eq:core-vorticity-evolution}
 YD_I\ww=\partial_{\boldsymbol d}^I\left[(\ww\cdot D_{\boldsymbol d})\uu
       -\left((\uu-\uu(0))\cdot D_{\boldsymbol d}\right)\ww\right](0).
\end{equation}
The rates in \eqref{eq:core-vorticity-evolution}, $\matA$, $\xxi$ and
$\matP$ are independent of $\uu(0)$,
whose only remaining occurrence in \eqref{eq:s4-8} is
$\uu(0)\cdot\partial_{\xx}\matS$. Here $\partial_{\xx}$ holds the
vorticity derivatives and the covector fixed. Since $\uu(0)$ varies
independently, $\partial_{\xx}\matS=0$. Integrating the vanishing
partial derivatives on smaller convex coordinate neighbourhoods gives
\begin{equation}\label{eq:s5-vorticity-factorization}
 \matS=\matS(t,j^m_{\xx}\ww,\xxi)
\end{equation}
for some finite $m$, where $j^m_{\xx}\ww$ denotes the list of the
vorticity value and all its spatial derivatives through order $m$
at the selected point.

\emph{2. Harmonic-gradient changes.}
For a divergence-free polynomial $\boldsymbol g$ with
$\boldsymbol g(0)=0$, put $\mathsfbi G=D_{\boldsymbol d}\boldsymbol g(0)$.
To compare the transport law under changes of the velocity polynomial,
introduce the following differential operator on $\matS$. The
bracket in its argument is defined in the second line:
\begin{equation}\label{eq:s5-finite-jet-control}
 \begin{aligned}
 \mathscr L_{\boldsymbol g}\matS={}&
 \mathrm d_{\ww}\matS\left(j^m_{\xx}[\ww,\boldsymbol g]\right)
 -(\mathsfbi G^{\Tr}\xxi)\cdot\partial_{\xxi}\matS
 -\mathsfbi G\matS-\matS\mathsfbi G^{\Tr},\\
 [\boldsymbol g,\boldsymbol g']={}&
 (\boldsymbol g\cdot D_{\boldsymbol d})\boldsymbol g'
 -(\boldsymbol g'\cdot D_{\boldsymbol d})\boldsymbol g.
 \end{aligned}
\end{equation}
The differential $\mathrm d_{\ww}\matS$ is taken with respect to the
entire list $j^m_{\xx}\ww$ in \eqref{eq:s5-vorticity-factorization},
with $t$ and $\xxi$ fixed.
Since $\boldsymbol g(0)=0$, the argument of $\mathrm d_{\ww}\matS$
involves only that finite list of vorticity derivatives and
degrees at most $m+1$ of $\boldsymbol g$. For symmetric $\mathsfbi G$
the last two terms equal $-\mathsfbi G^{\Tr}\matS-\matS\mathsfbi G$.
For general $\mathsfbi G$, the displayed order is required for the
bracket identity established in Step~3.

A velocity change by a harmonic gradient $\epsilon\boldsymbol g$,
$\boldsymbol g(0)=0$, leaves every argument in
\eqref{eq:s5-vorticity-factorization} unchanged, since the vorticity
is unchanged; the bracket in \eqref{eq:core-vorticity-evolution}
changes by
$\epsilon[(\ww\cdot D_{\boldsymbol d})\boldsymbol g
-(\boldsymbol g\cdot D_{\boldsymbol d})\ww]=\epsilon[\ww,\boldsymbol g]$,
so the rates change by $\epsilon j^m_{\xx}[\ww,\boldsymbol g]$, and
$\matA=D_{\boldsymbol d}\uu(0)$ changes by $\epsilon\mathsfbi G$, with
$\mathsfbi G$ symmetric and trace-free.
Differentiation of \eqref{eq:s4-8} at $\epsilon=0$ gives
$\matP(\mathscr L_{\boldsymbol g}\matS)\matP=0$.
To recover the full matrix, apply the directional derivative formed
by the first two terms of \eqref{eq:s5-finite-jet-control} to
$\matS\xxi=0$ and $\xxi^{\Tr}\matS=0$, with $t$ fixed. Its covector
direction is $-\mathsfbi G^{\Tr}\xxi$; the product rule, together
with the last two matrix terms in that definition, gives
\begin{equation}\label{eq:covariance-transversality-check}
 (\mathscr L_{\boldsymbol g}\matS)\xxi=0,\qquad
 \xxi^{\Tr}(\mathscr L_{\boldsymbol g}\matS)=0.
\end{equation}
These cancellations hold for every matrix $\mathsfbi G$ with the factors
ordered as in \eqref{eq:s5-finite-jet-control}.
Thus $\mathscr L_{\boldsymbol g}\matS$ equals its double projection.
For the harmonic-gradient changes considered here that projection vanishes, giving
\begin{equation}\label{eq:s5-Dh}
 \mathscr L_{\boldsymbol g}\matS=0
 \quad\text{for every harmonic gradient $\boldsymbol g$ with }
 \boldsymbol g(0)=0.
\end{equation}

\emph{3. All divergence-free tests.}
Equation \eqref{eq:s5-Dh} is so far established for harmonic gradients
only. It is extended to every divergence-free polynomial by a Lie
algebra argument. The first two terms of \eqref{eq:s5-finite-jet-control} form a
vector field $X_{\boldsymbol g}$ on the arguments
$(j^m_{\xx}\ww,\xxi)$, linear in those arguments:
$\ww\mapsto[\ww,\boldsymbol g]$ on the vorticity coefficients and
$\xxi\mapsto-\mathsfbi G^{\Tr}\xxi$ on the covector. For linear
vector fields $V_{\mathsfbi A}(\boldsymbol y)=\mathsfbi A\boldsymbol y$
the commutator is
$[V_{\mathsfbi A},V_{\mathsfbi B}]=V_{\mathsfbi B\mathsfbi A-\mathsfbi A\mathsfbi B}$.
On the vorticity coefficients this gives
$\ww\mapsto[[\ww,\boldsymbol g],\boldsymbol g']-[[\ww,\boldsymbol g'],\boldsymbol g]
=[\ww,[\boldsymbol g,\boldsymbol g']]$ by the Jacobi identity, and
on the covector
$\xxi\mapsto(\mathsfbi G'^{\Tr}\mathsfbi G^{\Tr}-\mathsfbi G^{\Tr}\mathsfbi G'^{\Tr})\xxi
=-\bigl(D_{\boldsymbol d}[\boldsymbol g,\boldsymbol g'](0)\bigr)^{\Tr}\xxi$,
since the linear part of the bracket is
\begin{equation}\label{eq:core-bracket-definitions}
 D_{\boldsymbol d}[\boldsymbol g,\boldsymbol g'](0)
 =\mathsfbi G'\mathsfbi G-\mathsfbi G\mathsfbi G',
 \qquad\text{and hence}\qquad
 [\mathscr L_{\boldsymbol g},\mathscr L_{\boldsymbol g'}]
 =\mathscr L_{[\boldsymbol g,\boldsymbol g']}.
\end{equation}
Indeed, $[X_{\boldsymbol g},X_{\boldsymbol g'}]=X_{[\boldsymbol g,\boldsymbol g']}$
by the two computations above. The last two terms of
\eqref{eq:s5-finite-jet-control} form $-\mathrm R_{\mathsfbi G}\matS$
with $\mathrm R_{\mathsfbi G}\matS=\mathsfbi G\matS+\matS\mathsfbi G^{\Tr}$;
the constant matrices commute with every $X_{\boldsymbol g'}$, and
$[\mathrm R_{\mathsfbi G},\mathrm R_{\mathsfbi G'}]
=\mathrm R_{\mathsfbi G\mathsfbi G'-\mathsfbi G'\mathsfbi G}$. Therefore
\[
 [\mathscr L_{\boldsymbol g},\mathscr L_{\boldsymbol g'}]
 =X_{[\boldsymbol g,\boldsymbol g']}+\mathrm R_{\mathsfbi G\mathsfbi G'-\mathsfbi G'\mathsfbi G}
 =X_{[\boldsymbol g,\boldsymbol g']}-\mathrm R_{D_{\boldsymbol d}[\boldsymbol g,\boldsymbol g'](0)}
 =\mathscr L_{[\boldsymbol g,\boldsymbol g']}.
\]
By \eqref{eq:core-bracket-definitions}, the divergence-free
polynomials $\boldsymbol g$ with $\boldsymbol g(0)=0$ satisfying
\eqref{eq:s5-Dh} are closed under linear combinations and brackets,
and therefore form a Lie algebra; its members are called \emph{tests}.
They include the gradients $\mathsfbi G\boldsymbol d$ of quadratic
harmonic polynomials, $\mathsfbi G$ symmetric and trace-free, and
hence $[\mathsfbi G\boldsymbol d,\mathsfbi G'\boldsymbol d]
 =(\mathsfbi G'\mathsfbi G-\mathsfbi G\mathsfbi G')\boldsymbol d$,
which for diagonal $\mathsfbi G$ and
$\mathsfbi G'=\mathsfbi E_{ij}+\mathsfbi E_{ji}$ equals
$(G_{jj}-G_{ii})(\mathsfbi E_{ij}-\mathsfbi E_{ji})\boldsymbol d$; so every
trace-free linear field is a test. Here $\mathsfbi E_{ij}$ has one
nonzero entry, equal to one, in position $(i,j)$.

Let $W_r\subset V_r$ be the tests homogeneous of degree $r\ge1$.
Closure under brackets with linear tests means that $W_r$ is preserved
by each of the operators
\begin{equation}\label{eq:s5-sl3-action}
 \begin{aligned}
 \mathscr A_{\mathsfbi G}\boldsymbol g
 &=\mathsfbi G\boldsymbol g
   -(\mathsfbi G\boldsymbol d\cdot D_{\boldsymbol d})\boldsymbol g
   =-[\mathsfbi G\boldsymbol d,\boldsymbol g]\\
 &=\dfrac{\mathrm d}{\mathrm d\epsilon}\Big|_{\epsilon=0}
  e^{\epsilon\mathsfbi G}\boldsymbol g\left(e^{-\epsilon\mathsfbi G}\boldsymbol d\right)
 ,\qquad \tr\mathsfbi G=0.
 \end{aligned}
\end{equation}
These operators describe infinitesimal volume-preserving linear changes
of coordinates. They preserve divergence-free fields and satisfy
\[
 [\mathscr A_{\mathsfbi G},\mathscr A_{\mathsfbi H}]
 =\mathscr A_{\mathsfbi G\mathsfbi H-\mathsfbi H\mathsfbi G}.
\]
They therefore give a representation of the trace-free matrix algebra
$\mathfrak{sl}(3,\R)$ on $V_r$. A subspace preserved by all these
operators is invariant. The representation is irreducible if its only
invariant subspaces are zero and $V_r$. It suffices to establish this
property: $W_r$ is invariant and nonzero, since it contains harmonic
gradients of degree $r$.

Allow complex coefficients in $V_r$ and consider the two elementary
matrices $\mathsfbi E_{12}$ and $\mathsfbi E_{23}$. Their operators are
\[
 \mathscr A_{\mathsfbi E_{ij}}\boldsymbol g
 =g_j\boldsymbol e_i-d_j\partial_i\boldsymbol g,
 \qquad \partial_i=\partial/\partial d_i.
\]
Every nonzero invariant subspace contains a vector annihilated by both.
Indeed, for $\mathsfbi D=\diag(1,0,-1)$ the operator
$\mathscr A_{\mathsfbi D}$ is diagonal on the monomials
$\boldsymbol d^I\boldsymbol e_i$, with eigenvalue $D_{ii}-(I_1-I_3)$,
so an invariant subspace is the direct sum of its intersections with
the eigenspaces. By
$[\mathscr A_{\mathsfbi D},\mathscr A_{\mathsfbi E_{12}}]=\mathscr A_{\mathsfbi E_{12}}$
and
$[\mathscr A_{\mathsfbi D},\mathscr A_{\mathsfbi E_{23}}]=\mathscr A_{\mathsfbi E_{23}}$,
each of the two operators raises the eigenvalue by one, so a vector of
the largest eigenvalue occurring in the subspace is annihilated by
both. Their joint kernel can be computed directly:
\[
 \begin{gathered}
 \partial_1g_2=\partial_1g_3=\partial_2g_1=\partial_2g_3=0,\\
 d_2\partial_1g_1=g_2,\quad d_3\partial_2g_2=g_3,\\
 \boldsymbol g=a_1\boldsymbol e_1d_3^r+a_2\boldsymbol d\,d_3^{r-1},
 \qquad \operatorname{div}_{\boldsymbol d}\boldsymbol g=(r+2)a_2\,d_3^{r-1}=0.
 \end{gathered}
\]
The joint kernel is consequently the single line
$\C\boldsymbol e_1d_3^r$. By complete reducibility, every
finite-dimensional complex representation of $\mathfrak{sl}(3,\C)$ is
a direct sum of irreducible invariant subspaces
\citep[Theorem~9.19]{FultonHarris1991}, applied here to the
complex-linear extension of the representation to $\mathfrak{sl}(3,\C)$
on $V_r\otimes\C$. Each nonzero summand would contribute a vector to
this joint kernel, and vectors from distinct summands are linearly
independent, so there can be only one summand. Thus $V_r\otimes\C$ is irreducible, and complexification
excludes a nonzero proper real invariant subspace of $V_r$ as well.
It follows that $W_r=V_r$: every divergence-free polynomial vanishing
at zero satisfies \eqref{eq:s5-Dh}.

\emph{4. Vorticity derivatives.}
The tests of degree at least two will now be used to change the
positive-order vorticity derivatives while keeping $\ww(0)$ fixed.
For $\ww(0)\ne0$, the map
$\boldsymbol g\mapsto j^m_{\xx}[\ww,\boldsymbol g]$ on
$V_2\oplus\cdots\oplus V_{m+1}$ is onto the divergence-free
vorticity increments with zero value. To establish this assertion,
order the homogeneous coefficients by degree. In the bracket
\eqref{eq:s5-finite-jet-control}, the term
$(\boldsymbol g_{r+1}\cdot D_{\boldsymbol d})\ww$ has no component of
degree $r$. In degree $r$, the only contribution from
$\boldsymbol g_{r+1}$ is therefore
$(\ww(0)\cdot D_{\boldsymbol d})\boldsymbol g_{r+1}$, and the other
contributions involve $\boldsymbol g$ through degree $r$.
This diagonal map $V_{r+1}\to V_r$ comes from vortex stretching in
\eqref{eq:core-vorticity-evolution}. In planar Euler flow this term
vanishes and the transport of scalar vorticity gives
\eqref{eq:intro-enstrophy}.
For $\ww(0)\ne0$ the diagonal map is onto. Indeed, after a rotation
$\ww(0)=|\ww(0)|\boldsymbol e_3$, for
$\boldsymbol w\in V_r$ and $\boldsymbol v\in V_{r+1}$ with
$\curl_{\boldsymbol d}\boldsymbol v=\boldsymbol w$ from
\eqref{eq:s5-curl-exact-sequence},
\begin{equation}\label{eq:core-solenoidal-primitive}
 \boldsymbol g(\boldsymbol d)
 =\curl_{\boldsymbol d}\int_0^{d_3}\boldsymbol v(d_1,d_2,d_3')\dd d_3',
 \qquad
 \partial_3\boldsymbol g=\curl_{\boldsymbol d}\boldsymbol v=\boldsymbol w,
 \qquad \operatorname{div}_{\boldsymbol d}\boldsymbol g=0,
\end{equation}
lies in $V_{r+1}$. Division by $|\ww(0)|$ gives a preimage under the
diagonal map, and successive solution in increasing degree gives the
asserted surjectivity. For these tests $\mathsfbi G=0$, so only the
vorticity differential remains in \eqref{eq:s5-Dh}. It must vanish
in every positive-order vorticity direction, with $t$, $\ww(0)$ and
$\xxi$ fixed. Continuity extends this vanishing to $\ww(0)=0$.
Integration along line segments in smaller coordinate neighbourhoods yields
$\matS=\matS(t,\ww,\xxi)$.

\emph{5. Matrix form.}
For the linear tests $\boldsymbol g=\mathsfbi G\boldsymbol d$,
$j^0_{\xx}[\ww,\boldsymbol g]=\mathsfbi G\ww$ and \eqref{eq:s5-Dh}
reads
\begin{equation}\label{eq:s5-linear-controls}
 \left[(\mathsfbi G\ww)\cdot\partial_{\ww}
  -(\mathsfbi G^{\Tr}\xxi)\cdot\partial_{\xxi}\right]\matS
 =\mathsfbi G\matS+\matS\mathsfbi G^{\Tr},\qquad \tr\mathsfbi G=0.
\end{equation}
To obtain algebraic restrictions from \eqref{eq:s5-linear-controls},
choose matrices $\mathsfbi G$ for which both argument changes vanish.
When $\chi\ne0$, $\R^3=\R\ww\oplus\xxi^\perp$; for $\mathsfbi G$ with
$\mathsfbi G\ww=0$ acting on $\xxi^\perp$ by a trace-free two-by-two
matrix $\mathsfbi H$ in an orthonormal basis, $\tr\mathsfbi G=0$,
$\mathsfbi G^{\Tr}\xxi=0$, and the matrix $\mathsfbi Z$ of $\matS$
on $\xxi^\perp$ satisfies
$\mathsfbi H\mathsfbi Z+\mathsfbi Z\mathsfbi H^{\Tr}=0$: in an
orthonormal basis $\boldsymbol e_a$ of $\xxi^\perp$ the matrix of
$\matP\mathsfbi G^{\Tr}\matP$ is $\mathsfbi H^{\Tr}$, because
$\boldsymbol e_a\cdot\mathsfbi G^{\Tr}\boldsymbol e_b
=(\mathsfbi G\boldsymbol e_a)\cdot\boldsymbol e_b=H_{ba}$, and
$\matS=\matS\matP$. Here
$\mathsfbi H=\diag(1,-1)$ gives $Z_{11}=Z_{22}=0$ and
$\mathsfbi H=\mathsfbi E_{12}$ gives $Z_{12}=-Z_{21}$, so $\matS$ is
skew-symmetric, also at $\chi=0$ by continuity, and since
$\matS\xxi=0$,
\begin{equation}\label{eq:s6-1}
 \matS=\beta\mathsfbi C_{\xxi},\qquad
 \beta=\dfrac{\tr(\mathsfbi C_{\xxi}^{\Tr}\matS)}{2q_{\xxi}},
\end{equation}
with $\beta$ smooth for $\xxi\ne0$. Substitution of this form into
\eqref{eq:s5-linear-controls}, using the identity of
Section~\ref{sec:transverse-test} in the form
\begin{equation}\label{eq:cross-product-covariance}
 \mathsfbi G\mathsfbi C_{\xxi}+\mathsfbi C_{\xxi}\mathsfbi G^{\Tr}
 =-\mathsfbi C_{\mathsfbi G^{\Tr}\xxi},\qquad \tr\mathsfbi G=0,
\end{equation}
gives a scalar equation for the remaining dependence:
\begin{equation}\label{eq:core-beta-invariance}
 \left[(\mathsfbi G\ww)\cdot\partial_{\ww}
  -(\mathsfbi G^{\Tr}\xxi)\cdot\partial_{\xxi}\right]\beta=0,
 \qquad \tr\mathsfbi G=0.
\end{equation}

\emph{6. Covector, vorticity and time.}
Fix $t$ and $\ww\ne0$. For $\boldsymbol w\perp\ww$ and a unit vector
$\boldsymbol n\perp\ww,\xxi$, the matrix
\[
 \mathsfbi G=\dfrac{\xxi\boldsymbol w^{\Tr}
   -(\boldsymbol w\cdot\xxi)\boldsymbol n\boldsymbol n^{\Tr}}{q_{\xxi}}
\]
is trace-free, with $\mathsfbi G\ww=0$ and
$\mathsfbi G^{\Tr}\xxi=\boldsymbol w$, so
\eqref{eq:core-beta-invariance} gives
$\boldsymbol w\cdot\partial_{\xxi}\beta=0$: $\beta$ is locally
constant on each plane $\ww\cdot\xxi=\chi$. These planes, including
the punctured one at $\chi=0$, are connected and lie in the domain,
so $\beta=F(t,\ww,\chi)$. To verify smoothness, choose near any
nonzero vorticity value a smooth unit vector
$\boldsymbol n_0(\ww)\perp\ww$ and restrict $\beta$ to the line
\[
 \xxi=\boldsymbol n_0(\ww)+\dfrac{\chi\ww}{|\ww|^2}.
\]
This line avoids the origin for every real $\chi$. The restrictions
from different choices agree because $\beta$ is constant on each plane.

It remains to remove the separate dependence on $\ww$ and $t$.
The joint changes of $\ww$ and $\xxi$ in
\eqref{eq:core-beta-invariance} leave $\chi$ unchanged, so
$(\mathsfbi G\ww)\cdot\partial_{\ww}F=0$; since for every
$\boldsymbol c\in\R^3$ the trace-free matrix
$\mathsfbi G=\boldsymbol c\ww^{\Tr}/|\ww|^2
-(\boldsymbol c\cdot\ww)|\ww|^{-2}\boldsymbol n\boldsymbol n^{\Tr}$,
with a unit vector $\boldsymbol n\perp\ww$, has
$\mathsfbi G\ww=\boldsymbol c$, the vector $\mathsfbi G\ww$ ranges
over $\R^3$, and $F$ is independent of $\ww$ on a punctured ball
around the base value, and by continuity $\beta=F(t,\chi)$, including
at $\ww=0$.

Finally, $Y\ww=\matA\ww$ on solutions, so
$Y\matS=\partial_t\matS+(\matA\ww)\cdot\partial_{\ww}\matS$ at fixed
$\xxi$. Since $\tr\matA=G=0$ on solutions, the choice
$\mathsfbi G=\matA$ in \eqref{eq:s5-linear-controls} is admissible,
and substitution of the resulting value of
$(\matA\ww)\cdot\partial_{\ww}\matS-(\matA^{\Tr}\xxi)\cdot\partial_{\xxi}\matS$
into the transport law \eqref{eq:s4-8}, used here for the first time
since Step~1, gives
\[
 \matP\left[\partial_t\matS+\mathsfbi C_{\ww}\matS
                         -\matS\mathsfbi C_{\ww}\right]\matP=0,
\]
since $\matA-\matA^{\Tr}=\mathsfbi C_{\ww}$. The commutator has zero
double projection, because
\begin{equation}\label{eq:projected-cross-commutator}
 [\mathsfbi C_{\ww},\mathsfbi C_{\xxi}]\boldsymbol v
 =(\ww\times\xxi)\times\boldsymbol v
 =\xxi(\ww\cdot\boldsymbol v),\qquad \xxi\cdot\boldsymbol v=0,
\end{equation}
so $F_t(t,\chi)\matP\mathsfbi C_{\xxi}\matP=F_t(t,\chi)\mathsfbi C_{\xxi}=0$
and $F=F(\chi)$, which is \eqref{eq:core-dependence-conclusion}.

If $\matS$ is homogeneous of degree $k$ in $\xxi$, then
$F(\lambda\chi)=\lambda^{k-1}F(\chi)$ for real $\lambda\ne0$, since
$\chi$ and $\mathsfbi C_{\xxi}$ have degree one. Setting $\chi=1$ and
then renaming $\lambda$ gives $F(\chi)=F(1)\chi^{k-1}$ for $\chi\ne0$;
continuity gives the same formula at zero. This is
\eqref{eq:covariance-scalar-form} with $c=F(1)$.
\end{proof}

\section{Local symmetries}\label{sec:sym}

Theorem~\ref{thm:symmetries} is proved in this section. Let
$q=(\boldsymbol r,r^0)$ be a local symmetry of order at most $n\ge1$,
so that $Lq=0$ on solutions by \eqref{eq:sym-determining}.
Lemma~\ref{lem:kernel} gives the determining identity
\eqref{eq:kernel-identity}, whose coefficient term enters the
comparisons only at $n=1$, and the kernel preservation
\eqref{eq:kernel-preservation} for the degree-$n$ symbol
$\sigma_{q,n}$. The proof proceeds in three stages. The harmonic
response of $q$ is determined first
(Section~\ref{sec:sym-harmonic}): it is a field-independent polynomial
times the kernel column. The transverse response is determined
second (Section~\ref{sec:sym-reduction}) by converting its transport
equation into the adjoint transport law and applying
Proposition~\ref{thm:dependence}. The two responses are polynomials in
the covector, and their comparison at the complex covectors
satisfying both $\alpha=0$ and $q_{\xxi}=0$ forces both to vanish
for $n\ge2$, so that Lemma~\ref{lem:descent} lowers the order to one.
The first-order equations are then solved and give the point family
\eqref{eq:sym-family} (Section~\ref{sec:sym-first}).

\subsection{Harmonic response}\label{sec:sym-harmonic}

Translations show what a harmonic response can look like. For a
constant vector $\boldsymbol c$,
\begin{equation}\label{eq:sym-translation-symbol}
 q=-(\boldsymbol c\cdot D_{\xx})(\uu,p),\qquad
 \sigma_{q,1}(\boldsymbol\zeta)\boldsymbol h_{\boldsymbol\zeta}
   =-(\boldsymbol c\cdot\xxi)\boldsymbol h_{\boldsymbol\zeta}:
\end{equation}
the kernel column is multiplied by a polynomial in the covector whose
coefficients do not depend on the fields. The following proposition
shows that every harmonic response has this form.

\begin{proposition}[Harmonic response of a symmetry]\label{prop:sym-harmonic}
Let $q$ be a local symmetry of order at most $N\ge1$. There is a
polynomial $\Phi_N(t,\xx,\tau,\xxi)$, homogeneous of degree $N$ in the
covector and independent of the fields and their derivatives, such that
\begin{equation}\label{eq:sym-harmonic-response}
 \sigma_{q,N}(\tau,\xxi)\binom{-\xxi}{\alpha}
 =\Phi_N(t,\xx,\tau,\xxi)\binom{-\xxi}{\alpha},
 \qquad \xxi^{\Tr}\xxi=0,\quad(\tau,\xxi)\ne0,
\end{equation}
including the values $\alpha=0$. At $\xxi=0$, the pure pressure
amplitude is multiplied by $\Phi_N(t,\xx,\tau,0)$.
\end{proposition}

\begin{proof}
Fix a constant complex covector $\boldsymbol\zeta=(\tau,\xxi)$ with
$\xxi^{\Tr}\xxi=0$, $\xxi\ne0$, work on a neighbourhood where
$\alpha\ne0$, and take the corrected test
$\boldsymbol a=\mathfrak e(\boldsymbol b_0+\kappa^{-1}\boldsymbol b_1)$
of Lemma~\ref{lem:harmonic-test}, Eq.~\eqref{eq:harmonic-corrected-test}.
By Lemma~\ref{lem:harmonic-test}(a), Eqs.~\eqref{eq:harmonic-leading-output}
and \eqref{eq:harmonic-scalar-output},
\begin{equation}\label{eq:sym-harmonic-leading-output}
 \begin{gathered}
 \ell_q(\boldsymbol a)=\mathfrak e(\kappa^N\boldsymbol c_0
 +\kappa^{N-1}\boldsymbol c_1+\cdots),\\
 \boldsymbol c_0=\sigma_{q,N}(\boldsymbol\zeta)\binom{-\xxi}{\alpha}
 =\boldsymbol h_{\boldsymbol\zeta}H_N,\qquad
 H_N=\alpha Z_N-\boldsymbol W_N\xxi,
 \end{gathered}
\end{equation}
where $(\boldsymbol W_N,Z_N)$ is the last row of
$\sigma_{q,N}(\boldsymbol\zeta)$ and $H_N$ is a scalar polynomial of
degree $N+1$ in the covector. The determining identity
\eqref{eq:kernel-identity} is now compared at degree $N$ in $\kappa$.
By the product rule \eqref{eq:harmonic-product-rule} and
Lemma~\ref{lem:harmonic-test}(a), the three terms of
\eqref{eq:kernel-identity} expand as
\[
 \begin{gathered}
 L\ell_q(\boldsymbol a)=\mathfrak e\left[\kappa^{N+1}\sigma_E(\boldsymbol\zeta)\boldsymbol c_0
  +\kappa^N\left(L\boldsymbol c_0+\sigma_E(\boldsymbol\zeta)\boldsymbol c_1\right)+\cdots\right],\\
 \deg_\kappa\left(\mathfrak e^{-1}TL\boldsymbol a\right)\le N-1,
 \qquad
 \mathfrak e^{-1}\binom{\mathcal C_L(\boldsymbol v)\boldsymbol r}{0}
 =\kappa\binom{\xxi(\xxi\cdot\boldsymbol r)}{0}+O(1),
 \end{gathered}
\]
where $\boldsymbol v$ is the velocity entry of the test column
$\boldsymbol a$. The degree-$N$ terms therefore give
$L\boldsymbol c_0+\sigma_E(\boldsymbol\zeta)\boldsymbol c_1
 =\delta_{N1}(\xxi(\xxi\cdot\boldsymbol r),0)^{\Tr}$.
Contraction with $\boldsymbol h_{\boldsymbol\zeta}^{\Tr}$ removes the
$\boldsymbol c_1$ term, because $\sigma_E$ is symmetric and
annihilates $\boldsymbol h_{\boldsymbol\zeta}$, and removes the right
side, because it contains the factor $\xxi^{\Tr}\xxi=0$. Hence
$\boldsymbol h_{\boldsymbol\zeta}^{\Tr}L(\boldsymbol h_{\boldsymbol\zeta}H_N)=0$,
and substitution into the contraction formula
\eqref{eq:harmonic-projection-product} gives
\begin{equation}\label{eq:sym-harmonic-transport}
 \begin{gathered}
 -\dfrac{2}{\alpha^2}\left(\alpha\,\xxi\cdot D_{\xx}H_N-\gamma H_N\right)=0,
 \\
 \text{that is,}\qquad
 \alpha\,\xxi\cdot D_{\xx}H_N=\gamma H_N\pmod{\xxi^{\Tr}\xxi},
 \end{gathered}
\end{equation}
with $\gamma=\xxi^{\Tr}\matA\xxi$ as defined in
\eqref{eq:harmonic-corrected-test}.
In the second line, the difference between the two sides is divisible by
$\xxi^{\Tr}\xxi$, which by Lemma~\ref{lem:null-cone} follows from
vanishing at every null covector with $\alpha\ne0$.

The term $\gamma H_N$ is removed by dividing out the factor $\alpha$.
At $\tau=-\uu\cdot\xxi$ with $\xxi$ real, the transverse kernel of
$\sigma_E$ consists of the columns $(\boldsymbol V,0)^{\Tr}$ with
$\xxi^{\Tr}\boldsymbol V=0$, so by the kernel preservation
\eqref{eq:kernel-preservation} the row
$\boldsymbol W_N(-\uu\cdot\xxi,\xxi)$ annihilates $\xxi^\perp$. By
Lemma~\ref{lem:harmonic-test}(c), Eq.~\eqref{eq:core-harmonic-root},
\begin{equation}\label{eq:sym-harmonic-root}
 \boldsymbol W_N(-\uu\cdot\xxi,\xxi)=w_N\xxi^{\Tr},\qquad
 H_N(-\uu\cdot\xxi,\xxi)=-w_N\xxi^{\Tr}\xxi,
\end{equation}
where $w_N$ is a scalar polynomial with smooth coefficients. The
difference $H_N(\tau,\xxi)-H_N(-\uu\cdot\xxi,\xxi)$ vanishes at the
root of the monic linear polynomial $\alpha$ in $\tau$, so division
defines a polynomial $K_N$ of degree $N$ with smooth coefficients:
\begin{equation}\label{eq:sym-division}
 H_N=\alpha K_N-w_N\,\xxi^{\Tr}\xxi,\qquad
 K_N=\dfrac{H_N(\tau,\xxi)-H_N(-\uu\cdot\xxi,\xxi)}{\tau+\uu\cdot\xxi}.
\end{equation}
Since $\nabla\alpha=\matA^{\Tr}\xxi$ by \eqref{eq:harmonic-fixed-parameters},
$\xxi\cdot D_{\xx}\alpha=\xxi^{\Tr}\matA^{\Tr}\xxi=\gamma$, and
$D_{\xx}(\xxi^{\Tr}\xxi)=0$; substitution of \eqref{eq:sym-division}
in \eqref{eq:sym-harmonic-transport} therefore gives
\[
 \begin{aligned}
 \alpha\,\xxi\cdot D_{\xx}H_N-\gamma H_N
 &=\alpha^2\,\xxi\cdot D_{\xx}K_N+\alpha(\xxi\cdot D_{\xx}\alpha-\gamma)K_N\\
 &\quad-\xxi^{\Tr}\xxi\left(\alpha\,\xxi\cdot D_{\xx}w_N-\gamma w_N\right)
 \equiv\alpha^2\,\xxi\cdot D_{\xx}K_N\pmod{\xxi^{\Tr}\xxi},
 \end{aligned}
\]
so that $\alpha^2\,\xxi\cdot D_{\xx}K_N\equiv0\pmod{\xxi^{\Tr}\xxi}$.
By Lemma~\ref{lem:null-cone}, the quadratic form $\xxi^{\Tr}\xxi$ is
prime and coprime to $\alpha$, so the factor $\alpha^2$ cancels:
\begin{equation}\label{eq:sym-K-transport}
 \xxi\cdot D_{\xx}K_N=0\pmod{\xxi^{\Tr}\xxi}.
\end{equation}
By Lemma~\ref{lem:harmonic-test}(d), on a smaller neighbourhood
$K_N\equiv\Phi_N(t,\xx,\tau,\xxi)\pmod{\xxi^{\Tr}\xxi}$, where
$\Phi_N$ is a polynomial independent of the fields and their
derivatives; it is homogeneous of degree $N$ because $H_N$ has degree
$N+1$, $\alpha$ has degree one, and the remainder modulo
$\xxi^{\Tr}\xxi$ preserves homogeneity.
Equations \eqref{eq:sym-harmonic-leading-output} and
\eqref{eq:sym-division} now give, on the null cone with $\alpha\ne0$,
\begin{equation}\label{eq:sym-harmonic-factorized-response}
 \sigma_{q,N}(\boldsymbol\zeta)\binom{-\xxi}{\alpha}
 =\boldsymbol h_{\boldsymbol\zeta}\,\alpha K_N
 =K_N\binom{-\xxi}{\alpha}
 =\Phi_N\binom{-\xxi}{\alpha}.
\end{equation}
This is \eqref{eq:sym-harmonic-response} for $\alpha\ne0$. At fixed
null $\xxi\ne0$ both sides of \eqref{eq:sym-harmonic-response} are
polynomials in $\tau$, so the identity extends to $\alpha=0$. Letting
$\xxi\to0$ along the null cone at fixed $\tau\ne0$ gives the statement
on the pressure line.
\end{proof}

\subsection{Orders above one}\label{sec:sym-reduction}

\begin{proposition}[Symmetries: reduction to first order]\label{thm:sym-reduction}
Every local symmetry has order at most one on solutions.
\end{proposition}

\begin{proof}
Let $q=(\boldsymbol r,r^0)$ have order at most $N\ge2$. Its transverse
response \eqref{eq:s4-7} is
\[
 \matS=\matP\sigma^u_{\boldsymbol r,N}(-\uu\cdot\xxi,\xxi)\matP,
\]
with the projector $\matP$
of \eqref{eq:s4-5} and the velocity-input symbol
$\sigma^u_{\boldsymbol r,N}$ of \eqref{eq:direct-velocity-symbol}.
Proposition~\ref{prop:sym-harmonic} gives the harmonic response of $q$
in terms of the field-independent polynomial $\Phi_N$. To reduce the
order of $q$, it remains to determine $\matS$ and compare the two
responses at the common complex covectors $\alpha=q_{\xxi}=0$.

\emph{Step 1: transverse transport.} Use the test
$\boldsymbol a=\mathfrak e(\boldsymbol v_0,0)^{\Tr}$ of
\eqref{eq:finite-transverse-test} with the phase and amplitude data
\eqref{eq:core-transverse-leading}; by Lemma~\ref{lem:transverse-data}
the phase and amplitude equations can be imposed through every Taylor
order used below. At the selected compatible derivative data, the
value of $\boldsymbol v_0$ is arbitrary in $\xxi^\perp$.
Write $T_N=\sigma_N(T)$ and
$\boldsymbol\zeta=(-\uu\cdot\xxi,\xxi)$. After division by
$\mathfrak e$, the determining identity \eqref{eq:kernel-identity} is a
polynomial identity in $\kappa$. The amplitude equation in
\eqref{eq:core-transverse-leading} gives
$Y\boldsymbol v_0+\matA\boldsymbol v_0=\rho\xxi$ for a scalar $\rho$,
so the remainder in \eqref{eq:finite-transverse-error} is
\[
 L(\boldsymbol v_0,0)^{\Tr}
 =\binom{\rho\xxi}{\Div\boldsymbol v_0}
 =\sigma_E(\boldsymbol\zeta)\boldsymbol c,\qquad
 \boldsymbol c=\binom{\xxi\Div\boldsymbol v_0/q_{\xxi}}{\rho},
\]
and $TL\boldsymbol a$ has degree at most $N$ in $\kappa$, with leading
coefficient $T_N\sigma_E\boldsymbol c$. The amplitude equation was
imposed for this reason: the remainder lies in the range of
$\sigma_E(\boldsymbol\zeta)$, and the symbol identity
\eqref{eq:kernel-preservation} converts $T_N\sigma_E\boldsymbol c$ into
$\sigma_E\sigma_{q,N}\boldsymbol c$, which transverse contractions
annihilate. Expand
$\ell_q(\boldsymbol a)=\mathfrak e(\kappa^N\boldsymbol y_0
 +\kappa^{N-1}\boldsymbol y_1+\cdots)$ with
$\boldsymbol y_0=\sigma_{q,N}(D\phi)(\boldsymbol v_0,0)^{\Tr}$.
To apply the first-order operator $L$ to this column, replace
$D_t\phi$ by $-\uu\cdot\xxi$ and $\boldsymbol v_0$ by
$\matP\boldsymbol v_0$. The discarded terms are smooth multiples of
$Y\phi$ and $\xxi^{\Tr}\boldsymbol v_0$, which vanish together with
their first derivatives by Lemma~\ref{lem:transverse-data}. By kernel
preservation \eqref{eq:kernel-preservation} and the definition
\eqref{eq:s4-7}, the resulting column is
$(\boldsymbol b,0)^{\Tr}$ with $\boldsymbol b=\matS\boldsymbol v_0$.
The two columns have the same value and first derivatives at the
selected point. In particular, the pressure entry of $\boldsymbol y_0$
has zero spatial gradient there.
Comparison at degree $N$ in \eqref{eq:kernel-identity} gives
$L\boldsymbol y_0+\sigma_E\boldsymbol y_1=T_N\sigma_E\boldsymbol c$;
the coefficient term of \eqref{eq:kernel-identity} has degree one,
below $N$. Contract with $\boldsymbol w=(\boldsymbol w_0,0)^{\Tr}$ for
arbitrary $\boldsymbol w_0\perp\xxi$. Symmetry of $\sigma_E$ and
$\sigma_E\boldsymbol w=0$ remove the $\boldsymbol y_1$ term. By the
full symbol identity \eqref{eq:kernel-preservation}, the contraction
of the right-hand side also vanishes,
\[
 \boldsymbol w^{\Tr}T_N\sigma_E\boldsymbol c
 =\boldsymbol w^{\Tr}\sigma_E\sigma_{q,N}\boldsymbol c
 =(\sigma_E\boldsymbol w)^{\Tr}\sigma_{q,N}\boldsymbol c=0.
\]
Hence $\matP(Y\boldsymbol b+\matA\boldsymbol b)=0$, and the chain rule
\eqref{eq:core-transverse-product} gives
\begin{equation}\label{eq:sym-transport}
 \matP\left[Y\matS-(\matA^{\Tr}\xxi)\cdot\partial_{\xxi}\matS
                  +\matA\matS-\matS\matA\right]\matP=0.
\end{equation}
For $N=1$ the coefficient term of \eqref{eq:kernel-identity}
contributes $-(\xxi\cdot\boldsymbol r)\boldsymbol v_0$ at degree one,
a transverse column retained by the projection. The homogeneous
transport equation \eqref{eq:sym-transport} therefore applies to
$N\ge2$; first-order symmetries require the separate calculation of
Section~\ref{sec:sym-first}.

\emph{Step 2: conversion to the adjoint transport law.}
Equation \eqref{eq:sym-transport} contains $+\matA\matS$, whereas
the adjoint transport law \eqref{eq:s4-8} contains
$-\matA^{\Tr}\matS$. The vorticity-amplitude identity
\eqref{eq:core-vorticity-amplitude} converts one into the other. The
column $\boldsymbol b=\matS\boldsymbol v_0$ is transverse and
satisfies $\matP(Y\boldsymbol b+\matA\boldsymbol b)=0$, the amplitude
equation of \eqref{eq:core-transverse-leading}. By
\eqref{eq:core-vorticity-amplitude}, the vorticity amplitude
$\xxi\times\boldsymbol b$ satisfies the adjoint amplitude equation
$\matP[Y(\xxi\times\boldsymbol b)-\matA^{\Tr}(\xxi\times\boldsymbol b)]=0$.
Put
\begin{equation}\label{eq:sym-sharp}
 \matS^\sharp=\mathsfbi C_{\xxi}\matS,\qquad
 \matS^\sharp\boldsymbol v_0=\xxi\times\boldsymbol b.
\end{equation}
Since $\matS^\sharp\xxi=0$, the chain rule
\eqref{eq:core-transverse-product} applies to $\matS^\sharp$, and the
adjoint amplitude equation for $\xxi\times\boldsymbol b$ becomes
\[
 \matP\left[Y\matS^\sharp-(\matA^{\Tr}\xxi)\cdot\partial_{\xxi}\matS^\sharp
   -\matA^{\Tr}\matS^\sharp-\matS^\sharp\matA\right]\boldsymbol v_0=0 .
\]
The value of $\boldsymbol v_0$ ranges over $\xxi^\perp$, so $\matS^\sharp$
satisfies the transport law \eqref{eq:s4-8}. By
$\xxi^{\Tr}\mathsfbi C_{\xxi}=0$ and $\matS\xxi=0$, the matrix
$\matS^\sharp=\matP\matS^\sharp\matP$ is transverse; it is smooth on
one coordinate neighbourhood of compatible derivative data in a
finite-order jet space for every real $\xxi\ne0$, and
homogeneous of degree $N+1$ under every nonzero real rescaling of
$\xxi$.

\emph{Step 3: the transverse matrix.}
By the homogeneous clause \eqref{eq:covariance-scalar-form} of
Proposition~\ref{thm:dependence}, with $k=N+1$,
\begin{equation}\label{eq:sym-sharp-form}
 \matS^\sharp=c_N\chi^N\mathsfbi C_{\xxi},\qquad \chi=\ww\cdot\xxi,
\end{equation}
with $c_N$ locally constant. Since
\begin{equation}\label{eq:sym-cross-product-inverse}
 \mathsfbi C_{\xxi}^2=-q_{\xxi}\matP,\qquad q_{\xxi}>0,
\end{equation}
multiplication of \eqref{eq:sym-sharp} and \eqref{eq:sym-sharp-form}
by $\mathsfbi C_{\xxi}$, together with $\matS=\matP\matS\matP$, gives
\begin{equation}\label{eq:sym-S-form}
 \matS=c_N\chi^N\matP.
\end{equation}

\emph{Step 4: comparison of the two responses.} Put
\begin{equation}\label{eq:sym-polynomial-block}
 \mathsfbi V(\xxi)=\sigma^u_{\boldsymbol r,N}(-\uu\cdot\xxi,\xxi),
\end{equation}
a polynomial matrix in $\xxi$ with smooth coefficients. For real
$\xxi\ne0$ and $\boldsymbol v\perp\xxi$, the kernel preservation
\eqref{eq:kernel-preservation} gives
$\mathsfbi V(\xxi)\boldsymbol v\perp\xxi$, so \eqref{eq:sym-S-form}
reads $(\mathsfbi V-c_N\chi^N\mathsfbi I)\boldsymbol v=0$. By
Lemma~\ref{lem:harmonic-test}(c), applied to each row, there is a
polynomial column $\boldsymbol a_N$ of degree $N-1$ with
\begin{equation}\label{eq:sym-block}
 \mathsfbi V(\xxi)=c_N\chi^N\mathsfbi I+\boldsymbol a_N(\xxi)\xxi^{\Tr},
\end{equation}
a polynomial identity in $\xxi$ that extends to complex $\xxi$ with
the real compatible derivative data fixed. Now let $\xxi$ be complex
null, $\xxi\ne0$, and
$\tau=-\uu\cdot\xxi$, so that $\alpha=0$. The velocity rows of
\eqref{eq:sym-harmonic-response}, extended to $\alpha=0$ in
Proposition~\ref{prop:sym-harmonic}, give
$\mathsfbi V(\xxi)\xxi=\Phi_N(t,\xx,-\uu\cdot\xxi,\xxi)\,\xxi$, while
\eqref{eq:sym-block} gives $\mathsfbi V(\xxi)\xxi=c_N\chi^N\xxi$. Hence
\begin{equation}\label{eq:sym-intersection}
 c_N\,(\ww\cdot\xxi)^N=\Phi_N(t,\xx,-\uu\cdot\xxi,\xxi),
 \qquad \xxi^{\Tr}\xxi=0,\quad\xxi\ne0.
\end{equation}
At fixed $(t,\xx,\uu,p)$, the only restriction on the velocity gradient
is $\tr\matA=0$, so $\ww$ varies over an open subset of $\R^3$.
The real and imaginary parts of the nonzero null covector $\xxi$ are
orthogonal and of equal length. Consequently $\ww\cdot\xxi$ ranges over
an open subset of $\C$. Since the right side of \eqref{eq:sym-intersection}
is fixed under these changes, $c_N=0$. With $t$, $\xx$ and $\xxi$ fixed, independent
changes of the real velocity $\uu$ then make $-\uu\cdot\xxi$ range over
an open subset of $\C$. Equation~\eqref{eq:sym-intersection} gives
infinitely many roots of the polynomial
$\Phi_N(t,\xx,\,\cdot\,,\xxi)$, which therefore vanishes identically
in its time-covector argument. Continuity includes $\xxi=0$, so
$\Phi_N=0$ on the whole null cone. Both degree-$N$
responses of $q$ are now zero: the harmonic response by
\eqref{eq:sym-harmonic-response}, and the transverse response by
\eqref{eq:sym-block} and the kernel preservation
\eqref{eq:kernel-preservation}. By Lemma~\ref{lem:descent} the order
of $q$ is lowered by one, and repetition ends at order one.
\end{proof}

\subsection{First order}\label{sec:sym-first}

Let $q=(\boldsymbol r,r^0)$ be a local symmetry of order at most one.
The harmonic response identifies the derivative terms of $q$, which
are subtracted; the remainder depends on the first derivatives only
through the vorticity, and the vorticity term is then excluded by a
quadratic term in the determining equation. The remaining equations
for the coefficients of a point characteristic are solved last.

By Proposition~\ref{prop:sym-harmonic}, $\Phi_1$ is linear in the
covector. Denote its coefficients by $-X^0$ and $-\boldsymbol X$,
consistent with \eqref{eq:sym-point-characteristic}, and subtract the
corresponding derivative terms from $q$:
\begin{equation}\label{eq:sym-first-order-phi}
 \Phi_1=-X^0(t,\xx)\,\tau-\boldsymbol X(t,\xx)\cdot\xxi,\qquad
 q_1=-X^0D_t(\uu,p)-(\boldsymbol X\cdot D_{\xx})(\uu,p).
\end{equation}
The pair $q_1$ has symbol $\sigma_{q_1,1}=\Phi_1\mathsfbi I$, hence
the same harmonic response \eqref{eq:sym-harmonic-response} as $q$.
Put $(\boldsymbol R,R^0)=q-q_1$. Its harmonic response is zero, and
since $\Phi_1\mathsfbi I$ commutes with $\sigma_E$, the symbol
identity \eqref{eq:kernel-preservation} for $q$ gives
\[
 \sigma_E\sigma_{(\boldsymbol R,R^0),1}
   =(T_1-\Phi_1\mathsfbi I)\sigma_E.
\]
The degree-one symbol of $(\boldsymbol R,R^0)$ therefore maps each
characteristic kernel into itself. On solutions,
$L(\boldsymbol R,R^0)=-Lq_1$. Expanding $Lq_1$ by the product rule,
the terms multiplied by $X^0$ and $\boldsymbol X$ are
$-X^0D_tE-(\boldsymbol X\cdot D_{\xx})E$ with the Euler column $E$ of
\eqref{eq:s1-4}, differentiated Euler equations, which vanish on
solutions. The remaining terms contain derivatives of these
coefficients and give the value on solutions,
\begin{equation}\label{eq:sym-point-residual}
 Lq_1=-\binom{(YX^0)\uu_t+\matA(Y\boldsymbol X)+p_t\nabla X^0
              +(D_{\xx}\boldsymbol X)^{\Tr}\nabla p}
             {\nabla X^0\cdot\uu_t+\tr[(D_{\xx}\boldsymbol X)\matA]}
 \qquad\text{on solutions}.
\end{equation}

\emph{Harmonic response in free coordinates.} Compatible first-order
derivative data can be specified by the coordinates
\begin{equation}\label{eq:sym-free-first-order}
 (t,\xx,\uu,p,\matA,\boldsymbol b_p,p_t),\qquad
 \tr\matA=0,\quad \boldsymbol b_p=\nabla p,
 \qquad \uu_t=-\matA\uu-\boldsymbol b_p.
\end{equation}
By \eqref{eq:s2-first-order-increments}, the harmonic increment with
amplitude $(-\xxi,\alpha)^{\Tr}$ changes $\matA$ by $-\xxi\xxi^{\Tr}$,
$\boldsymbol b_p$ by $\alpha\xxi$ and $p_t$ by $\alpha\tau$. For each
component $F$ of $(\boldsymbol R,R^0)$, the vanishing harmonic
response reads
\begin{equation}\label{eq:sym-first-order-harmonic}
 -\mathrm d_{\matA}F[\xxi\xxi^{\Tr}]+\alpha F_{\boldsymbol b_p}\cdot\xxi
 +\alpha\tau F_{p_t}=0,\qquad \xxi^{\Tr}\xxi=0,
\end{equation}
where $\mathrm d_{\matA}F$ is the differential of $F$ in the trace-free
entries of $\matA$, with the other arguments fixed, extended complex
linearly. The rank-one matrix $\xxi\xxi^{\Tr}$ has zero trace for a
null covector, so this
evaluation is well defined. With the compatible derivative data and the null
covector $\xxi$ fixed, both sides of \eqref{eq:sym-first-order-harmonic}
are polynomials in $\tau$. The coefficients of $\tau^2$ and $\tau$
give $F_{p_t}=0$ and $F_{\boldsymbol b_p}\cdot\xxi=0$ for every null
$\xxi$; null covectors span $\C^3$, so $F_{\boldsymbol b_p}=0$. The
constant term implies that $\mathrm d_{\matA}F$ vanishes on the matrices
$\xxi\xxi^{\Tr}$ with $\xxi$ null. Their real and imaginary parts span the
trace-free symmetric matrices: for
$\xxi=\boldsymbol e_1+\mathrm i\boldsymbol e_2$,
$\xxi\xxi^{\Tr}=\mathsfbi E_{11}-\mathsfbi E_{22}
+\mathrm i(\mathsfbi E_{12}+\mathsfbi E_{21})$, and the cyclic
permutations of the indices give the five matrices
$\mathsfbi E_{11}-\mathsfbi E_{22}$, $\mathsfbi E_{22}-\mathsfbi E_{33}$,
$\mathsfbi E_{12}+\mathsfbi E_{21}$, $\mathsfbi E_{23}+\mathsfbi E_{32}$,
$\mathsfbi E_{31}+\mathsfbi E_{13}$, a basis of that five-dimensional
space. Hence $(\boldsymbol R,R^0)$ depends on
the first derivatives only through the skew part of $\matA$, that is,
through the vorticity:
\begin{equation}\label{eq:sym-first-order-vorticity}
 \boldsymbol R=\boldsymbol R(t,\xx,\uu,p,\ww),\qquad
 R^0=R^0(t,\xx,\uu,p,\ww).
\end{equation}

\emph{Transverse kernel preservation.} By
\eqref{eq:s2-first-order-increments}, a transverse increment
$(\boldsymbol v,0)^{\Tr}$ with $\boldsymbol v\perp\xxi$ changes $\ww$
by $\xxi\times\boldsymbol v$ and leaves $\boldsymbol b_p$ and $p_t$
fixed, so the degree-one responses of
\eqref{eq:sym-first-order-vorticity} are
$\boldsymbol R_{\ww}(\xxi\times\boldsymbol v)$ and
$R^0_{\ww}\cdot(\xxi\times\boldsymbol v)$. Kernel preservation
requires the scalar response to vanish and the vector response to be
transverse:
\begin{equation}\label{eq:sym-first-order-kernel}
 R^0_{\ww}\cdot(\xxi\times\boldsymbol v)=0,\qquad
 \xxi^{\Tr}\boldsymbol R_{\ww}(\xxi\times\boldsymbol v)=0.
\end{equation}
For a fixed nonzero real $\xxi$, the vectors $\xxi\times\boldsymbol v$
range over $\xxi^\perp$. Varying $\xxi$ in the first identity of
\eqref{eq:sym-first-order-kernel} therefore gives $R^0_{\ww}=0$.
By the second identity, $\boldsymbol R_{\ww}^{\Tr}\xxi$ is orthogonal
to $\xxi^\perp$, hence parallel to $\xxi$: every real $\xxi$ is an
eigenvector of $\boldsymbol R_{\ww}^{\Tr}$. Taking for $\xxi$ the coordinate vectors
and their pairwise sums shows that all three eigenvalues are equal, so
$\boldsymbol R_{\ww}=k\mathsfbi I$. For fixed $i\ne j$, with no
summation, commutation of $\partial_{\omega_j}\partial_{\omega_i}R_i$
gives $\partial_{\omega_j}k=0$, so $k$ is independent of the vorticity.
Integrating,
\begin{equation}\label{eq:sym-first-order-form}
 \boldsymbol R=k\ww+\boldsymbol U(t,\xx,\uu,p),\qquad
 R^0=U^0(t,\xx,\uu,p),\qquad k=k(t,\xx,\uu,p).
\end{equation}

\emph{The vorticity term.} On solutions, the vorticity equation
$Y\ww=\matA\ww$ and the multiplication term $+\matA$ in $L$ give
\begin{equation}\label{eq:sym-quadratic}
 \begin{gathered}
 Y(k\ww)+\matA(k\ww)=(Yk)\ww+2k\matA\ww,\\
 Yk=k_t+\uu\cdot k_{\xx}-k_{\uu}\cdot\boldsymbol b_p
      +k_p(p_t+\uu\cdot\boldsymbol b_p).
 \end{gathered}
\end{equation}
In the coordinates \eqref{eq:sym-free-first-order}, $Yk$ is independent of
$\matA$, and every other term of the vector determining equation for
$q=q_1+(\boldsymbol R,R^0)$ is affine in $\matA$: the source
\eqref{eq:sym-point-residual}, because $\uu_t=-\matA\uu-\boldsymbol b_p$,
and $L(\boldsymbol U,U^0)$, because $Y\boldsymbol U$, $\matA\boldsymbol U$
and $\nabla U^0$ are affine in the first derivatives. With the other
coordinates fixed, this determining equation is polynomial in the
trace-free entries of $\matA$ and holds on an open set of these
entries, so its quadratic part $2k\matA\ww$ vanishes identically.
Write $\matA=\mathsfbi S+\tfrac12\mathsfbi C_{\ww}$ with $\mathsfbi S$
symmetric and trace-free; then $\matA\ww=\mathsfbi S\ww$, and
$\mathsfbi S$ varies independently of $\ww$ in the fibre. Where
$\ww\ne0$, the choice
$\mathsfbi S=\ww\ww^{\Tr}-\tfrac13|\ww|^2\mathsfbi I$ gives
$\mathsfbi S\ww=\tfrac23|\ww|^2\ww\ne0$, so $k=0$ there; since $k$
does not depend on $\ww$, $k=0$ on the chart.

\emph{Point coefficients.} The symmetry is now the point characteristic
\eqref{eq:sym-point-characteristic} with coefficients
$X^0,\boldsymbol X$ depending on $(t,\xx)$ and $\boldsymbol U,U^0$
depending on $(t,\xx,\uu,p)$. Write
$\mathsfbi B=D_{\xx}\boldsymbol X$, $\boldsymbol T=\nabla X^0$, and let
subscripts on $\boldsymbol U,U^0$ denote partial derivatives with the
other arguments fixed. Since $X^0$ and $\boldsymbol X$ have no field
arguments, their total spatial derivatives are ordinary partial
derivatives in $\xx$. In the coordinates
\eqref{eq:sym-free-first-order}, the momentum and divergence equations
of \eqref{eq:sym-determining} become
\begin{equation}\label{eq:sym-point-equations}
 \begin{aligned}
 \boldsymbol Z={}&\boldsymbol U+(X^0_t+\uu\cdot\boldsymbol T)\uu
               -\boldsymbol X_t-\mathsfbi B\uu,\\
 0={}&(\boldsymbol U_p-\boldsymbol T)p_t
  +\left[-\boldsymbol U_{\uu}+(U^0_p+X^0_t+\uu\cdot\boldsymbol T)\mathsfbi I
         -\mathsfbi B^{\Tr}+\boldsymbol U_p\uu^{\Tr}\right]\boldsymbol b_p\\
  &+\matA\boldsymbol Z+\matA^{\Tr}U^0_{\uu}
  +\boldsymbol U_t+(\partial_{\xx}\boldsymbol U)\uu+\partial_{\xx}U^0,\\
 0={}&\operatorname{div}_{\xx}\boldsymbol U+\tr[(\boldsymbol U_{\uu}-\mathsfbi B)\matA]
  +\boldsymbol T\cdot\matA\uu+(\boldsymbol U_p+\boldsymbol T)\cdot\boldsymbol b_p.
 \end{aligned}
\end{equation}
Here $\boldsymbol U_{\uu}$ is the velocity Jacobian, $U^0_{\uu}$ a
column, and $\partial_{\xx}$ and $\operatorname{div}_{\xx}$ hold the field values
fixed. At fixed $(t,\xx,\uu,p)$ the entries $p_t$, $\boldsymbol b_p$ and
the trace-free part of $\matA$ are independent, so the coefficients
of each are compared separately. The coefficient of
$p_t$ in the momentum equation and of $\boldsymbol b_p$ in the
divergence equation give
$\boldsymbol U_p=\boldsymbol T=-\boldsymbol T$, so $X^0=X^0(t)$ and
$\boldsymbol U_p=0$. The coefficient identity
$\matA\boldsymbol Z+\matA^{\Tr}U^0_{\uu}=0$ holds for every trace-free
$\matA$, and $\boldsymbol Z$ does not contain $\matA$. Evaluation on
the elementary off-diagonal matrices $\matA=\mathsfbi E_{ij}$, $i\ne j$,
for which $\mathsfbi E_{ij}\boldsymbol Z=Z_j\boldsymbol e_i$ and
$\mathsfbi E_{ij}^{\Tr}U^0_{\uu}=U^0_{u_i}\boldsymbol e_j$, gives
$Z_j=0$ and $U^0_{u_i}=0$ for all $i\ne j$; hence $\boldsymbol Z=0$ and
$U^0_{\uu}=0$. Setting $\boldsymbol Z=0$ determines $\boldsymbol U$,
and then $\boldsymbol U_{\uu}=\mathsfbi B-X^{0\prime}\mathsfbi I$.
The coefficient of $\boldsymbol b_p$ in the vector determining
equation becomes
$-(\mathsfbi B+\mathsfbi B^{\Tr})+(U^0_p+2X^{0\prime})\mathsfbi I=0$.
Its trace gives $2\operatorname{div}_{\xx}\boldsymbol X=3(U^0_p+2X^{0\prime})$,
which determines $U^0_p$; integration with respect to $p$, at
fixed $(t,\xx)$, gives
\begin{equation}\label{eq:sym-point-steps}
 \begin{gathered}
 \boldsymbol U=\boldsymbol X_t+(\mathsfbi B-X^{0\prime}\mathsfbi I)\uu,\qquad
 \mathsfbi B+\mathsfbi B^{\Tr}=(U^0_p+2X^{0\prime})\mathsfbi I,\\
 \varsigma=\tfrac13\operatorname{div}_{\xx}\boldsymbol X,\qquad
 \mathsfbi B+\mathsfbi B^{\Tr}=2\varsigma\mathsfbi I,\\
 U^0=2(\varsigma-X^{0\prime})p+\varrho(t,\xx).
 \end{gathered}
\end{equation}
With $\boldsymbol T=0$, $\boldsymbol U_p=0$,
$\boldsymbol U_{\uu}-\mathsfbi B=-X^{0\prime}\mathsfbi I$ and
$\tr\matA=0$, the divergence equation in
\eqref{eq:sym-point-equations} reduces to
$\operatorname{div}_{\xx}\boldsymbol U
=\partial_t\operatorname{div}_{\xx}\boldsymbol X
+\uu\cdot\nabla\operatorname{div}_{\xx}\boldsymbol X
=3(\varsigma_t+\uu\cdot\nabla\varsigma)=0$. Since the velocity values
are independent, both $\varsigma_t$ and $\nabla\varsigma$ vanish;
thus $\varsigma$ is constant. Differentiating the resulting relation
$D_iX_j+D_jX_i=2\varsigma\delta_{ij}$ and combining its three index
permutations gives
$2D_iD_jX_k=D_i(D_jX_k+D_kX_j)+D_j(D_iX_k+D_kX_i)-D_k(D_iX_j+D_jX_i)=0$.
Consequently $\boldsymbol X=\mathsfbi M(t)\xx+\boldsymbol b(t)$ with
$\mathsfbi M+\mathsfbi M^{\Tr}=2\varsigma\mathsfbi I$, so that
$\mathsfbi M-\varsigma\mathsfbi I$ is skew and
\[
 \boldsymbol X=\varsigma\xx+\boldsymbol\Omega(t)\times\xx+\boldsymbol b(t),
\]
where $\boldsymbol\Omega(t)$ is a vector function of time.
The momentum equation becomes
\begin{equation}\label{eq:sym-point-final}
 (2\mathsfbi C_{\boldsymbol\Omega'}-X^{0\prime\prime}\mathsfbi I)\uu
 +\boldsymbol X_{tt}+\nabla\varrho=0,
\end{equation}
with $\uu$ arbitrary at fixed $(t,\xx)$. The skew and scalar parts
of its velocity coefficient give $\boldsymbol\Omega'=0$ and
$X^{0\prime\prime}=0$, respectively. The remaining gradient equation
integrates to $\varrho=-\boldsymbol b''(t)\cdot\xx+\varpi(t)$. With
$X^0=\psi_1t+\psi_0$ this is the family
\eqref{eq:sym-family}, the point symmetry algebra of
\citet{Buchnev1971} and \citet[Theorem~2.1]{Olver1982}.

\emph{Sufficiency.} For the family \eqref{eq:sym-family}, direct
differentiation with unrestricted derivative data gives
\begin{equation}\label{eq:sym-sufficiency}
 \begin{aligned}
 Y\boldsymbol r+\matA\boldsymbol r+\nabla r^0
   &=\bigl[(\varsigma-2\psi_1)\mathsfbi I
          +\mathsfbi C_{\boldsymbol\Omega}\bigr]\boldsymbol F
     -X^0D_t\boldsymbol F-(\boldsymbol X\cdot D_{\xx})\boldsymbol F,\\
 \Div\boldsymbol r&=-\psi_1G-X^0D_tG-\boldsymbol X\cdot D_{\xx}G.
 \end{aligned}
\end{equation}
The right-hand sides vanish on solutions, so every member of the family is
a local symmetry. This completes the proof of Theorem~\ref{thm:symmetries}.

\section{Hamiltonian structure and adjoint symmetries}\label{sec:8}

The operator $\mathscr N$ of \eqref{eq:noether-map} is first shown
to satisfy the Hamiltonian criterion of
\citet{KerstenKrasilshchikVerbovetskyVitolo2009}. Its kernel on
adjoint symmetries is then determined at arbitrary finite order.
For an arbitrary adjoint symmetry $Q$, the image $\mathscr NQ$ is
already classified by Theorem~\ref{thm:symmetries}. Its action on
the two multipliers \eqref{eq:hamiltonian-casimirs}
excludes the two dilations. Subtracting a classical preimage then
reduces the adjoint classification to the kernel calculation.

\subsection{The Hamiltonian property}\label{sec:hamiltonian-proof}

The framework of \citet{KerstenKrasilshchikVerbovetskyVitolo2009}
uses formal adjoints in all independent variables and applies to
equations whose linearisation has the following property, their
normality condition.

\begin{lemma}[Normality of the linearised Euler operator]\label{lem:normality}
If a finite-order differential operator $\Delta$ satisfies
$\Delta L=0$ on solutions, then $\Delta=0$ there. The same conclusion
holds with $L^*$ in place of $L$.
\end{lemma}
\begin{proof}
In the tilted coordinates \eqref{eq:s2-tilted-coordinates}, the
linearisation \eqref{eq:finite-LM} has the form
\[
 L=\cN D_s+T_y,\qquad \det\cN=-b_*^2\ne0,
 \qquad b_*=a+u_3,
\]
where $T_y$ contains only tangential differentiations, including
the time derivative, and $\cN$ is the matrix in
\eqref{eq:s3-normal-form}. An operator vanishes on solutions when all
its coefficients do. Write $\Delta=\sum_{i=0}^k\Delta_i^yD_s^i$, with
tangential differential operators $\Delta_i^y$. The coefficient of
$D_s^{k+1}$ in $\Delta L$ is $\Delta_k^y\circ\cN$, which therefore
vanishes on solutions. Composition on the right with $\cN^{-1}$, whose
coefficients are smooth, preserves this property, so $\Delta_k^y=0$ on
solutions; the term $\Delta_k^yD_s^kL$ then also vanishes there, and
descending in $k$ gives $\Delta=0$. For $L^*$ the leading coefficient
is $-\cN^{\Tr}$, and the same argument applies. The constant $a$
depends on the selected derivative data, so the conclusion holds on a
neighbourhood of every point of the solution manifold.
\end{proof}

For a system with differential constraints, the criterion for a
\emph{variational bivector} is an identity involving the linearisation
and its formal adjoint \citep[Sec.~3]{KerstenKrasilshchikVerbovetskyVitolo2009}.
For the operator \eqref{eq:noether-map} this identity reads
\begin{equation}\label{eq:noether-intertwining}
 L\mathscr N=\mathscr N^*L^*
 \qquad\text{on solutions}.
\end{equation}
In particular, \eqref{eq:noether-intertwining} implies that
$\mathscr NQ$ is a symmetry whenever $Q$ is an adjoint symmetry.
The Hamiltonian condition is the vanishing of the variational
Schouten bracket \eqref{eq:ham-schouten-definition}, verified in the
proof below.

\begin{proof}[Proof of the Hamiltonian assertion in Theorem~\ref{thm:hamiltonian}]
Since $Y^*=-Y-G$, $\nabla^*=-\Div$ and a cross-product matrix
is skew-symmetric, the unrestricted formal adjoints are
\begin{equation}\label{eq:ham-exact-adjoints}
 \begin{aligned}
 L^*(\LL,\Lambda^0)
   &=\binom{-Y\LL-G\LL+\matA^{\Tr}\LL-\nabla\Lambda^0}
                {-\Div\LL},\\
 \mathscr N^*(\boldsymbol v,\pi)
   &=\binom{\ww\times\boldsymbol v}
                {\Div\boldsymbol v-Y\pi-G\pi}.
 \end{aligned}
\end{equation}
Thus $M=-L^*$ on solutions, as in \eqref{eq:finite-LM}.
The same bivector also has a representative that is formally
skew-adjoint on solutions. Adding the scalar component of $L^*Q$
in \eqref{eq:ham-exact-adjoints} to the scalar component of
$\mathscr NQ$ in \eqref{eq:noether-map} gives
\[
 \widetilde{\mathscr N}
 =\begin{pmatrix}
   -\mathsfbi C_{\ww} & -\nabla\\
   -\Div & Y
  \end{pmatrix}.
\]
This addition vanishes on every adjoint symmetry and factors through
$L^*$. The off-diagonal entries are negative formal adjoints of each
other, and $Y^*=-Y$ when $G=0$. Hence
$\widetilde{\mathscr N}^*=-\widetilde{\mathscr N}$ on solutions.

For the representative $\mathscr N$ in \eqref{eq:noether-map} and an
arbitrary $Q=(\LL,\Lambda^0)$, direct differentiation
gives the unrestricted identity
\begin{equation}\label{eq:ham-bivector-defect}
 \begin{gathered}
 L\mathscr NQ-\mathscr N^*L^*Q=B_Q(E),\\
 B_Q(E)=
 \binom{-(\Curl\boldsymbol F)\times\LL
                       +G\,\ww\times\LL}{0}.
 \end{gathered}
\end{equation}
Here $B_Q$ is linear in the Euler column $E=(\boldsymbol F,G)^{\Tr}$ of
\eqref{eq:s1-4}.
To verify \eqref{eq:ham-bivector-defect}, the required product and
matrix identities are
\begin{equation}\label{eq:ham-residual-calculation}
 \begin{gathered}
 \Curl\boldsymbol F=Y\ww-\matA\ww+G\ww,\qquad
 \matA^{\Tr}\ww=\matA\ww,\\
 \matA\mathsfbi C_{\ww}
       +\mathsfbi C_{\ww}\matA^{\Tr}
       =G\mathsfbi C_{\ww}-\mathsfbi C_{\matA^{\Tr}\ww},\\
 \Div(Y\LL-\matA\LL)=Y\Div\LL-\LL\cdot\nabla G .
 \end{gathered}
\end{equation}
The first identity follows by taking the curl of the momentum
equation $\boldsymbol F$, and the second uses
$\matA-\matA^{\Tr}=\mathsfbi C_{\ww}$. Substitution of
\eqref{eq:ham-exact-adjoints} into the left side of
\eqref{eq:ham-bivector-defect} gives vector component
$-(Y\ww-\matA\ww)\times\LL$ by the second and third identities in
\eqref{eq:ham-residual-calculation}; the first identity then gives
the displayed residual terms. For the scalar component, the last
identity in \eqref{eq:ham-residual-calculation} leaves
\[
 -\LL\cdot\nabla G+\Div(G\LL)-G\Div\LL=0.
\]
The right side of \eqref{eq:ham-bivector-defect} vanishes on
solutions, proving \eqref{eq:noether-intertwining}.

The Jacobi condition involves the linearisation of the coefficients
of $\mathscr N$ with respect to the fields. With $Q$ held fixed, this
coefficient linearisation is
\begin{equation}\label{eq:ham-coefficient-linearisation}
 \ell_{\mathscr N,Q}(\boldsymbol v,\pi)
 :=\ell_{\mathscr NQ}(\boldsymbol v,\pi)
       -\mathscr N\ell_Q(\boldsymbol v,\pi)
 =\binom{-(\Curl\boldsymbol v)\times\LL}
              {(\boldsymbol v\cdot\nabla)\Lambda^0}.
\end{equation}
The other required operation is integration by parts in the
argument $E$ of $B_Q$. For an arbitrary column
$C=(\boldsymbol c,c^0)$ this gives
\begin{equation}\label{eq:ham-residual-adjoint}
 B_Q^*C
 =\binom{\Curl(\boldsymbol c\times\LL)}
              {\boldsymbol c\cdot(\ww\times\LL)}.
\end{equation}
Indeed, the coefficient of $\boldsymbol F$ in
$C\cdot B_Q(E)$ after integration by parts is
$\Curl(\boldsymbol c\times\LL)$, since $\Curl^*=\Curl$.
For two arbitrary input columns
$Q_i=(\LL_i,\Lambda_i^0)$, set
$B^*(Q_1,Q_2)=B_{Q_2}^*Q_1$. This bilinear expression is
skew-symmetric in its two inputs.

For a variational bivector, the variational Schouten bracket
$[\![\mathscr N,\mathscr N]\!]$ is represented by the bilinear
differential expression
\begin{equation}\label{eq:ham-schouten-definition}
 \begin{aligned}
 \tfrac12[\![\mathscr N,\mathscr N]\!](Q_1,Q_2)
 &=\ell_{\mathscr N,Q_1}(\mathscr NQ_2)
  -\ell_{\mathscr N,Q_2}(\mathscr NQ_1)\\
 &\quad-\mathscr N B^*(Q_1,Q_2).
 \end{aligned}
\end{equation}
It represents zero if it is a sum of terms vanishing on solutions
and terms factoring through $L^*$ in one input or the other.
Vanishing in this sense is the Jacobi condition, and a variational
bivector satisfying it is \emph{Hamiltonian}
\citep{KerstenKrasilshchikVerbovetskyVitolo2009}.

Write the vector part of $MQ_i$ as
\[
 \boldsymbol m_i
 =Y\LL_i-\matA^{\Tr}\LL_i+\nabla\Lambda_i^0.
\]
Substitution of \eqref{eq:noether-map},
\eqref{eq:ham-coefficient-linearisation} and
\eqref{eq:ham-residual-adjoint} in
\eqref{eq:ham-schouten-definition} yields, for unrestricted fields
and inputs,
\begin{equation}\label{eq:ham-schouten-remainder}
 \begin{aligned}
 &\tfrac12[\![\mathscr N,\mathscr N]\!](Q_1,Q_2)\\
 &\quad=
 \begin{pmatrix}
   (\Div\LL_2)\,\ww\times\LL_1
       -(\Div\LL_1)\,\ww\times\LL_2\\[3pt]
   \begin{gathered}
   (\Curl\boldsymbol F)\cdot(\LL_1\times\LL_2)\\[-2pt]
       {}+\ww\cdot(\LL_1\times\boldsymbol m_2
                         -\LL_2\times\boldsymbol m_1)
   \end{gathered}
 \end{pmatrix}.
 \end{aligned}
\end{equation}
For the vector component, expand each curl of a cross product as
\[
 \Curl(\boldsymbol a\times\boldsymbol b)
 =\boldsymbol a\,\Div\boldsymbol b-\boldsymbol b\,\Div\boldsymbol a
   +(\boldsymbol b\cdot\nabla)\boldsymbol a
   -(\boldsymbol a\cdot\nabla)\boldsymbol b
\]
and use $\Div\ww=0$. For the scalar component, substitute
$Y\ww=\matA\ww-G\ww+\Curl\boldsymbol F$ from
\eqref{eq:ham-residual-calculation} and
$Y\LL_i=\matA^{\Tr}\LL_i-\nabla\Lambda_i^0+\boldsymbol m_i$.
The gradient terms cancel. The remaining matrix terms cancel by
$\matA^{\Tr}\ww=\matA\ww$ and
\[
 (\matA^{\Tr}\boldsymbol a)\times\boldsymbol b
  +\boldsymbol a\times(\matA^{\Tr}\boldsymbol b)
 =(G\mathsfbi I-\matA)(\boldsymbol a\times\boldsymbol b),
\]
leaving exactly the scalar entry of
\eqref{eq:ham-schouten-remainder}.

On solutions, the term containing $\Curl\boldsymbol F$ vanishes.
Every remaining term contains a component of $MQ_1$ or $MQ_2$;
these are $-L^*Q_1$ and $-L^*Q_2$ there. Hence
\eqref{eq:ham-schouten-remainder} represents zero and
$\mathscr N$ is Hamiltonian. For $Q=(\boldsymbol0,t)$,
\eqref{eq:ham-exact-adjoints} gives $L^*Q=0$, whereas
\eqref{eq:noether-map} gives $\mathscr NQ=(\boldsymbol0,1)$.
Every operator that factors through $L^*$ vanishes on this input
on solutions. Thus $\mathscr N$ represents a nonzero Hamiltonian structure.
\end{proof}

\subsection{Kernel at arbitrary finite order}

Proposition~\ref{thm:dependence} has a scalar consequence that
determines the entire kernel of $\mathscr N$ on adjoint symmetries.
The scalar identity in the following statement is required on an
open set of compatible derivative data.

\begin{corollary}[Local material scalar invariants]\label{cor:material-scalars}
If a smooth finite-order differential function $f$ satisfies
$Yf=0$ on solutions, then its restriction to compatible Euler
data is locally constant.
\end{corollary}
\begin{proof}
Set $\matS=f\mathsfbi C_{\xxi}$ for every real $\xxi\ne0$ over
the same derivative-data neighbourhood. This matrix is transverse
and homogeneous of degree one in $\xxi$. The expression in its
adjoint transport law \eqref{eq:s4-8} is
\begin{equation}\label{eq:ham-material-transport}
 \begin{aligned}
 &Y\matS-(\matA^{\Tr}\xxi)\cdot\partial_{\xxi}\matS
                  -\matA^{\Tr}\matS-\matS\matA\\
 &\qquad=(Yf)\mathsfbi C_{\xxi}
       +f\bigl([\mathsfbi C_{\ww},\mathsfbi C_{\xxi}]
                                      -G\mathsfbi C_{\xxi}\bigr).
 \end{aligned}
\end{equation}
Here $[U,V]=UV-VU$ is the matrix commutator. On solutions,
$Yf=G=0$, and the commutator has zero double transverse projection
by \eqref{eq:projected-cross-commutator}. Proposition~\ref{thm:dependence}
therefore applies with degree one and gives
$\matS=c\mathsfbi C_{\xxi}$ with $c$ locally constant. Since
$\mathsfbi C_{\xxi}\ne0$ for $\xxi\ne0$, $f=c$.
\end{proof}

\begin{proposition}[Kernel on local adjoint symmetries]\label{prop:hamiltonian-kernel}
If a smooth finite-order pair $Q=(\LL,\Lambda^0)$ satisfies the
adjoint equations \eqref{eq:s3-adjoint-equations} and
$\mathscr NQ=0$ on solutions, then, locally on solutions,
\begin{equation}\label{eq:ham-kernel}
 Q=cQ_H+dQ_G=c(\ww,0)+d(\boldsymbol0,1),
 \qquad c,d\in\R.
\end{equation}
\end{proposition}
\begin{proof}
The scalar component of $\mathscr NQ=0$ gives $Y\Lambda^0=0$.
By Corollary~\ref{cor:material-scalars}, $\Lambda^0=d$ is
locally constant; $\Lambda^0-d$ vanishes on solutions together with
its total derivatives, so $\nabla\Lambda^0=0$ there, and the vector
component gives $\ww\times\LL=0$. Where $\ww\ne0$, write $\LL=k\ww$,
with $k=(\ww\cdot\LL)/|\ww|^2$. Substitution in the vector
adjoint equation \eqref{eq:s3-adjoint-equations}, using
$Y\ww=\matA\ww$ and $\matA^{\Tr}\ww=\matA\ww$, gives
$(Yk)\ww=0$. The corollary therefore gives $k=c$, locally constant.

The conclusion also holds near data with zero vorticity. At order
one the velocity gradient is arbitrary subject to its zero trace,
so its three vorticity components are independent compatible-data
coordinates. Lemma~\ref{lem:free} extends these coordinates to
higher orders. Choose a product coordinate neighbourhood in which
the vorticity values form a ball about zero. Removing zero
vorticity leaves a connected neighbourhood, since the punctured
three-dimensional ball is connected. A locally constant function on a
connected set is constant, so $k=c$ throughout that neighbourhood;
the pair $\LL-c\ww$ vanishes on the dense subset where $\ww\ne0$
and is continuous, hence $\LL=c\ww$ throughout. Conversely, both
pairs in
\eqref{eq:ham-kernel} satisfy the adjoint equations and are
annihilated by $\mathscr N$.
\end{proof}

\subsection{Classification of adjoint symmetries}

An adjoint symmetry is classified by determining its image under
$\mathscr N$ and then using Proposition~\ref{prop:hamiltonian-kernel}.
The image agrees on solutions with a point characteristic by
Theorem~\ref{thm:symmetries}.
Its two dilation coefficients are determined by its action on
the two multipliers \eqref{eq:hamiltonian-casimirs}.

For a local symmetry $\phi$, choose a differential operator
$\Delta$ such that $L\phi=\Delta(E)$ for unrestricted fields.
Such an operator exists by residual factorisation, since $L\phi$
vanishes on solutions. If $C$ is a multiplier, differentiate
its conservation identity $C\cdot E=D_\mu J^\mu$ along $\phi$.
The product rule gives
$\ell_C(\phi)\cdot E+C\cdot L\phi$; integration by parts in
the second term gives the multiplier of the resulting conservation
law:
\begin{equation}\label{eq:hamiltonian-multiplier-action}
 \Pi_\phi C=\ell_C(\phi)+\Delta^*C,
 \qquad L\phi=\Delta(E).
\end{equation}
Here $\Pi_\phi C$ denotes this action, and $\ell_C$ is the
linearisation \eqref{eq:s2-linearization} of the column $C$.

The restriction of \eqref{eq:hamiltonian-multiplier-action} to
solutions depends only on the symmetry characteristic there.
To verify this directly, let
$\phi-\phi_0=T(E)$ and $L\phi_0=\Delta_0(E)$, and put
$U=\Delta-\Delta_0-LT$. Then $U(E)=0$ identically.
Linearisation on solutions gives $UL=0$: derivatives of the
coefficients of $U$ multiply residuals and vanish there.
Lemma~\ref{lem:normality} gives $U=0$ on solutions.
For a multiplier $C$, linearise its identity
$C\cdot E=D_\mu J^\mu$ along an arbitrary test column.
Integrating by parts in both terms on the left gives coefficient
$\ell_C^*(E)+L^*C$ of that column. The right side is a
divergence and has zero coefficient after the same operation,
so $L^*C=-\ell_C^*(E)$. Consequently,
\begin{equation}\label{eq:ham-action-representatives}
 \Pi_\phi C-\Pi_{\phi_0}C
 =\ell_C(T(E))+T^*L^*C+U^*C=0
 \qquad\text{on solutions}.
\end{equation}
The case $\phi=\phi_0$ also proves independence of the choice
of $\Delta$.

Now let $Q=(\LL,\Lambda^0)$ be any local adjoint symmetry, that is,
any pair satisfying \eqref{eq:s3-adjoint-equations} on solutions.
The factorisation
\eqref{eq:s2-residual-factorization}, together with
$E=\cN\boldsymbol n$, gives $L^*Q=R_Q(E)$ for a
finite-order differential operator $R_Q$. Thus
\eqref{eq:ham-bivector-defect} gives
\begin{equation}\label{eq:ham-image-residual}
 \phi=\mathscr NQ,\qquad
 L\phi=\Delta(E),\qquad
 \Delta=\mathscr N^*R_Q+B_Q .
\end{equation}
For $C=Q_G$ or $Q_H$, $\mathscr NC=0$ identically.
The term $R_Q^*\mathscr NC$ in $\Delta^*C$ therefore vanishes.
For $Q_G=(\boldsymbol0,1)$, both
$\ell_{Q_G}(\phi)$ and $B_Q^*Q_G$ are zero. For $Q_H=(\ww,0)$,
\eqref{eq:noether-map} and \eqref{eq:ham-residual-adjoint} give
\[
 \ell_{Q_H}(\phi)
 =\binom{-\Curl(\ww\times\LL)}{0},\qquad
 B_Q^*Q_H=\binom{\Curl(\ww\times\LL)}{0}.
\]
These terms cancel, so every adjoint image satisfies
\begin{equation}\label{eq:ham-casimir-preservation}
 \Pi_{\mathscr NQ}Q_G=0,\qquad
 \Pi_{\mathscr NQ}Q_H=0 .
\end{equation}

Let $\phi_0$ be the point-symmetry representative
\eqref{eq:sym-point-characteristic} of $\mathscr NQ$ given by
Theorem~\ref{thm:symmetries}, with coefficients
\eqref{eq:sym-family}. The operator $\Delta_0$ with
$L\phi_0=\Delta_0(E)$ is read off from \eqref{eq:sym-sufficiency}.
Its scalar block
is $-\psi_1-X^0D_t-\boldsymbol X\cdot D_{\xx}$, whose formal
adjoint applied to $1$ is
$-\psi_1+D_tX^0+\Div\boldsymbol X=3\varsigma$.
For the helicity multiplier, write $\boldsymbol\phi_0$ for the
velocity component of $\phi_0$. Differentiation gives
\[
 \Curl\boldsymbol\phi_0
 =-\psi_1\ww+\boldsymbol\Omega\times\ww
     -X^0D_t\ww-(\boldsymbol X\cdot D_{\xx})\ww .
\]
Adding the adjoint vector block of $\Delta_0$
applied to $\ww$ cancels the rotation and derivative terms.
The two actions are therefore
\begin{equation}\label{eq:ham-action-weights}
 \Pi_{\phi_0}Q_G=3\varsigma Q_G,\qquad
 \Pi_{\phi_0}Q_H=(4\varsigma-2\psi_1)Q_H
 \qquad\text{on solutions}.
\end{equation}
By \eqref{eq:ham-action-representatives} these are the actions
of $\mathscr NQ$, which vanish by
\eqref{eq:ham-casimir-preservation}. The first equality gives
$\varsigma=0$. The vorticity values range over an open set,
so the second then gives $\psi_1=0$.

\begin{proposition}[Local adjoint symmetries]\label{prop:adjoint-classification}
Every smooth local adjoint symmetry of finite joint order agrees
on solutions with a pair \eqref{eq:s1-5}.
\end{proposition}
\begin{proof}
For any such adjoint symmetry $Q$, its image $\mathscr NQ$
has a point representative with $\varsigma=\psi_1=0$, as just
proved. Take $Q_0$ from the explicit multiplier family
\eqref{eq:s1-5}, with the constants $\boldsymbol\Omega$ and the
function $\boldsymbol b(t)$ of the point representative, with $c_0(t)$
a primitive of $\varpi(t)$, and with
\begin{equation}\label{eq:ham-classical-preimage}
 c=0,\qquad \nu=-\psi_0,\qquad c_0'(t)=\varpi(t).
\end{equation}
With
$\boldsymbol\beta=\boldsymbol\Omega\times\xx+\boldsymbol b(t)$,
substitution in \eqref{eq:noether-map} gives
\begin{equation}\label{eq:ham-preimage-identity}
 \mathscr NQ_0-\phi_0
 =\binom{-\nu\boldsymbol F}
              {(\nu\uu+\boldsymbol\beta)\cdot\boldsymbol F}.
\end{equation}
Thus $Q-Q_0$ is an adjoint symmetry with zero image on solutions.
Proposition~\ref{prop:hamiltonian-kernel} gives
$Q-Q_0=cQ_H+dQ_G$ there. Absorbing $d$ into the integration
constant in $c_0$ yields precisely \eqref{eq:s1-5}.
\end{proof}

The pairs \eqref{eq:s1-5} are multipliers by the unrestricted
identities \eqref{eq:s9-current-identities}, so they are adjoint
symmetries. The converse is Proposition~\ref{prop:adjoint-classification},
which completes the adjoint-symmetry statement of
Theorem~\ref{thm:main}. The image of $\mathscr N$ consists exactly
of time translation, rotations, moving frames and pressure gauge:
\eqref{eq:ham-action-weights} excludes both dilations, and
\eqref{eq:ham-classical-preimage} gives a preimage for every
remaining point symmetry. Together with
Proposition~\ref{prop:hamiltonian-kernel}, this proves the kernel
and image assertions of Theorem~\ref{thm:hamiltonian}.

\section{Conservation currents}\label{sec:9}

The explicit currents \eqref{eq:s9-explicit-currents} realise the
pairs \eqref{eq:s1-5} as multipliers. To pass from the adjoint
classification of Section~\ref{sec:8} to the conservation-law
statement of Theorem~\ref{thm:main}, two facts are needed: every
current class has a representative with an exact multiplier identity
\eqref{eq:s1-4}, and the values of that multiplier on solutions
determine the class. Both follow from the solved form
\eqref{eq:s3-normal-form} of the Euler equations and integration by
parts in the free tangential derivatives.

Two integration-by-parts operations are used. The Euler operator
gives the variational derivative of a scalar differential function
$H$ with respect to the fields; the formal adjoint of a matrix
differential operator $\mathscr O$ transfers derivatives from one
factor of a scalar product to the other, modulo a total divergence:
\begin{equation}\label{eq:s9-Euler-and-adjoint}
 (\EulerOp_w H)_a
    =\sum_I(-D)^I\dfrac{\partial H}{\partial w_I^a},
 \qquad
 \mathscr O=\sum_I\mathscr O_ID_I,
 \qquad \mathscr O^*=\sum_I(-D)^I\circ\mathscr O_I^{\Tr}.
\end{equation}
Here $w=(\uu,p)$ is the four-component field column, $a$ labels a
field component, $w_I^a=D_Iw^a$, and $I$ runs over space--time
multi-indices. The notation $(-D)^I=(-1)^{|I|}D_I$ accounts for the
sign from each integration by parts; $\circ$ denotes operator
composition. Every sum is finite on a finite-order expression.
In the formal adjoint
$\mathscr O^*$, the derivatives act on both the coefficient matrices
and the argument to their right. The Euler operator annihilates
total divergences. Applying it to the multiplier identity
\eqref{eq:s1-4} and using the product rule gives
\begin{equation}\label{eq:s9-multiplier-Euler-identity}
 \begin{gathered}
 0=\EulerOp_w(E^{\Tr}Q)=L^*Q+\ell_Q^*E,
 \qquad L=\ell_E,\\
 L^*(\LL,\Lambda^0)=
 \binom{-Y\LL-G\LL+\matA^{\Tr}\LL-\nabla\Lambda^0}
       {-\Div\LL}.
 \end{gathered}
\end{equation}
The operator $\ell_Q$ is the linearisation of the multiplier column
$Q$ with respect to $w$, as defined in
\eqref{eq:s2-linearization}. Since the residual $E$ of
\eqref{eq:s1-4} and its derivatives vanish
on solutions, \eqref{eq:s9-multiplier-Euler-identity} reduces there to
the adjoint equations
\eqref{eq:s3-adjoint-equations}: the restriction of every multiplier
to solutions is a local adjoint symmetry. This is the statement made
after \eqref{eq:s3-adjoint-equations}.

\begin{proposition}[Current correspondence]\label{prop:current-correspondence}
Every local conservation-current class contains a representative with
an exact finite-order multiplier identity \eqref{eq:s1-4}. Two
currents with exact multiplier identities are equivalent exactly when
their multipliers agree on solutions. In particular, a current whose
multiplier vanishes on solutions is trivial.
\end{proposition}
\begin{proof}
\emph{Step 1: an exact multiplier.}
Let $(\rho,\boldsymbol J)$ be a local conservation current of finite
order. The current is first expressed in the free tangential
coordinates of \eqref{eq:s2-residual-coordinates}. The difference
between its unrestricted divergence and its divergence on solutions can
then be expressed through the residuals
$\boldsymbol n=\boldsymbol w_s-\cF$. In the tilted coordinates
\eqref{eq:s2-tilted-coordinates}, the divergence transforms as
\begin{equation}\label{eq:s3-current-transformation}
 \begin{gathered}
 D_t\rho+D_iJ^i
   =D_s(a\rho+J^3)+D_{y_0}\rho
                      +D_{y_1}J^1+D_{y_2}J^2,\\
 \widehat\rho=\cR(a\rho+J^3),\qquad
 \widehat{\boldsymbol J}=(\cR\rho,\cR J^1,\cR J^2).
 \end{gathered}
\end{equation}
The second line defines the reduced current components. The reduction
$\cR$ of \eqref{eq:s2-tangential-representative} sets the residuals
$\boldsymbol n_I$ in \eqref{eq:s2-residual-coordinates} to zero.
Here $I=(k,J)$ is a multi-index in the tilted coordinates and
$\boldsymbol n_I=D_s^kD_y^J\boldsymbol n$.
To separate differentiation on compatible derivative data from
unrestricted differentiation, denote the total derivative in the
$s$ direction after substituting
$\boldsymbol w_s=\cF$ from \eqref{eq:s3-normal-form} by
$\overline D_s$. Its action and the conservation equation are
\begin{equation}\label{eq:s9-normal-total-derivative}
 \begin{gathered}
 \overline D_s
   =\partial_s+\sum_{a,J}D_y^J\cF^a
                         \dfrac{\partial}{\partial w_J^a},
 \qquad w_J^a=D_y^Jw^a,\\
 \overline D_s\widehat\rho+D_y\cdot\widehat{\boldsymbol J}=0.
 \end{gathered}
\end{equation}
The second line of \eqref{eq:s9-normal-total-derivative} is obtained
as follows. On compatible derivative data every differential function
$f$ agrees with $\cR(f)$, and $f-\cR(f)$ is a sum of terms linear in
the residuals by the factorisation \eqref{eq:s2-residual-factorization},
whose total derivatives vanish on compatible data; hence
$\cR(D_sf)=\overline D_s\cR(f)$ and $\cR(D_{y_j}f)=D_{y_j}\cR(f)$.
Applying $\cR$ to \eqref{eq:s3-current-transformation}, with
$D_t\rho+D_iJ^i=0$ on solutions, gives the second line. By
Lemma~\ref{lem:free} the tangential coordinates of compatible
derivative data are free, so that line is an identity in those
coordinates. On unrestricted derivative data
$\boldsymbol w_s=\cF+\boldsymbol n$ by \eqref{eq:s3-normal-form}; the
chain rule and integration by parts in $y$ give
\begin{equation}\label{eq:s3-canonical-construction}
 \begin{aligned}
 D_s\widehat\rho-\overline D_s\widehat\rho
     &=\sum_{a,J}\dfrac{\partial\widehat\rho}{\partial w_J^a}
                         D_y^Jn^a
       =q_{\rm can}\cdot\boldsymbol n+D_y\cdot\boldsymbol B,\\
 (q_{\rm can})_a
     &=\sum_J(-D_y)^J
                    \dfrac{\partial\widehat\rho}{\partial w_J^a}.
 \end{aligned}
\end{equation}
Here $q_{\rm can}$ is the multiplier column obtained by transferring
all tangential derivatives from $\boldsymbol n$ to the coefficients
by integration by parts. Its four components depend only on the
tangential coordinates. The three-component flux $\boldsymbol B$ is
linear in the residuals and their tangential derivatives, and therefore
vanishes on solutions. Subtracting $\boldsymbol B$ from the reduced flux gives an exact
multiplier identity. Using $E=\cN\boldsymbol n$ from
\eqref{eq:s3-normal-form} gives
\begin{equation}\label{eq:s9-canonical-multiplier}
 D_s\widehat\rho+D_y\cdot(\widehat{\boldsymbol J}-\boldsymbol B)
       =q_{\rm can}\cdot\boldsymbol n=Q_{\rm can}^{\Tr}E,
 \qquad Q_{\rm can}=\cN^{-\Tr}q_{\rm can}.
\end{equation}
Transformed back by \eqref{eq:s3-current-transformation}, the
corrected current $(\rho',\boldsymbol J')$ satisfies
$D_t\rho'+\Div\boldsymbol J'=Q_{\rm can}^{\Tr}E$ and agrees with
$(\rho,\boldsymbol J)$ on solutions, where $\cR f=f$ and
$\boldsymbol B=0$.

\emph{Step 2: the multiplier values on solutions.}
The multiplier restricted to solutions must be independent of the
representative of the current class. To establish that independence, let $C$
be a current in the class of $(\rho',\boldsymbol J')$ with exact
multiplier $Q$. By the definition of current equivalence, subtraction
of a current with identically zero divergence preserves the multiplier
identity and makes the difference
$\Delta=C-(\rho',\boldsymbol J')$ vanish on solutions. By
\eqref{eq:s3-current-transformation} and
\eqref{eq:s9-canonical-multiplier}, the tilted components
$\Delta^\mu$, $\mu\in(s,y_0,y_1,y_2)$, satisfy on unrestricted derivative data
\begin{equation}\label{eq:s9-difference-divergence}
 \sum_\mu D_\mu\Delta^{\mu}
 =(Q-Q_{\rm can})^{\Tr}E
 =\left(\cN^{\Tr}Q-q_{\rm can}\right)\cdot\boldsymbol n .
\end{equation}
In the coordinates \eqref{eq:s2-residual-coordinates}, with the
tangential coordinates fixed, each $\Delta^\mu$ vanishes at
$\boldsymbol n_I=0$; its first Taylor differential there, linear in
the residual entries, is denoted $\Delta_{(1)}^\mu$. The terms linear
in the residuals determine the multiplier values on solutions. Keeping
only these terms requires smooth dependence on the derivative data.
Their derivatives are computed using $\mathscr D_\mu$: these
operators differentiate the tangential coefficients on solutions and
increase the residual derivative index in direction $\mu$:
\begin{equation}\label{eq:s9-residual-derivations}
 \mathscr D_s f=\overline D_s f,\qquad
 \mathscr D_{y_j}f=D_{y_j}f,
 \qquad
 \mathscr D_\mu\boldsymbol n_I=\boldsymbol n_{I+e_\mu}.
\end{equation}
Here $f$ depends only on the tangential coordinates and $e_\mu$ is
the unit multi-index in direction $\mu$. By
\eqref{eq:s3-canonical-construction},
$D_sf-\overline D_sf$ is linear in the residuals. Multiplication
of that difference by a residual produces a quadratic term, which
does not contribute to the first Taylor differential. The part of
\eqref{eq:s9-difference-divergence} linear in the residuals is therefore
\begin{equation}\label{eq:s9-linear-residual-identity}
 \left((\cN^{\Tr}Q)|_{\cE_\infty}-q_{\rm can}\right)\cdot\boldsymbol n
       =\sum_\mu\mathscr D_\mu\Delta_{(1)}^\mu .
\end{equation}
In \eqref{eq:s9-linear-residual-identity} the residual
entries are independent variables. To extract their coefficients
while eliminating the divergence on the right, apply the Euler
operator with respect to the residual variables,
\begin{equation}\label{eq:s9-residual-Euler-operator}
 (\EulerOp_{\boldsymbol n} H)_a
       =\sum_I(-\mathscr D)^I\dfrac{\partial H}{\partial n_I^a}.
\end{equation}
The operators $\mathscr D_\mu$ commute and satisfy
$\partial_{n^a_K}\circ\mathscr D_\mu
=\mathscr D_\mu\circ\partial_{n^a_K}+\partial_{n^a_{K-e_\mu}}$,
because $\mathscr D_\mu H$ consists of $\mathscr D_\mu$ applied to the
tangential coefficients of $H$ plus
$\sum_I n^a_{I+e_\mu}\,\partial H/\partial n^a_I$. These two
properties are all that the standard proof that the Euler operator
annihilates total divergences uses; hence
$\EulerOp_{\boldsymbol n}$ annihilates every sum
$\sum_\mu\mathscr D_\mu H^\mu$ with $H^\mu$ polynomial in the
residual entries. Applied to the left side of
\eqref{eq:s9-linear-residual-identity}, it returns the coefficient of
$\boldsymbol n$, because only the term with $I=0$ contributes. The
resulting equality is
\begin{equation}\label{eq:s9-multiplier-values}
 (\cN^{\Tr}Q)|_{\cE_\infty}=q_{\rm can},
 \qquad\text{that is,}\qquad
 Q|_{\cE_\infty}=Q_{\rm can}|_{\cE_\infty}.
\end{equation}

\emph{Step 3: recovery of the current class.}
If a current of the class of $(\rho,\boldsymbol J)$ has an exact
multiplier $Q$ vanishing on solutions, then by
\eqref{eq:s9-multiplier-values} $q_{\rm can}$ vanishes at every
set of compatible derivative data, hence identically, since its arguments are the free
tangential coordinates. By \eqref{eq:s9-canonical-multiplier} the
corrected current then has identically zero divergence, so the class
is trivial. Applying the same argument to the difference of two
currents whose exact multipliers agree on solutions proves their equivalence.
\end{proof}

The adjoint statement of Theorem~\ref{thm:main} was proved in
Proposition~\ref{prop:adjoint-classification}. To complete the statement
for currents, let $(\rho,\boldsymbol J)$ be any local conservation
current. By Proposition~\ref{prop:current-correspondence}, its class
contains a current with an exact multiplier $Q$. The restriction of
$Q$ to solutions is an adjoint symmetry by
\eqref{eq:s9-multiplier-Euler-identity}, and therefore agrees with a
pair \eqref{eq:s1-5}. The explicit current
\eqref{eq:s9-explicit-currents} with the same constants and functions
has that pair as its unrestricted multiplier, by
\eqref{eq:s9-current-identities}. The two currents are equivalent by
Proposition~\ref{prop:current-correspondence}, completing the proof.

In \eqref{eq:s9-explicit-currents}, the term with $c$ is
equivalent to $c/2$ times the helicity current
\eqref{eq:intro-helicity} by \eqref{eq:s9-helicity-flux-equivalence},
and the remaining terms are the currents \eqref{eq:intro-energy},
\eqref{eq:intro-angular}, \eqref{eq:intro-general-momentum} and
\eqref{eq:intro-mass}.

\section{Conclusions}\label{sec:conclusions}

The local conservation currents \eqref{eq:s9-explicit-currents} and
their multiplier representatives \eqref{eq:s1-5} give the complete
classification for the Euler equations \eqref{eq:s1-1} in the
primitive variables $\uu,p$ (Theorem~\ref{thm:main}). The classical
incompressibility, momentum, angular momentum, energy, helicity and
generalised momentum currents exhaust the local classes, and every
local adjoint symmetry has a multiplier representative in this
family. Explicit dependence on time and position and smooth
dependence on derivatives of arbitrary finite order, including mixed
time and spatial derivatives, are allowed.

Every local symmetry is equivalent to a characteristic of the point
family \eqref{eq:sym-family} (Theorem~\ref{thm:symmetries}), proving
the conjecture of \citet{Olver1982}. The decisive order-reduction
step compares the harmonic and transverse responses on their common
complex covectors. The harmonic response has field-independent
coefficients, whereas Proposition~\ref{thm:dependence} restricts
the transverse response to a vorticity-dependent form. Their equality
in \eqref{eq:sym-intersection} excludes every order above one.

The local Hamiltonian structure \eqref{eq:noether-map} converts this
symmetry classification into the adjoint classification. Energy
generates time translation through \eqref{eq:hamiltonian-energy};
the helicity and incompressibility multipliers
\eqref{eq:hamiltonian-casimirs} span its kernel on local adjoint symmetries
(Theorem~\ref{thm:hamiltonian}). The kernel calculation uses
Corollary~\ref{cor:material-scalars}: a finite-order differential
function whose material derivative vanishes identically on compatible
Euler data is locally constant there. Preservation of the two
multipliers \eqref{eq:hamiltonian-casimirs} forces the dilation
coefficients $\varsigma,\psi_1$ in \eqref{eq:sym-family} to vanish
for every image symmetry. Subtracting an explicit classical preimage then
leaves a kernel element, completing the adjoint classification at
arbitrary order. The induced bracket on conservation laws is given
in Appendix~\ref{app:poisson}; it includes the incompressibility term
\eqref{eq:poisson-central-example} in the bracket of the translation
and Galilean-boost currents.

The classification concerns volume currents expressed through finite
derivatives of $\uu,p$. Kelvin circulation and the Cauchy invariants
involve material loops or the flow map, and Ertel potential vorticity
\eqref{eq:intro-ertel} involves an additional transported field.
Those settings permit invariants outside the class considered here.
Symmetry-reduced equations may also admit additional local
conservation laws, including the generalised enstrophies
\eqref{eq:intro-enstrophy} of two-dimensional flow. For a chosen
reduction, inheritance of the Hamiltonian structure
\eqref{eq:noether-map} requires the compatibility conditions of the
invariant-reduction framework \citep{DruzhkovCheviakov2026} to be
checked. Extensions of the classification to augmented or
compressible systems require the compatible derivative data and
transported responses of those systems to be determined.

\begin{bmhead}{Acknowledgements}
The author thanks K.~Druzhkov for helpful discussions 
and suggestions that clarified certain arguments.
OpenAI Codex and Anthropic Claude were used in
manuscript preparation, proof analysis, and optimisation. Selected algebraic identities
were independently checked by symbolic computation.
\end{bmhead}

\section*{Statements and Declarations}

\begin{bmhead}{Funding}
This work was supported by the Natural Sciences and Engineering Research
Council of Canada through Discovery Grant RGPIN-2024-04308.
\end{bmhead}

\begin{bmhead}{Competing interests}
The author reports no conflict of interest.
\end{bmhead}

\begin{bmhead}{Data availability}
No experimental or numerical simulation data were generated in this
theoretical study. Symbolic verification scripts accompany the manuscript.
\end{bmhead}

\appendix
\section{Poisson bracket of local conservation laws}\label{app:poisson}

The Hamiltonian structure \eqref{eq:noether-map} defines a Poisson
bracket on conservation-law classes
\citep{KerstenKrasilshchikVerbovetskyVitolo2009}. In terms of their
multiplier representatives, the convention used here is
\begin{equation}\label{eq:poisson-definition}
 \{Q_1,Q_2\}_{\mathscr N}
   =\ell_{Q_2}(\phi_1)+\Delta_1^*Q_2,
 \qquad \phi_1=\mathscr NQ_1,
 \qquad L\phi_1=\Delta_1(E).
\end{equation}
Thus \eqref{eq:poisson-definition} is the action
\eqref{eq:hamiltonian-multiplier-action} of the symmetry generated by
$Q_1$ on the conservation law with multiplier $Q_2$. Equalities
between the resulting multipliers are understood on Euler solutions.
The representative independence established for
\eqref{eq:hamiltonian-multiplier-action} permits the use of the point
representative of $\phi_1$. Replacing $Q_1$ or $Q_2$ by a multiplier
agreeing with it on solutions changes the right side of
\eqref{eq:poisson-definition} by terms containing $E$ or its
derivatives, so the bracket is defined on multiplier classes.

Denote the multiplier family \eqref{eq:s1-5} by
$Q(c,\nu,\boldsymbol\Omega,\boldsymbol b,c_0)$, and let $Q_i$ have
parameters with index $i$. Put
\[
 \boldsymbol\beta_i=\boldsymbol\Omega_i\times\xx+\boldsymbol b_i(t).
\]
For the point representative of $\mathscr NQ_1$, the identity
\eqref{eq:sym-sufficiency} gives
\begin{equation}\label{eq:poisson-point-residual}
 \Delta_1=
 \begin{pmatrix}\mathsfbi C_{\boldsymbol\Omega_1}&0\\0&0\end{pmatrix}
 +(\nu_1D_t-\boldsymbol\beta_1\cdot D_{\xx})\mathsfbi I_4,
\end{equation}
where $\mathsfbi I_4$ is the four-dimensional identity matrix.
Substituting \eqref{eq:s1-5} and
\eqref{eq:poisson-point-residual} into \eqref{eq:poisson-definition}
gives another member of the same family,
\begin{equation}\label{eq:poisson-family}
 \begin{aligned}
 \{Q_1,Q_2\}_{\mathscr N}
   &=Q(0,0,\boldsymbol\Omega_{12},\boldsymbol b_{12},c_{0,12}),\\
 \boldsymbol\Omega_{12}
   &=\boldsymbol\Omega_2\times\boldsymbol\Omega_1,\\
 \boldsymbol b_{12}
   &=\boldsymbol\Omega_2\times\boldsymbol b_1
      -\boldsymbol\Omega_1\times\boldsymbol b_2
      +\nu_2\boldsymbol b_1'-\nu_1\boldsymbol b_2',\\
 c_{0,12}
   &=\nu_2c_{0,1}'-\nu_1c_{0,2}'
      +\boldsymbol b_2\cdot\boldsymbol b_1'
      -\boldsymbol b_1\cdot\boldsymbol b_2'.
 \end{aligned}
\end{equation}
Primes denote ordinary time derivatives; \eqref{eq:poisson-family}
holds for unrestricted fields. The helicity coefficient
does not enter \eqref{eq:poisson-family}, and a constant $c_0$ also
has zero bracket with every member. These are the two multipliers
\eqref{eq:hamiltonian-casimirs}, and by \eqref{eq:poisson-family} they
are the only members with zero bracket with every member of the family.
The bracket \eqref{eq:poisson-family} is skew-symmetric and satisfies
the Jacobi identity, and $Q\mapsto\mathscr NQ$ is a Lie algebra
homomorphism onto the point symmetries \eqref{eq:sym-family} with
$\varsigma=\psi_1=0$, whose kernel is spanned by
\eqref{eq:hamiltonian-casimirs}. The classified conservation laws
therefore form a central extension of that symmetry algebra.

The last two terms of $c_{0,12}$ can give a nonzero incompressibility
multiplier even when the corresponding point symmetries commute.
For a constant vector $\boldsymbol a$, the translation and Galilean
boost multipliers are
\[
 Q_{\boldsymbol a}
   =(\boldsymbol a,\boldsymbol a\cdot\uu),\qquad
 Q_{t\boldsymbol a}
   =(t\boldsymbol a,t\boldsymbol a\cdot\uu-\boldsymbol a\cdot\xx).
\]
Their bracket is
\begin{equation}\label{eq:poisson-central-example}
 \{Q_{\boldsymbol a},Q_{t\boldsymbol a}\}_{\mathscr N}
   =-|\boldsymbol a|^2(\boldsymbol0,1).
\end{equation}
The pair on the right has zero image under $\mathscr N$ and is
represented by the incompressibility current \eqref{eq:intro-mass},
whose density is zero. Equation~\eqref{eq:poisson-central-example}
is therefore a property of the local currents: the translation and
boost classes have a nonzero bracket although the corresponding point
symmetries commute, and the bracket is the mass class.

\begingroup
\interlinepenalty=10000
\bibliography{references}

@article{CraikCriminale1986,
  author  = {Craik, A. D. D. and Criminale, W. O.},
  title   = {Evolution of wavelike disturbances in shear flows: a class of exact solutions of the {Navier--Stokes} equations},
  journal = {Proc. R. Soc. Lond. A},
  year    = {1986},
  volume  = {406},
  number  = {1830},
  pages   = {13--26},
  doi     = {10.1098/rspa.1986.0061}
}

@article{LifschitzHameiri1991,
  author  = {Lifschitz, A. and Hameiri, E.},
  title   = {Local stability conditions in fluid dynamics},
  journal = {Phys. Fluids A},
  year    = {1991},
  volume  = {3},
  number  = {11},
  pages   = {2644--2651},
  doi     = {10.1063/1.858153}
}

@article{Serre1984,
  author  = {Serre, D.},
  title   = {Les invariants du premier ordre de l'\'{e}quation d'{Euler} en dimension trois},
  journal = {Physica D},
  volume  = {13},
  number  = {1--2},
  pages   = {105--136},
  year    = {1984},
  doi     = {10.1016/0167-2789(84)90273-2}
}

@article{CheviakovOberlack2014,
  author = {Cheviakov, A. F. and Oberlack, M.},
  title = {Generalized {Ertel}'s theorem and infinite hierarchies of conserved quantities for three-dimensional time-dependent {Euler} and {Navier--Stokes} equations},
  journal = {J. Fluid Mech.},
  volume = {760},
  pages = {368--386},
  year = {2014},
  doi = {10.1017/jfm.2014.611}
}

@article{AncoBluman2002b,
  author  = {Anco, S. C. and Bluman, G.},
  title   = {Direct construction method for conservation laws of partial differential equations. {Part II}: General treatment},
  journal = {Eur. J. Appl. Math.},
  volume  = {13},
  pages   = {567--585},
  year    = {2002},
  doi     = {10.1017/S0956792501004661}
}

@article{Moffatt1969,
  author = {Moffatt, H. K.},
  title = {The degree of knottedness of tangled vortex lines},
  journal = {J. Fluid Mech.},
  volume = {35},
  number = {1},
  pages = {117--129},
  year = {1969},
  doi = {10.1017/S0022112069000991}
}

@article{EncisoPeraltaSalasTorres2016,
  author = {Enciso, A. and Peralta-Salas, D. and Torres de Lizaur, F.},
  title = {Helicity is the only integral invariant of volume-preserving transformations},
  journal = {Proc. Natl Acad. Sci. USA},
  volume = {113},
  number = {8},
  pages = {2035--2040},
  year = {2016},
  doi = {10.1073/pnas.1516213113}
}

@article{IzosimovKhesin2017,
  author = {Izosimov, A. and Khesin, B.},
  title = {Classification of {Casimirs} in {2D} hydrodynamics},
  journal = {Mosc. Math. J.},
  volume = {17},
  number = {4},
  pages = {699--716},
  year = {2017},
  doi = {10.17323/1609-4514-2017-17-4-699-716}
}

@article{DeRosa2020,
  author = {De Rosa, L.},
  title = {On the helicity conservation for the incompressible {Euler} equations},
  journal = {Proc. Amer. Math. Soc.},
  volume = {148},
  number = {7},
  pages = {2969--2979},
  year = {2020},
  doi = {10.1090/proc/14952}
}

@article{WangWeiWuYe2024,
  author = {Wang, Y. and Wei, W. and Wu, G. and Ye, Y.},
  title = {On the energy and helicity conservation of the incompressible {Euler} equations},
  journal = {J. Nonlinear Sci.},
  volume = {34},
  pages = {art.~63},
  year = {2024},
  doi = {10.1007/s00332-024-10040-8}
}

@article{GusyatnikovaYumaguzhin1989,
  author  = {Gusyatnikova, V. N. and Yumaguzhin, V. A.},
  title   = {Symmetries and conservation laws of {Navier--Stokes} equations},
  journal = {Acta Appl. Math.},
  volume  = {15},
  number  = {1--2},
  pages   = {65--81},
  year    = {1989},
  doi     = {10.1007/BF00131930}
}

@incollection{KhorkovaVerbovetsky1995,
  author    = {Khor'kova, N. G. and Verbovetsky, A. M.},
  title     = {On symmetry subalgebras and conservation laws for the
               {$k$--$\varepsilon$} turbulence model and the
               {Navier--Stokes} equations},
  editor    = {Lychagin, V. V.},
  booktitle = {The Interplay between Differential Geometry and
               Differential Equations},
  series    = {American Mathematical Society Translations, Series 2},
  volume    = {167},
  pages     = {61--90},
  publisher = {American Mathematical Society},
  address   = {Providence, RI},
  year      = {1995},
  doi       = {10.1090/trans2/167/04}
}

@book{BlumanCheviakovAnco2010,
  author    = {Bluman, G. W. and Cheviakov, A. F. and Anco, S. C.},
  title     = {Applications of Symmetry Methods to Partial Differential Equations},
  series    = {Applied Mathematical Sciences},
  volume    = {168},
  publisher = {Springer},
  address   = {New York},
  year      = {2010},
  doi       = {10.1007/978-0-387-68028-6}
}

@book{CheviakovZhao2024,
  author    = {Cheviakov, A. and Zhao, P.},
  title     = {Analytical Properties of Nonlinear Partial Differential
               Equations: with Applications to Shallow Water Models},
  series    = {CMS/CAIMS Books in Mathematics},
  volume    = {10},
  publisher = {Springer},
  address   = {Cham},
  year      = {2024},
  doi       = {10.1007/978-3-031-53074-6}
}

@book{Batchelor2000,
  author = {Batchelor, G. K.},
  title = {An Introduction to Fluid Dynamics},
  publisher = {Cambridge University Press},
  address = {Cambridge},
  year = {2000},
  doi = {10.1017/CBO9780511800955}
}

@article{Kraichnan1973,
  author = {Kraichnan, R. H.},
  title = {Helical turbulence and absolute equilibrium},
  journal = {J. Fluid Mech.},
  volume = {59},
  number = {4},
  pages = {745--752},
  year = {1973},
  doi = {10.1017/S0022112073001837}
}

@article{AlexakisBiferale2018,
  author = {Alexakis, A. and Biferale, L.},
  title = {Cascades and transitions in turbulent flows},
  journal = {Phys. Rep.},
  volume = {767--769},
  pages = {1--101},
  year = {2018},
  doi = {10.1016/j.physrep.2018.08.001}
}

@article{BesseFrisch2017,
  author = {Besse, N. and Frisch, U.},
  title = {Geometric formulation of the {Cauchy} invariants for incompressible
           {Euler} flow in flat and curved spaces},
  journal = {J. Fluid Mech.},
  volume = {825},
  pages = {412--478},
  year = {2017},
  doi = {10.1017/jfm.2017.402}
}

@article{Salmon1988,
  author = {Salmon, R.},
  title = {Hamiltonian fluid mechanics},
  journal = {Annu. Rev. Fluid Mech.},
  volume = {20},
  pages = {225--256},
  year = {1988},
  doi = {10.1146/annurev.fl.20.010188.001301}
}

@article{AncoWebb2020,
  author = {Anco, S. C. and Webb, G. M.},
  title = {Hierarchies of new invariants and conserved integrals in inviscid fluid flow},
  journal = {Phys. Fluids},
  volume = {32},
  number = {8},
  pages = {086104},
  year = {2020},
  doi = {10.1063/5.0011649}
}

@book{Olver1993,
  author = {Olver, P. J.},
  title = {Applications of {Lie} Groups to Differential Equations},
  edition = {2},
  series = {Graduate Texts in Mathematics},
  volume = {107},
  publisher = {Springer},
  address = {New York},
  year = {1993},
  doi = {10.1007/978-1-4612-4350-2}
}

@article{KelbinCheviakovOberlack2013,
  author = {Kelbin, O. and Cheviakov, A. F. and Oberlack, M.},
  title = {New conservation laws of helically symmetric, plane and
           rotationally symmetric viscous and inviscid flows},
  journal = {J. Fluid Mech.},
  volume = {721},
  pages = {340--366},
  year = {2013},
  doi = {10.1017/jfm.2013.72}
}

@article{Moreau1961,
  author = {Moreau, J.-J.},
  title = {Constantes d'un {\^{i}}lot tourbillonnaire en fluide parfait barotrope},
  journal = {C. R. Acad. Sci. Paris},
  volume = {252},
  pages = {2810--2812},
  year = {1961}
}

@article{Ertel1942,
  author = {Ertel, H.},
  title = {Ein neuer hydrodynamischer {Wirbelsatz}},
  journal = {Meteorol. Z.},
  volume = {59},
  number = {9},
  pages = {277--281},
  year = {1942}
}

@article{ConstantinETiti1994,
  author = {Constantin, P. and E, W. and Titi, E. S.},
  title = {{Onsager}'s conjecture on the energy conservation for
           solutions of {Euler}'s equation},
  journal = {Commun. Math. Phys.},
  volume = {165},
  number = {1},
  pages = {207--209},
  year = {1994},
  doi = {10.1007/BF02099744}
}

@article{Isett2018,
  author = {Isett, P.},
  title = {A proof of {Onsager}'s conjecture},
  journal = {Ann. Math.},
  volume = {188},
  number = {3},
  pages = {871--963},
  year = {2018},
  doi = {10.4007/annals.2018.188.3.4}
}

@article{Buchnev1971,
  author  = {Buchnev, A. A.},
  title   = {Lie group admitted by the equations of motion of an ideal
             incompressible fluid},
  journal = {Dinamika Sploshnoi Sredy (Continuum Dynamics)},
  volume  = {7},
  pages   = {212--214},
  year    = {1971},
  note    = {Institute of Hydrodynamics, Novosibirsk; in Russian}
}

@article{Olver1982,
  author  = {Olver, P. J.},
  title   = {A nonlinear {Hamiltonian} structure for the {Euler} equations},
  journal = {J. Math. Anal. Appl.},
  volume  = {89},
  number  = {1},
  pages   = {233--250},
  year    = {1982},
  doi     = {10.1016/0022-247X(82)90100-7}
}

@article{BihloPopovych2020,
  author  = {Bihlo, A. and Popovych, R. O.},
  title   = {Zeroth-order conservation laws of two-dimensional
             shallow water equations with variable bottom topography},
  journal = {Stud. Appl. Math.},
  volume  = {145},
  number  = {2},
  pages   = {291--321},
  year    = {2020},
  doi     = {10.1111/sapm.12320}
}

@book{FultonHarris1991,
  author    = {Fulton, W. and Harris, J.},
  title     = {Representation Theory: A First Course},
  series    = {Graduate Texts in Mathematics},
  volume    = {129},
  publisher = {Springer},
  address   = {New York},
  year      = {1991},
  doi       = {10.1007/978-1-4612-0979-9}
}

@article{MartinezAlonso1979,
  author  = {Mart{\'i}nez Alonso, L.},
  title   = {On the {Noether} map},
  journal = {Letters in Mathematical Physics},
  volume  = {3},
  number  = {5},
  pages   = {419--424},
  year    = {1979},
  doi     = {10.1007/BF00397216}
}

@article{Kraichnan1967,
  author  = {Kraichnan, Robert H.},
  title   = {Inertial ranges in two-dimensional turbulence},
  journal = {Physics of Fluids},
  year    = {1967},
  volume  = {10},
  number  = {7},
  pages   = {1417--1423},
  doi     = {10.1063/1.1762301}
}

@article{DruzhkovCheviakov2026,
  author  = {Druzhkov, Kostya and Shevyakov, Alexey},
  title   = {Invariant reduction for partial differential equations.
             {IV}: Symmetries that rescale geometric structures},
  journal = {Communications in Nonlinear Science and Numerical Simulation},
  year    = {2026},
  volume  = {163},
  pages   = {110510},
  doi     = {10.1016/j.cnsns.2026.110510}
}

@incollection{KerstenKrasilshchikVerbovetskyVitolo2009,
  author    = {Kersten, P. and Krasil'shchik, I. S. and Verbovetsky, A. M. and Vitolo, R.},
  title     = {{Hamiltonian} structures for general {PDEs}},
  booktitle = {Differential Equations: Geometry, Symmetries and Integrability: The {Abel} Symposium 2008},
  editor    = {Kruglikov, B. and Lychagin, V. and Straume, E.},
  series    = {Abel Symposia},
  volume    = {5},
  pages     = {187--198},
  publisher = {Springer},
  address   = {Berlin, Heidelberg},
  year      = {2009},
  doi       = {10.1007/978-3-642-00873-3_9}
}

@article{KrasilshchikMorozov2025,
  author  = {Krasil'shchik, I. S. and Morozov, O. I.},
  title   = {Integrability structures of the $(2+1)$-dimensional {Euler} equation},
  journal = {Journal of Geometry and Physics},
  volume  = {215},
  pages   = {105543},
  year    = {2025},
  doi     = {10.1016/j.geomphys.2025.105543}
}

@article{Arnold1966,
  author  = {Arnold, V. I.},
  title   = {Sur la g\'eom\'etrie diff\'erentielle des groupes de {Lie} de dimension infinie et ses applications \`a l'hydrodynamique des fluides parfaits},
  journal = {Annales de l'Institut Fourier},
  volume  = {16},
  number  = {1},
  pages   = {319--361},
  year    = {1966}
}

@article{AndersonTorre1996,
  author  = {Anderson, Ian M. and Torre, Charles G.},
  title   = {Classification of local generalized symmetries for the vacuum {Einstein} equations},
  journal = {Commun. Math. Phys.},
  volume  = {176},
  number  = {3},
  pages   = {479--539},
  year    = {1996},
  doi     = {10.1007/BF02099248}
}

@misc{Morozov2026EulerHelmholtz,
  author       = {Morozov, Oleg I.},
  title        = {A {Lax} representation, symmetries, and conservation laws for the three-dimensional {Euler--Helmholtz} equations},
  year         = {2026},
  howpublished = {arXiv:2409.05752v3},
  note         = {Version 3, 9 September 2026},
  doi          = {10.48550/arXiv.2409.05752},
  url          = {https://arxiv.org/abs/2409.05752v3}
}
\endgroup
\end{document}